\documentclass[reqno]{amsart}

\usepackage{amsmath,amssymb,amsbsy,amsfonts,bbm,float,graphics}

\usepackage{mathdots}
\usepackage{etex, xy}
\xyoption{all}

\let\set\mathbbm

\newcommand{\RPTB}{{\it RPTB}}
\newcommand{\DR}{{\it DR}}

\newtheorem{theorem}{Theorem}
\newtheorem{definition}[theorem]{Definition}
\newtheorem{example}[theorem]{Example}
\newtheorem{lemma}[theorem]{Lemma}
\newtheorem{corollary}[theorem]{Corollary}
\newtheorem{proposition}[theorem]{Proposition}
\newtheorem{remark}[theorem]{Remark}
\newtheorem{assum}[theorem]{Assumption}

\newcommand{\sprod}[2]{\text{\small$\prod_{#1}^{#2}$}}

\makeatletter 
\def\newdef#1{\newtheorem{@#1}[Theorem]{{\it #1}}%
  \newenvironment{#1}{\begin{@#1}\em}{\end{@#1}}}
\makeatother

\newcommand{\shiftS}{{S}}

\newcommand{\notion}[1]{{\em #1}}

\newcommand\ToDo[1][x]{\fbox{\textbf{TODO} \ifx#1x\fi}}

\newcommand{\vect}[1]{\boldsymbol{#1}}

\newcommand{\AR}{\set A}
\newcommand{\BB}{\set B}
\newcommand{\QQ}{\set Q}
\newcommand{\ZZ}{\set Z}
\newcommand{\NN}{\set N}
\newcommand{\KK}{\set K}
\newcommand{\HH}{\set H}
\newcommand{\GG}{\set G}
\newcommand{\FF}{\set F}
\newcommand{\SA}{\set S}

\newcommand{\EE}{\set E}

\newcommand{\fct}[3]{#1\colon #2 \to #3}
\newcommand{\dfield}[2]{({#1},{#2})}
\newcommand{\const}[2]{{\rm const}_{#2}{#1}}

\newcommand{\sigmaE}{$\Sigma$}
\newcommand{\piE}{$\Pi$}

\newcommand{\rE}{$R$}

\newcommand{\pisiE}{$\Pi\Sigma$}

\newcommand{\rpiE}{$R\Pi$}

\newcommand{\xdeg}{\overline{\deg}}
\newcommand{\ord}{\text{ord}}
\newcommand{\lcm}{\text{lcm}}

\newcommand{\seqK}{\textbf{S}(\KK)}
\newcommand{\Shift}{S}

\newcommand{\id}{\operatorname{id}}

\newcommand{\ev}{\operatorname{ev}}

\allowdisplaybreaks[3]

\author{Carsten Schneider}
\address{Research Institute for Symbolic Computation\\
J. Kepler University Linz\\
A-4040 Linz, Austria}
\email{Carsten.Schneider@risc.jku.at}
\thanks{This work was supported by the Austrian Science Fund (FWF) grant SFB F50 (F5009-N15).}

\title[Minimal representations and algebraic relations of products]{Minimal representations and algebraic relations for single nested products}

\begin{document}

\begin{abstract}
Recently, it has been shown constructively how a finite set of hypergeometric products, multibasic hypergeometric products or their mixed versions can be modeled properly in the setting of formal difference rings. Here special emphasis is put on robust constructions: whenever further products have to be considered, one can reuse --up to some mild modifications-- the already existing difference ring. In this article we relax this robustness criteria and seek for another form of optimality. We will elaborate a general framework to represent a finite set of products in a formal difference ring where the number of transcendental product generators is minimal. As a bonus we are able to describe explicitly all relations among the given input products.
\end{abstract}

\maketitle

\section{Introduction}\label{Sec:Intro}

An important milestone of symbolic summation has been carried out by S.A. Abramov~\cite{Abramov:71,Abramov:75} to simplify indefinite sums defined over rational functions. In particular, Gosper's algorithm~\cite{Gosper:78} for the simplification of indefinite hypergeometric sums and Zeilberger's extension to definite sums via his creative telescoping paradigm~\cite{Zeilberger:91,PauleSchorn:95,AequalB} made symbolic summation highly popular in many areas of sciences. This successful story has been pushed forward for single nested sums and related problems, see, e.g.,~\cite{Paule:95,PauleRiese:97,CK:12} . Further generalizations opened up substantially the class of applications, like the holonomic approach~\cite{Zeilberger:90a,Chyzak:00,Koutschan:13} dealing with objects that can be desribed by recurrence systems or the multi-summation approach of \hbox{($q$--)}hypergeometric products~\cite{Wilf:92,Wegschaider,AZ:06}.

In this regard, also the difference field/ring approach initiated by M. Karr~\cite{Karr:81,Karr:85} and extended further in~\cite{DR1,DR2,DR3} has been applied to non-trivial problems arising, e.g., in particle physics; for recent calculations see~\cite{CALadder:16,HugeSummation:18}. In the latter case one can represent indefinite nested sums defined over ($q$-)hypergeometric products in the setting of $R\Pi\Sigma$-difference ring extensions. As a side product, one can simplify the sum expressions w.r.t.\ certain optimality criteria, like finding sum representations with optimal nesting depth~\cite{Schneider:08c,Schneider:10b,Schneider:15}, with a minimal number of summation objects in the summands~\cite{Schneider:10a}, or with minimal degrees arising in the numerators and denominators~\cite{Schneider:07d}. In particular, the occurring sums and products in the reduced expression are algebraically independent among each other~\cite{Schneider:10c,Singer:08,DR3}. 

Various improvements have been derived for optimal representations of sums, but much less has been gained for products so far. Concerning the simplification of one product in the setting of difference fields we refer to~\cite{Schneider:05c,Petkov:10,ZimingLi:11}. For the simplification of several products, only few algorithms have been developed. All of them can be related to the following problem. 

\medskip

\noindent \textbf{Problem~\RPTB} (Representation of Products in a Transcendental Basis): \textit{Given}
\begin{equation}\label{Equ:InputProds}
F_1(n)=\prod_{k=l_1}^n f_1(k),\dots,F_r(n)=\prod_{k=l_r}^n f_r(k),
\end{equation}
where the multiplicands $f_i(k)$ are represented in an appropriate difference field $\FF$ and\footnote{For $1\leq i\leq r$ we assume that $f_i(k)\neq0$ for all $k\in\NN=\{0,1,2,\dots\}$ with $k\geq l_i$.} $l_i\in\NN$; \textit{find} an alternative set of such products
\begin{equation}\label{Equ:OutputProds}
\Phi_1(n)=\prod_{k=\ell_1}^n \phi_1(k),\dots,\Phi_s(n)=\prod_{k=\ell_s}^n \phi_s(k)
\end{equation}
and
\begin{equation}\label{Equ:RootProd}
\Phi_0(n)=\rho^n=\prod_{k=1}^n\rho\quad\quad\text{ with a primitive root of unity $\rho$ of order $\lambda$} 
\end{equation}
where $\Phi_0(n)^{\lambda}=1$ such that
\begin{enumerate}
\item[(i)]  any product in~\eqref{Equ:InputProds} can be rewritten in a Laurent polynomial expression in terms of the products given in~\eqref{Equ:OutputProds} and~\eqref{Equ:RootProd};
\item[(ii)] the sequences produced by the products in~\eqref{Equ:OutputProds} are algebraically independent over their ground field of sequences\footnote{We assume that $\FF$ can be embedded into the ring of sequences.} adjoined with the sequence $(\phi_0(n))_{n\geq0}$.
\end{enumerate}
Internally, the available algorithms~\cite{Schneider:05c,OS:18} represent a finite set of such products automatically   
in a difference ring built by $R\Pi$-extensions~\cite{Karr:81,DR1} and exploit results from the Galois theory of difference rings elaborated in~\cite{Schneider:10c,Singer:08,DR3}.
We note that the algorithms presented in~\cite{Schneider:05c,DR2} and implemented in \texttt{Sigma}~\cite{Schneider:07a} can handle only hypergeometric products. In particular, the products must evaluate to elements in a field $\KK$ that can be built by a multivariate rational function field defined over $\QQ$ or $\QQ[\iota]$ where $\iota$ denotes the imaginary unit with $\iota^2=-1$.  
Recently, these ideas have been generalized in~\cite{OS:18} for mixed-multibasic hypergeometric products~\cite{Bauer:99} defined over a more general field $\KK$.

\begin{definition}\label{Def:MMHypergeometric}
Let $\KK=\KK'(q_1,\dots,q_v)$ be a rational function field where $\KK'$ is a field of characteristic $0$. 
Whenever we focus on algorithmic aspects, we restrict $\KK'$ further to a rational function field defined over an algebraic number field.
A product $\prod_{j=l}^kf(j,q_1^j,\dots,q_v^j)$, $l\in\NN$, is called \emph{mixed-multibasic hypergeometric}~\cite{Bauer:99} (in short \emph{mixed-hypergeometric}) in $k$ over $\KK$ if $f(x,y_1,\dots,y_v)$ is an element from the rational function field $\KK(x,y_1,\dots,y_v)$ where the numerator and denominator of
$f(j,q_1^j,\dots,q_v^j)$ are nonzero for all $j\in\ZZ$ with $j\geq l$. 
Such a product is called \emph{multibasic hypergeometric} if $f$ is free of $x$ and \emph{$q$-hypergeometric} if $f$ is free of $x$, $v=1$ and $q_1=q$. It is called \emph{hypergeometric} if $v=0$, i.e., $f\in\KK(x)$ with $\KK=\KK'$.
\end{definition}

\begin{example}\label{Exp:MainExampleSum}
Consider the hypergeometric products
\begin{equation}\label{Equ:DefineFi}
\begin{aligned}
F_1(n)&=\prod_{k=1}^n\underbrace{\frac{-13122 k (1+k)}{(3+k)^3}}_{=f_1(k)},&F_2(n)&=\prod_{k=1}^n\underbrace{\frac{26244 k^2 (2+k)^2}{(3+k)^2}}_{=f_2(k)},\\ 
F_3(n)&=\prod_{k=1}^n\underbrace{\frac{\iota k (2+k)^3}{729 (5+k)}}_{=f_3(k)},&F_4(n)&=\prod_{k=1}^n\underbrace{-\frac{162 k (2+k)}{5+k}}_{=f_4(k)}.
\end{aligned}
\end{equation}
Then the algorithms from~\cite{Schneider:05c,DR2,OS:18} find $\Phi_1(n)=n!$, $\Phi_2(n)=2^n$, $\Phi_3(n)=3^n$ (whose sequences are algebraically independent among each other) and the algebraic product $\Phi_0(n)=\iota^n$ satisfying the relation $\phi_0(n)^4=1$ with the following property: the input products can be rephrased in terms of the output products with   
\begin{equation}\label{Equ:OSRepresentation}
\begin{aligned}
F_1(n)&=\frac{216 \left(\iota^n\right)^2 2^n
   \left(3^n\right)^8}{(n+1)^2 (n+2)^3 (n+3)^3
   n!}, & F_2(n)&=\frac{9 \left(2^n\right)^2
   \left(3^n\right)^8 (n!)^2}{(n+3)^2},\\ 
F_3(n)&=\frac{15
   (n+1)^2 (n+2)^2 \iota^n (n!)^3}{(n+3) (n+4) (n+5)
   \left(3^n\right)^6},& F_4(n)&=\frac{60
   \left(\iota^n\right)^2 2^n \left(3^n\right)^4
   n!}{(n+3) (n+4) (n+5)}.
\end{aligned}
\end{equation}
\end{example}
In~\cite{Schneider:05c,DR2} and more generally in~\cite{OS:18} the algorithms are designed to treat products with highest possible flexibility. They split the input products as much as possible into irreducible elements. As a consequence, when further products arise in a later construction phase, 
the already obtained products can --up to some mild modifications-- be reused. Note that such robust constructions are crucial for large-scale calculations that arise, e.g., in particle physics~\cite{CALadder:16,HugeSummation:18}.
The algorithms for single-nested products described in~\cite{OS:18}, and even more general algorithms for multiple nested products~\cite{Ocansey:19}, are available in Ocansey's package \texttt{NestedProducts}.

In this article we will supplement this more practical oriented toolbox with theoretical aspects.
We will provide a general framework that solves Problem~\RPTB\ not only for mixed-hypergeometric products, but for general difference rings that satisfy certain (algorithmic) properties; for further details we refer to Subsection~\ref{Subsec:ProblemDescription} and Section~\ref{Sec:BasicProperties}. 
In particular, the output of Problem~\RPTB\ is in the following sense optimal: among all possible products in~\eqref{Equ:OutputProds} and~\eqref{Equ:RootProd} that provide a solution of Problem~\RPTB,
the number $s\geq0$ of products and the order $\lambda\geq1$ of the root of unity $\rho$ in~\eqref{Equ:RootProd} are minimal (if $\lambda=1$, the product in~\eqref{Equ:RootProd} simplifies to $\gamma^n=1^n=1$.)

\begin{example}\label{Exp:MinimalProductRep}
With our new algorithmic framework we will calculate the minimal number of products
\begin{align*}
\Phi_1(n)&=\sprod{k=1}n\tfrac{-162 k (2+k)}{5+k},& \Phi_2(n)&=\sprod{k=1}n\tfrac{-\iota (3+k)^6}{9 k (1+k)^2 (2+k) (5+k)}
\end{align*} (again the produced sequences are algebraically independent) and the alternating sign $\Phi_0(n)=(-1)^n$  such that the input products can be rephrased in the form
\begin{equation}\label{Equ:FiEvalNew}
\begin{aligned}
F_1(n)&=\tfrac{5 (1+n)^2 (2+n)^5 (3+n)^8}{52488 (4+n) (5+n)} (-1)^n
        \Big(\!\sprod{k=1}n \tfrac{-162 k (2+k)}{5+k}\!\Big)\Big(\!\sprod{k=1}n \tfrac{-\iota (3+k)^6}{9 k (1+k)^2 (2+k) (5+k)}\!\Big)^{-2},\\
F_2(n)&=\tfrac{(4+n)^2 (5+n)^2 }{400} \Big(
        \sprod{k=1}n \tfrac{-162 k (2+k)}{5+k}\Big)^2,\\
F_3(n)&=\tfrac{2754990144 (4+n)^2 (5+n)^2}{25 (1+n)^4 (2+n)^{10} (3+n)^{16}} \Big(
        \sprod{k=1}n \tfrac{-\iota(3+k)^6}{9 k (1+k)^2 (2+k) (5+k)}\Big)^3,\\
F_4(n)&=\sprod{k=1}n \tfrac{-162 k (2+k)}{5+k}.
\end{aligned}
\end{equation}
In particular, we can conclude that Problem~\RPTB\ can be solved only by using an algebraic product of the form~\eqref{Equ:RootProd} where the order $\lambda=2$ is minimal. 
\end{example}

In contrast to~\cite{Schneider:05c,DR2,OS:18,Ocansey:19} this representation with a minimal number of products has one essential disadvantage: When a new product has to be treated in addition, a complete redesign of the already produced products might be necessary. However, our new approach will provide further insight: given the special representation proposed in this article, one can read off straightforwardly a finite set of generators that describe all relations of the input products.

\begin{example}\label{Exp:RelationsInProdRep}
Given the hypergeometric products $F_1(n),F_2(n),F_3(n),F_4(n)$ from Example~\ref{Exp:MainExampleSum}, we can 
compute all algebraic relations among them. More precisely, take the ring $\EE=\QQ(\iota)(x)[y_,y_1^{-1}][y_2,y_2^{-1}][y_3,y_3^{-1}][y_4,y^{-1}_4]$ of Laurent polynomials in the variables $y_1,y_2,y_3,y_4$ with coefficients from the rational function field $\QQ(\iota)(x)$ and consider the ideal\footnote{In this example the evaluation of an element from $\QQ(\iota)(x)$ is carried out by replacing $x$ with concrete values $n\in\NN$. Later we will generalize this simplest case to formal difference rings equipped with an evaluation function acting on the ring elements.}   
\begin{multline*}
Z=\{p(x,y_1,y_2,y_3,y_4)\in\EE\mid\exists \delta\in\NN\,\forall n\in\NN\\
n\geq\delta \Rightarrow p(n,F_1(n),F_2(n),F_3(n),F_4(n))=0\}
\end{multline*}
in $\EE$ that encodes all algebraic relations among the products $F_i(n)$ with $i=1,2,3,4$. Then with our new algorithms we can compute the two generators 
\begin{equation}\label{Equ:riDef}
\begin{aligned}
e_1=&\frac{y_2}{y_4^2}-\frac1{400} (4+x)^2 (5+x)^2,\\
e_2=&\frac{y_2^2 y_4^2}{y_1^6 y_3^4}
-\tfrac1{4199040^2}(1+x)^4 (2+x)^{10} (3+x)^{16} (4+x)^2 (5+x)^2
\end{aligned}
\end{equation}
that span the ideal $Z$. This means\footnote{More generally, if $R$ is a commutative ring with $1$, we define the ideal $I$ generated by $a_1,\dots,a_r\in R$ with $I=\langle a_1,\dots,a_r\rangle_R=\{f_1\,a_1+\dots+f_r\,a_r\mid f_1,\dots,f_r\in R\}$.}
\begin{equation}\label{Equ:IdealSetZ}
Z=\langle e_1,e_2\rangle_{\EE}=\{f_1\,e_1+f_2\,e_2\mid f_1,f_2\in\EE\}.
\end{equation}
\end{example}

This result is connected to~\cite{Kauers:08,Singer:16} where all relations of a finite set of sequences can be computed that satisfy homogeneous linear recurrences with constant coefficients. In particular in~\cite{Kauers:08} their algorithm is reduced to find all relations of a finite set of geometric sequences, i.e., sequences produced by the products in~\eqref{Equ:InputProds} with $f_1(k),\dots,f_r(k)\in\KK^*$, which are a subclass of hypergeometric sequences. Further strategies for mixed-multibasic hypergeometric products are also mentioned in~\cite{ZimingLi:11}. 

The outline of the article is as follows. In Section~\ref{Sec:ProbelmSpec} we will formalize the sketched construction from above in details: we will represent the products~\eqref{Equ:InputProds} in a formal difference ring equipped with an evaluation function. In particular, we will rephrase the problem specification \RPTB\ to the problem specification~\DR\ in this formal setting. In Section~\ref{Sec:BasicProperties} we will list the basic properties of our difference ring theory~\cite{DR1,DR3} and will enhance it for the constructions required in this article. In Section~\ref{Sec:SpecialCase} we will restrict to a special case of products from which an optimal product representation
can be read off straightforwardly. Finally, using the Smith normal form of integer matrices we will show in Section~\ref{Sec:GeneralCase} how the general problem can be reduced to the special case treated in Section~\ref{Sec:SpecialCase}. A conclusion is given in Section~\ref{Sec:Conclusion}.

\section{The problem description in the setting of difference rings}\label{Sec:ProbelmSpec}

In the following it will be convenient to represent the products under consideration in a formal ring\footnote{Throughout this article, all rings and fields have characteristic $0$ and with $\AR^*$ we denote the group of units. Furthermore, all rings are commutative. The order of $a\in\AR^*$, denoted by $\ord(a)$, is the smallest positive integer $k$ with $a^k=1$. If such a $k$ does not exist, we set $\ord(a)=0$.}, often denoted by $\AR$ or $\EE$. In this regard, it is essential to define in addition 
\begin{itemize}
\item an evaluation function that describes how the elements in the formal ring are interpreted as sequences (see Subsection~\ref{Subsec:EvFu});
\item a ring automorphism that describes how the elements in the formal ring are shifted (see Subsection~\ref{Subsec:automorphism}).
\end{itemize}
In particular, we will take care that the evaluation function and the ring automorphism are compatible: applying the automorphism to an element in the formal ring and evaluating it afterwards at the $n$th sequence entry must equal to the evaluation at the $(n+1)$th sequence entry (see Subsection~\ref{Subsec:compatibility}). In Subsection~\ref{Subsec:ProblemDescription} we will finally rephrase Problem~\RPTB\ and the examples from the introduction in this algebraic setting. 

\subsection{The evaluation function -- sequence domains}\label{Subsec:EvFu}
Inspired by~\cite{PauleNemes:97} we will provide a so-called evaluation function which maps the elements from a formal ring $\AR$ to sequences with entries from a field $\KK\subseteq\AR$.
More precisely, we will take care that the following functions will be available within our constructions~\cite{Schneider:10c,DR3}.

\begin{definition}\label{Def:EvZ}
Let $\AR$ be a ring and let $\KK$ be a subfield of $\AR$.
\begin{enumerate}
\item A function $\fct{\ev}{\AR\times\NN}{\KK}$ is called \notion{evaluation function} for $\AR$ if for all $f,g\in\AR$ and $c\in\KK$ there exists a $\lambda\in\NN$ with the following properties:
\begin{align}
\forall n\geq\lambda:&\,\ev(c,n)=c,\label{Equ:evC}\\
\forall n\geq\lambda:&\,\ev(f+g,n)=\ev(f,n)+\ev(g,n),\label{Equ:evPlus}\\
\forall n\geq\lambda:&\,\ev(f\,g,n)=\ev(f,n)\,\ev(g,n).\label{Equ:evTimes}
\end{align}
\item A function $\fct{L}{\AR}{\NN}$ is called \notion{operation-function} (in short \notion{$o$-function}) for $\AR$ and $\ev$ if for any $f,g\in\AR$ with $\lambda=\max(L(f),L(g))$ the properties~\eqref{Equ:evPlus} and~\eqref{Equ:evTimes} hold. 
If such a function exist, $\ev$ is also called \notion{operation-bounded} (in short \notion{$o$-bounded}).
In particular,
$\ev$ is called \notion{$o$-computable} if the function $\ev$ is computable and there is a computable $o$-function $L$ for $\ev$. 
\item $\fct{z}{\AR}{\NN}$ is called \notion{$z$-function} for $\ev$ if for any $f\in\AR\setminus\{0\}$ and for any integer $n\geq Z(f)$ we have $\ev(f,n)\neq0$. 
\end{enumerate}
\end{definition}

Later we will rely on the following simple observation.

\begin{lemma}\label{Lemma:ZFunctionForField}
Let $\fct{\ev}{\AR\times\NN}{\KK}$ be a an evaluation function for a ring $\AR$. 
\begin{enumerate}
 \item If $a\in\AR^*$, then there is a $\delta\in\NN$ such that $\ev(a,n)\neq0$ for all $n\geq\delta$.
 \item If $\AR$ is a field, there is a $z$-function for $\AR$.
\end{enumerate}
\end{lemma}
\begin{proof}
(1) For $a\in\AR^*$ there is a $\delta\in\NN$ with $\ev(a,n)\,\ev(\frac1a,n)=\ev(a\frac1a,n)=1$ for all $n\geq\delta$. Consequently, $\ev(a,n)\neq0$ for all $n\geq\delta$. (2) follows by (1).
\end{proof}


In short, a ring/field $\AR$ equipped with such an evaluation function $\ev$ will be called a \notion{sequence domain} and will be denoted by $(\AR,\ev)$; see also~\cite{Schneider:10b}. 
In all our examples we will always start with the following ground field.

\begin{example}\label{Exp:RationalDField1}
Take the rational function field $\FF=\KK(x)$ over a field $\KK$ with characteristic $0$ and consider the evaluation function $\fct{\ev}{\FF\times\NN}{\KK}$ defined by
\begin{equation}\label{Equ:EvalRat}
\ev(\tfrac{p}{q},n)=\begin{cases}
0&\text{if }q(n)=0\\
\frac{p(n)}{q(n)}&\text{if }q(n)\neq0
\end{cases}
\end{equation}
where $p,q\in\KK[x]$, $q\neq0$ and $p,q$ are co-prime;
here $p(n),q(n)$ is the usual evaluation of polynomials at $n\in\NN$. We define the $o$-function $L(\tfrac{p}{q})$ by taking the minimal $l\in\NN$ with $q(n+l)\neq0$ for all $n\in\NN$, and we define the $z$-function by $Z(\tfrac{p}{q})=L(p\,q)$.
\end{example}

Our product expressions will be rephrased in terms of Laurent polynomials with coefficients from a ring $\AR$.  More precisely, we will denote by $\AR\langle \hat{x}\rangle$ the ring of Laurent polynomials $\AR[\hat{x},\hat{x}^{-1}]$ and by $\AR\langle\hat{x}_1\rangle\dots\langle\hat{x}_r\rangle$ a tower of such ring extensions. 
In order to represent products of the form~\eqref{Equ:RootProd}, we will work also with ring extensions of the form $\AR[z]$ over a ring $\AR$ subject to a relation $z^{\lambda}=1$. 
Note that such a ring contains zero-divisors that originate from
$$(1-z)(1+z+z^2+\dots+z^{\lambda-1})=1-z^{\lambda}=0.$$
More precisely, 
by~\cite[Lemma~5.4 (parts 1,3)]{DR3} one can straightforwardly construct an evaluation function for single nested products.

\begin{lemma}\label{Lemma:evExt}
Let $\FF$ be a field with subfield $\KK$ and let $\fct{\ev}{\FF\times\NN}{\KK}$ be an evaluation function of $\FF$.
Let $\EE=\FF\langle\hat{x}_1\rangle\dots\langle\hat{x}_r\rangle[z]$ be a ring where the $\hat{x}_i$ are Laurent polynomial variables and $z$ is a ring generator subject to the relation $z^{\lambda}=1$.
Let $a_1,\dots,a_r\in\FF^*$ and let $l_1,\dots,l_r\in\NN$ where for all $1\leq i\leq r$ we have $\ev(a_i,n)\neq0$ with $n\geq l_i$. Furthermore, let $\rho\in\KK^*$ with $\rho^{\lambda}=1$. 
Then $\fct{\hat{\ev}}{\EE\times\NN}{\KK}$ defined by
\begin{equation}\label{Equ:EvForProdExt}
\hat{\ev}(f,n)=\hspace*{-1cm}
\sum_{(m_1,\dots,m_r,\mu)\in\ZZ^r\times\NN}\hspace*{-0.8cm}\ev(f_{(m_1,\dots,m_r,\mu)},n)\Big(\prod_{k=l_1}^n\ev(a_1,k)\Big)^{m_1}\dots\Big(\prod_{k=l_r}^n\ev(a_r,k)\Big)^{m_r}\big(\rho^{n}\big)^{\mu}
\end{equation}
with
\begin{equation}\label{Equ:fForEvDef}
f=\sum_{(m_1,\dots,m_r,\mu)\in\ZZ^r\times\NN}f_{(m_1,\dots,m_l,\mu)}\hat{x}_1^{m_1}\dots \hat{x}_r^{m_r}z^{\mu}\in\EE 
\end{equation}
is an evaluation function. In particular if $\fct{L}{\FF}{\NN}$ is an o-function for $\FF$, then
$\fct{\hat{L}}{\EE}{\NN}$ defined by
\begin{equation}\label{Equ:DefL'}
\hat{L}(f)=\max\begin{array}[t]{l}\{L(f_{(m_1,\dots,m_l,\mu)})\mid(m_1,\dots,m_r,\mu)\in\ZZ^r\times\NN\}\cup\\
\{l_i\mid1\leq i\leq r\text{ where $\hat{x}_i$ occurs in $f$}\}\cup\{1\mid f\text{ depends on }z\}
\end{array}
\end{equation}
for~\eqref{Equ:fForEvDef} 
is an $o$-func\-tion. If $L$ is computable, then $\hat{L}$ is computable.
\end{lemma}

\noindent\textit{Remark.} If there is a computable $z$-function $Z$ for $\ev$ and $\FF$, the lower bounds $l_i$ can be computed with $l_i=Z(a_i)$. 

\begin{example}[Cont. Ex.~\ref{Exp:RationalDField1}]\label{Exp:EvMainExp}
We specialize the ground field $\FF=\KK(x)$ with the evaluation function~\eqref{Equ:EvalRat} and the corresponding $o$-function from Example~\ref{Exp:RationalDField1} by choosing the algebraic number field $\KK=\QQ(\iota)$. Now take the Laurent polynomial ring $\EE=\FF\langle\hat{x}_1\rangle\langle\hat{x}_2\rangle\langle\hat{x}_3\rangle\langle\hat{x}_4\rangle $. Following Lemma~\ref{Lemma:evExt} we can construct the evaluation function $\fct{\hat{\ev}}{\EE\times\NN}{\KK}$ for $\EE$ with $\hat{\ev}|_{\KK(x)\times\NN}=\ev$ and 
$\hat{\ev}(\hat{x}_i,n)=F_i(n)$
for $1\leq i\leq 4$ where the $F_i(n)$ are given in~\eqref{Equ:DefineFi}. In particular, we obtain the $o$-function defined by~\eqref{Equ:DefL'}. In the following examples we will work with the evaluation domain $(\EE,\ev)$ in which the product expressions under consideration can be represented formally.   
\end{example}

Let $\KK$ be a subfield of a ring $\AR$. An evaluation function $\fct{\ev}{\AR\times\NN}{\KK}$ for $\AR$ naturally produces sequences in the commutative ring $\KK^{\NN}$ with the identity element $\vect{1}=(1,1,1,\dots)$ with component-wise addition and multiplication. More precisely, we can define the function
$\fct{\tau}{\AR}{\KK^{\NN}}$ with
\begin{equation}\label{Equ:tauDef}
\tau(f)=(\ev(f,0),\ev(f,1),\ev(f,2),\dots).
\end{equation}
Due to \eqref{Equ:evPlus} and~\eqref{Equ:evTimes} the map
$\tau$ can be turned to a ring homomorphism by defining the equivalence relation 
$(f_n)_{n\geq0}\equiv (g_n)_{n\geq0}$ with $f_j=g_j$ for all $j\geq\lambda$ for some $\lambda\in\NN$; compare~\cite{AequalB}. It is easily seen that the set of equivalence classes $[f]$ with $f\in\KK^{\NN}$ forms 
with $[f]+[g]:=[f+g]$ and $[f][g]:=[f g]$
again a commutative ring with $[1]$ which we will denote by $\seqK$. In the following we will simply write $f$ instead of $[f]$. In this setting, $\fct{\tau}{\AR}{\seqK}$ forms a ring homomorphism. 

\subsection{The ring automorphism -- difference rings}\label{Subsec:automorphism}

So far, we elaborated how sequences can be formulated in a formal ring $\AR$ equipped with an evaluation function. 
Finally, we will introduce in addition a ring automorphism $\fct{\sigma}{\AR}{\AR}$ in order to model the shift operator acting on sequences.
Such a ring $\AR$ equipped with a ring automorphism $\sigma$ is also called a \notion{difference ring} denoted by $\dfield{\AR}{\sigma}$.

In order to construct difference rings iteratively, we will use the notion of difference ring/field extensions. $\dfield{\EE}{\sigma'}$ is called a \notion{difference ring extension} of $\dfield{\AR}{\sigma}$ if $\AR$ is a subring of $\EE$ and $\sigma'|_{\AR}=\sigma$. If $\EE$ and $\AR$ are fields, we call such an extension a \notion{difference field extension}. 
In the following we will need the following type of difference ring extensions $\dfield{\EE}{\sigma}$ of $\dfield{\AR}{\sigma}$; for more details see~\cite{DR1,DR3}.
\begin{itemize}
\item A \notion{$P$-extension (product-extension)} if $\EE=\AR\langle\hat{x}\rangle$ is a ring of Laurent polynomials with $\sigma(\hat{x})=\hat{\alpha}\,\hat{x}
$ for some unit $\hat{\alpha}\in\AR^*$. More precisely, for $f=\sum_{k=l}^rf_k\hat{x}_k$ with $l,r\in\ZZ$ we have $\sigma'(f)=\sum_{k=l}^r\sigma(f_k)\hat{\alpha}^k\hat{x}^k$. $\hat{x}$ is also called a \notion{$P$-monomial}.
\item An \notion{$S$-extension (sum-extension)} if $\EE=\AR[\hat{x}]$ is a polynomial ring with $\sigma(\hat{x})=\hat{x}+\hat{\beta}
$ for some $\hat{\beta}\in\AR$. More precisely, for $f=\sum_{k=0}^rf_k\hat{x}_k$ with $r\in\NN$ we have $\sigma'(f)=\sum_{k=0}^r\sigma(f_k)(\hat{x}+\hat{\beta})^k$. $\hat{x}$ is also called an \notion{$S$-monomial}.
\item An \notion{$A$-extension (algebraic extension)} of order $\lambda>1$ if $\EE=\AR[z]$ is a ring subject to the relation $z^{\lambda}=1$ (i.e., $\ord(z)=\lambda$) with $\sigma(z)=\rho\,z$ where $\rho\in\AR^*$ is a $\lambda$th root of unity (i.e., $\rho^{\lambda}=1$). $z$ is also called an \notion{$A$-monomial}. 
\end{itemize}
Since $\sigma'$ and $\sigma$ agree on $\AR$, we will not distinguish them anymore. In particular, a \notion{$PS$-extension} (resp.\ \notion{$AP$-extension/$APS$-extension}) is a $P$ or $S$-extension (resp.\ an $A$-extension or $P$-extension/an $A$-extension, $P$-extension or $S$-extension).
More generally we call $\dfield{\EE}{\sigma}$ a \notion{(nested) $P$-extension/$S$-extension/$A$-extension/$PS$-extension/$AP$-extension/$APS$-extension)} of $\dfield{\AR}{\sigma}$ if it is built by a tower of such extensions over a difference ring $\dfield{\AR}{\sigma}$.\\
Let $\dfield{\EE}{\sigma}$ be a difference field extension of a difference field $\dfield{\FF}{\sigma}$. It is a \notion{$P$-field extension (resp.\ $S$-field extension)} if $\EE=\FF(\hat{x})$ is a rational function and $\sigma(\hat{x})=\hat{\alpha}\,\hat{x}$ with $\hat{\alpha}\in\FF^*$ (resp.\ $\sigma(\hat{x})=\hat{x}+\hat{\beta}$ with $\hat{\beta}\in\FF$). More generally, a \notion{(nested) $S$-field extension/$P$-field extension/$PS$-field extension} is a tower of such extensions.\\ 
\textit{Remark.} The quotient field of a (nested) $P$-extension/$S$-extension/$PS$-extension of a difference field is a special class of $P$-field/$S$-field/$PS$-field extensions (the multiplicands/summands can be chosen only from a subring of the ground field). 

\begin{example}[Cont. Ex.~\ref{Exp:EvMainExp}]\label{Exp:MainPExt}
Consider the difference field $\dfield{\FF}{\sigma}$ with $\FF=\KK(x)$ where $\KK=\QQ(\iota)$ and with the field automorphism $\fct{\sigma}{\FF}{\FF}$ defined by $\sigma|_{\KK}=\id$ and $\sigma(x)=x+1$. In the following we call this difference field also the \notion{rational difference field}. Note that $\dfield{\FF}{\sigma}$ is an $S$-field extension of $\dfield{\KK}{\sigma}$.
In Example~\ref{Exp:EvMainExp} we have introduced already the Laurent polynomial ring $\EE=\FF\langle \hat{x}_1\rangle\langle\hat{x}_2\rangle\langle\hat{x}_3\rangle\langle\hat{x}_4\rangle$ where the products~\eqref{Equ:DefineFi} 
are represented by $\hat{x_i}$ for $i=1,2,3,4$ and the evaluation function~\eqref{Equ:EvalhatXi}. We can now extend the automorphism $\sigma$ from $\FF$ to $\EE$ by a tower of $P$-extensions with 
$\sigma(\hat{x}_i)=\hat{\alpha}_i\,\hat{x}_i$ for $1\leq i\leq 4$ where
\begin{equation}\label{Equ:MainAlpha}
\begin{aligned}
\hat{\alpha}_1&=\frac{-13122 (x+1)
   (x+2)}{(x+4)^3},&
\hat{\alpha}_2&=\frac{26244 (x+1)^2
   (x+3)^2}{(x+4)^2},\\
\hat{\alpha}_3&=\frac{\iota (x+1) (x+3)^3}{729
   (x+6)},&
\hat{\alpha}_4&=\frac{-162 (x+1) (x+3)}{x+6}.
\end{aligned}
 \end{equation}
Note that $F_i(n+1)=\hat{\alpha}_i(n)\,F_i(n)$ for $1\leq i\leq4$, i.e., $\sigma$ acting on $\hat{x}_i$ models the shift operator applied to $F_i(n)$. 
\end{example}
Similarly, $S$-extensions are used to model indefinite nested sums. Since we focus mainly on products, we skip these aspects and refer the reader to~\cite{Schneider:08c,Schneider:10a,DR1,DR3}. 

In order to solve the Problem~\RPTB\ introduced in Section~\ref{Sec:Intro},
the difference ring/field extensions from above have to be refined. In this regard, we introduce the \notion{set of constants} $$\const{\AR}{\sigma}=\{c\in\AR\mid\sigma(c)=c\}$$
of a difference ring $\dfield{\AR}{\sigma}$. In general, $\KK=\const{\AR}{\sigma}$ is a subring of $\AR$ that contains $\QQ$ as subfield. In particular, if $\FF$ is a field, $\KK$ is automatically a subfield of $\FF$.
In this article we will take care that $\KK$ is always a subfield of $\AR$ which we will also call the \notion{constant field of $\dfield{\AR}{\sigma}$}. 

\begin{definition}
A \notion{(nested) \piE-extension} (resp.\ \notion{\sigmaE-/$R$-/$R\Pi$-/$R\Sigma$-/$\Sigma\Pi$-/$R\Pi\Sigma$- exten\-sion}) $\dfield{\EE}{\sigma}$ of $\dfield{\AR}{\sigma}$ is a $P$-extension (resp.\ $S$-/$A$-/$AP$-/$AS$-/$SP$-/$APS$-exten\-sion) with $\const{\EE}{\sigma}=\const{\AR}{\sigma}$. In this case, an $A$-/$P$-/$S$-monomial is also called an \notion{$R$-/\piE-/\sigmaE-monomial}. Similarly, a \notion{(nested) \piE-field extension} (resp.\  \notion{(nested) \sigmaE-/\pisiE-field extension}) is a $P$-field extension (resp.\ $S$-/$PS$-field extension) where the constants remain unchanged. 
Finally, a \notion{\pisiE-field} $\dfield{\FF}{\sigma}$ over $\KK$ is a (nested) \pisiE-field extension of $\dfield{\KK}{\sigma}$ with $\const{\FF}{\sigma}=\const{\KK}{\sigma}$.   
\end{definition}

\begin{example}\label{Exp:MainPExtSpec}
Consider the rational difference field $\dfield{\FF}{\sigma}$ from Example~\ref{Exp:MainPExt}. It is not difficult to see that $\const{\FF}{\sigma}=\KK$. Consequently, $\dfield{\FF}{\sigma}$ is a \pisiE-field over $\KK$.
\end{example}

\noindent We remark that these extensions are motivated by Karr's work~\cite{Karr:81,Karr:85}.
More precisely, the \piE-field and \sigmaE-field extensions and in particular \pisiE-fields have been introduced in~\cite{Karr:81,Karr:85} and explored further, e.g., in~\cite{Bron:00,Schneider:01,Schneider:08c,Schneider:10a}.

\begin{remark}\label{Remark:Rearrange}
In the following we will restrict to $AP$-extensions $\dfield{\AR}{\sigma}$ of a difference field $\dfield{\FF}{\sigma}$ with $\AR=\FF\langle x_1\rangle\dots\langle x_s\rangle[z_1]\dots[z_l]$ where the $\frac{\sigma(x_i)}{x_i}\in\FF^*$ for $1\leq i\leq s$ are $P$-monomials and the $\frac{\sigma(z_i)}{z_i}\in\const{\FF}{\sigma}^*$ for $1\leq i\leq l$ are $A$-monomials. One can rearrange the generators in any order and obtains again an $AP$-extension. For instance, $\dfield{\AR'}{\sigma}$ is an $AP$-extension of $\dfield{\FF}{\sigma}$ with $\AR'=\FF[z_1]\dots[z_l]\langle x_1\rangle\dots\langle x_s\rangle$. In particular, if $\const{\AR}{\sigma}=\const{\FF}{\sigma}$ holds then also $\const{\AR'}{\sigma}=\const{\FF}{\sigma}$ holds. I.e., if $\dfield{\AR}{\sigma}$ is an $R\Pi$-extension of $\dfield{\FF}{\sigma}$, $\dfield{\AR'}{\sigma}$ is an $R\Pi$-extension of $\dfield{\FF}{\sigma}$. 
\end{remark}

\subsection{The compatibility of $\ev$ and $\sigma$}\label{Subsec:compatibility}

Let $\dfield{\AR}{\sigma}$ be a difference ring with constant field $\KK$. $\fct{\ev}{\AR\times\NN}{\KK}$ is called an \notion{evaluation function for $\dfield{\AR}{\sigma}$} if $\ev$ is an evaluation function for the ring $\AR$ and $\ev$ and $\sigma$ satisfy the following compatibility property: for all $f\in\AR$ and $l\in\ZZ$ we have
\begin{equation}\label{Equ:evShift}
\forall n\geq\lambda:\,\ev(\sigma^l(f),n)=\ev(f,n+l)
\end{equation}
for some $\lambda\in\NN$. $L$ is called an \notion{$o^{\sigma}$-function for $\ev$} if $L$ is and $o$-function for $\ev$ and for all $f\in\AR$ and $l\in\ZZ$  with $\lambda=L(f)+\max(0,-l)$ property~\eqref{Equ:evShift} holds. $\ev$ is called \notion{operation bounded} for $\dfield{\AR}{\sigma}$ if there is such a function $L$: In particular, 
$\ev$ is called \notion{$o^{\sigma}$-computable} if $\ev$ is a computable function and there is a computable $o^{\sigma}$-function $\fct{L}{\AR}{\NN}$ for $\ev$.

\begin{example}\label{Exp:RatDFPlusEv}
Consider the rational difference field $\dfield{\FF}{\sigma}$ with $\FF=\KK(x)$, $\KK=\const{\KK}{\sigma}$ and $\sigma(x)=x+1$. Then~\eqref{Equ:EvalRat} is an evaluation function for $\dfield{\FF}{\sigma}$ and the function $L$ from Example~\ref{Exp:RationalDField1} is an $o^{\sigma}$-function. 
\end{example}

\noindent Using~\cite[Lemma~5.4 (parts 1,3)]{DR3} Lemma~\ref{Lemma:evExt} can be extended to difference rings.

\begin{lemma}\label{Lemma:LiftEvForShift}
Let $\dfield{\FF}{\sigma}$ be a difference field with $\KK=\const{\FF}{\sigma}$ equipped with an evaluation function $\fct{\ev}{\FF\times\NN}{\KK}$.
Let $\dfield{\EE}{\sigma}$ with $\EE=\FF\langle\hat{x}_1\rangle\dots\langle\hat{x}_r\rangle[z]$ be an $AP$-extension of $\dfield{\FF}{\sigma}$ where the $\hat{x}_i$ with $1\leq i\leq r$ are $P$-monomials  with $a_i=\frac{\sigma(\hat{x}_i)}{\hat{x}_i}$ and $z$ is an $A$-monomial with $\rho=\frac{\sigma(z)}{z}\in\KK^*$. For $1\leq i\leq r$, let $l_i\in\NN$ such that $\ev(a_i,n)\neq0$ holds for all $n\geq l_i$. Then $\fct{\hat{\ev}}{\EE\times\NN}{\KK}$ defined by~\eqref{Equ:EvForProdExt} is an evaluation function for $\dfield{\EE}{\sigma}$. If $L$ is an $o^{\sigma}$-function for $ev$, $\hat{L}$ defined in~\eqref{Equ:DefL'} is an $o^{\sigma}$-function for $\hat{L}$. If $L$ is computable, then $\hat{L}$ is computable.
\end{lemma}

\begin{example}[Cont.~Example~\ref{Exp:MainPExt}]\label{Exp:MainPExtPlusEv}
Consider the $P$-extension
$\dfield{\EE}{\sigma}$ of $\dfield{\FF}{\sigma}$ from Example~\ref{Exp:MainPExt} with $\EE=\FF\langle\hat{x}_1\rangle\langle\hat{x}_2\rangle\langle\hat{x}_3\rangle\langle\hat{x}_4\rangle$ and $\sigma(\hat{x}_i)=\hat{\alpha}_i\,\hat{x}_i$ for $1\leq i\leq 4$ where the $\hat{\alpha}_i$ are given in~\eqref{Equ:MainAlpha}, and take the evaluation function $\hat{\ev}$ defined in Example~\ref{Exp:EvMainExp}.
Then by Lemma~\ref{Lemma:LiftEvForShift} it follows that $\hat{\ev}$ is an evaluation function for $\dfield{\EE}{\sigma}$. 
\end{example}

It will be convenient to use the following convention.
Let $\dfield{\AR}{\sigma}$ be a difference ring with constant field $\KK$ and let $\ev$ be an evaluation function for $\dfield{\AR}{\sigma}$. We say that a \notion{sequence $(F(n))_{n\geq0}$ is modeled by $a\in\AR$} if there exists a $\lambda\in\NN$ with $\ev(a,n)=F(n)$ for all $n\geq\lambda$.  

\begin{example}[Cont.\ Example~\ref{Exp:MainPExtPlusEv}]
For $i=1,2,3,4$ the sequences $(F_i(n))_{n\geq0}$ are modeled by $x_i$, respectively. 
\end{example}

Let $\ev$ be an evaluation function for $\dfield{\AR}{\sigma}$ with constant field $\KK$. Then the ring homomorphism $\fct{\tau}{\AR}{\seqK}$ defined by~\eqref{Equ:tauDef} turns to a difference ring homomorphism.

More precisely, let $\dfield{\BB}{\sigma'}$ be another difference ring. Then a ring homomorphism (resp.\ injective ring homomorphism) $\fct{\lambda}{\AR}{\BB}$ is called a \notion{difference ring homomorphism} (resp \notion{difference ring embedding}) if for all $f\in\AR$ we have $\lambda(\sigma(f))=\sigma'(\lambda(f))$.

To turn $\fct{\tau}{\AR}{\seqK}$ to such a difference ring homomorphism, we consider the shift operation $\fct{\shiftS}{\seqK}{\seqK}$ defined by
$$\shiftS((a_0,a_1,a_2,\dots))=(a_1,a_2,\dots).$$
Then it can be easily verified that $\shiftS$ forms a ring automorphism and thus $\dfield{\seqK}{\shiftS}$ forms a difference ring; compare~\cite{AequalB}. In particular, due to~\eqref{Equ:evShift} it follows that the ring homomorphism $\tau$ defined by~\eqref{Equ:tauDef} is a difference ring homomorphism, i.e., we have
$$\forall f\in\AR:\quad \tau(\sigma(f))=\shiftS(\tau(f)).$$
Furthermore, by~\eqref{Equ:evC} it follows that $\tau(c)=\vect{c}$ with $\vect{c}:=(c,c,c,\dots)$ for all $c\in\KK$. Such a map will be also called $\KK$-homomorphism.

\begin{definition}\label{Def:ComputableTau}
Let $\dfield{\AR}{\sigma}$ be a difference ring with constant-field $\KK$. A function $\fct{\tau}{\AR}{\seqK}$ 
is called \notion{$\KK$-homomorphism} (resp.\ \notion{$\KK$-embedding}) if it is a difference ring homomorphism (resp.\ difference ring embedding) where $\tau(c)=\vect{c}$ for all $c\in\KK$. A $\KK$-homomorphism (resp.\ $\KK$-embedding) is called \notion{$o^{\sigma}$-computable} if there is a computable evaluation function $\fct{\ev}{\AR\times\NN}{\KK}$ with a computable $o^{\sigma}$-function $\fct{L}{\AR}{\NN}$ where~\eqref{Lemma:evExt} holds
for all $f\in\AR$.
\end{definition}

\begin{example}\label{Exp:RatDFPlusEv2}
Consider the rational difference field $\dfield{\FF}{\sigma}$ with $\FF=\KK(x)$ and $\sigma(x)=x+1$ from Example~\ref{Exp:RatDFPlusEv} equipped with the evaluation function~\eqref{Equ:EvalRat}.
Then $\fct{\tau}{\FF}{\seqK}$ given by~\eqref{Equ:tauDef} is a $\KK$-homomorphism. In particular, since any non-zero polynomial in $\KK[x]$ has only finitely many roots, $\tau(\tfrac{p}{q})=\vect{0}$ with $p\in\KK[x]$ and $q\in\KK[x]\setminus\{0\}$ iff $p=0$. Thus $\tau$ is injective, i.e., it is a $\KK$-embedding.
\end{example}

\noindent More generally, we will consider the class of mixed-rational difference fields~\cite{Bauer:99}.

\begin{example}\label{Exp:MixedDF}
Take the rational function field $\KK=\KK'(q_1,\dots,q_v)$ 
and  consider the rational function field $\FF=\KK(x,y_1,\dots,y_v)$ on top together with the field automorphism $\fct{\sigma}{\FF}{\FF}$ defined by $\sigma_{\KK}=\id$, $\sigma(x)=x+1$ and $\sigma(y_i)=q_i\,y_i$ for $1\leq i\leq v$. We noted already that $\const{\KK(x)}{\sigma}=\KK$; see Example~\ref{Exp:MainPExtSpec}.
It is not too difficult to see that there does not exist a $g\in\KK(x)^*$ and $(0,\dots,0)\neq(m_1,\dots,m_v)\in\ZZ^v$ with $\sigma(g)=q_1^{m_1}\dots q_v^{m_v}\,g$. Hence by Proposition~\ref{Prop:PiEquivalences} below it follows that $\const{\FF}{\sigma}=\const{\KK(x)}{\sigma}=\KK$. In short, $\dfield{\FF}{\sigma}$ is a \pisiE-field over $\KK$. $\dfield{\FF}{\sigma}$ is also called the \notion{mixed-rational difference field}.
Furthermore, for $f=\frac{p}{q}$ with $p,q\in\KK[x,y_1,\dots,y_v]$, $q\neq0$ and $p,q$ being co-prime we define
\begin{equation}\label{Equ:EvalMixed}
\ev(f,n)=\begin{cases}
0&\text{if }q(n,q_1^n,\dots,q_v^n)=0\\
\frac{p(n,q_1^n,\dots,q_v^n)}{q(n,q_1^n,\dots,q_v^n)}&\text{if }q(n,q_1^n,\dots,q_v^n)\neq0.
\end{cases}
\end{equation}
By~\cite[Sec.~3.7]{Bauer:99} there is a minimal
$\delta\in\NN$ with $q(n,q_1^n,\dots,q_v^n)\neq0$ for all $n\geq\delta$.
Hence we can define the $o^{\sigma}$-function $\fct{L}{\FF}{\NN}$ by $L(f)=\delta$; a $z$-function $\fct{Z}{\FF}{\NN}$ can be defined by $Z(f)=L(p\,q)$. 
Finally, consider the difference ring homomorphism $\fct{\tau}{\FF}{\seqK}$ defined by~\eqref{Equ:tauDef}.
Suppose that $\tau(f)=\vect{0}$ with $f=\frac{p}{q}$. Since $q(n,q_1^{n},\dots,q_v^{n})$ is non-zero for all $n\geq Z(f)$ and $p(n,q_1^{n},\dots,q_v^{n})$ is non-zero for all $n\geq Z(f)$ provided that $p\neq0$, it follows that $p=0$. Hence $\tau$ is injective.\\ 
We will restrict for algorithmic reasons to the case that $\KK'$ is a rational function field over an algebraic number field.
In this case, $L$ and $Z$ are computable by~\cite[Sec.~3.7]{Bauer:99}.
Summarizing, we obtain a $o^{\sigma}$-computable $\KK$-embedding $\fct{\tau}{\FF}{\seqK}$.
\end{example}

\begin{example}[Cont.\ Ex.~\ref{Exp:MainPExtPlusEv}]\label{Exp:MainPExtTau}
Consider the difference ring $\dfield{\EE}{\sigma}$ from Example~\ref{Exp:MainPExtPlusEv} equipped with the evaluation function defined in Example~\ref{Exp:EvMainExp} and the corresponding computable $o^{\sigma}$-function $L$ and the computable $z$-function $Z$. 
Thus $\fct{\hat{\tau}}{\EE}{\seqK}$ defined by $\hat{\tau}(f)=(\hat{\ev}(f,n))_{n\geq0}$ for $f\in\EE$ is a $o^{\sigma}$-computable $\KK$-homomorphism.
\end{example}

So far, we exploited the fact that an evaluation function produces a $\KK$-homo\-morphism. Later we will use the reverse construction: for a $\KK$-homomorphism of a $P$-extension there exists an evaluation function that is based on product evaluations. More precisely, \cite[Lemma~5.4 (parts 2,3)]{DR3} provides the following result.

\begin{lemma}\label{Lemma:ExtractEvFunction}
Let $\dfield{\FF}{\sigma}$ be a difference field with constant field $\KK$ equipped with a $z$-function $\fct{Z}{\FF}{\NN}$. Let $\dfield{\EE}{\sigma}$ be an $AP$-extension of $\dfield{\FF}{\sigma}$ with $\EE=\FF\langle\hat{x}_1\rangle\dots\langle \hat{x}_r\rangle[z]$ where for $1\leq i\leq r$ the $\hat{x}_i$ with $\hat{\alpha}_i=\frac{\sigma(\hat{x}_i)}{\hat{x}_i}\in\FF^*$ are $P$-monomials and $z$ with $\rho=\frac{\sigma(z)}{z}\in\KK^*$ is an $A$-monomial of order $\lambda$. Suppose that there is a $\KK$-homomorphism $\fct{\hat{\tau}}{\EE}{\seqK}$ and let
$\fct{\ev}{\FF\times\NN}{\KK}$ be an evaluation function with $\hat{\tau}(f)=(\ev(f,n))_{n\geq0}$ for $f\in\FF$. Then there is an evaluation function $\fct{\hat{ev}}{\EE\times\NN}{\KK}$ defined by $\hat{\ev}|_{\FF\times\NN}=\ev$,
$$\hat{\ev}(\hat{x}_i,n)=\kappa_i\,\prod_{i=l_i}^n\ev(\hat{\alpha}_i,k-1),\quad1\leq i\leq r$$
for some $\kappa_i\in\KK^*$ and\footnote{Note that $\hat{\ev}(\hat{x}_i,n)\neq0$ for all $n\geq l_i$ by part~(1) of Lemma~\ref{Lemma:ZFunctionForField}.} $l_i\in\NN$, and 
$\hat{\ev}(z,n)=\langle c\,\rho^n\rangle_{n\geq0}$
for some $c\in\KK^*$ with $c^\lambda=1$ such that
$\hat{\tau}(f)=\langle\hat{\ev}(f,n)\rangle_{n\geq0}$ holds
for all $f\in\EE$.
Furthermore, if there is a computable $o^{\sigma}$-function $\fct{L}{\FF}{\NN}$ for $\ev$ and $Z$ is computable, there is a computable $o^{\sigma}$-function $\fct{\hat{L}}{\EE}{\NN}$ for $\hat{\ev}$. In particular, $c\in\KK^*$, and by choosing $l_i=\max(Z(\hat{\alpha}_i)-1,L(\hat{\alpha}_i))$ for $1\leq i\leq r$ the $\kappa_i\in\KK^*$ can be computed. 
\end{lemma}

\subsection{The problem description in the setting of $R\Pi$-extensions}\label{Subsec:ProblemDescription}
Consider the products given in~\eqref{Equ:InputProds} where the multiplicands $f_1(k),\dots,f_r(k)$
can be modeled in a difference field $\dfield{\FF}{\sigma}$ with constant field $\KK$. This means that there is an evaluation function $\fct{\bar{\ev}}{\FF\times\NN}{\KK}$ with an $o^{\sigma}$-function $\fct{\bar{L}}{\FF}{\NN}$ such that we can find $\hat{f}_1,\dots,\hat{f}_r\in\FF^*$ with the following property: for all $1\leq i\leq r$ we have $\bar{\ev}(\hat{f}_i,n)=f_i(n)$ for all $n\geq l_i$ with $l_i\geq \bar{L}(\hat{f}_i)$. 
Then we can model also the products of~\eqref{Equ:InputProds} in a $P$-extension $\dfield{\EE}{\sigma}$ of $\dfield{\FF}{\sigma}$ with $\EE=\FF\langle \hat{x}_1\rangle\dots\langle \hat{x}_r\rangle$ and $\sigma(\hat{x}_i)=\hat{\alpha}_i\,\hat{x}_i$ where $\hat{\alpha}_i=\sigma(\hat{f}_i)$:
By Lemma~\ref{Lemma:evExt} we can define the evaluation function $\fct{\ev}{\EE\times\NN}{\KK}$ with $\ev|_{\FF\times\NN}=\bar{\ev}$ and
\begin{equation}\label{Equ:EvalhatXi}
\ev(\hat{x}_i,n)=F_i(n)=\prod_{k=l_i}^n f_i(k).
\end{equation}
In addition, we obtain an $o^{\sigma}$-function $L$ for $\ev$.
Finally, we take the $\KK$-homo\-morphism $\fct{\tau}{\EE}{\seqK}$ defined by~\eqref{Equ:tauDef}. For further considerations we require in addition that $\tau|_{\FF}$ is a $\KK$-embedding. If we focus on algorithmic aspects, we assume that the $\KK$-embedding $\tau$ is $o^{\sigma}$-computable.

As elaborated in the next remark, a finite set of mixed-hypergeometric products (see  Definition~\ref{Def:MMHypergeometric})
can be modeled in such a $P$-extension.

\begin{remark}[Representation of mixed-hypergeometric products]\label{Remark:MixedCaseWorks}
Take the rational function field $\KK=\KK'(q_1,\dots,q_v)$ and consider the mixed-rational difference field $\dfield{\FF}{\sigma}$ with $\FF=\KK(x,y_1,\dots,y_v)$ where $\fct{\sigma}{\FF}{\FF}$ is defined by $\sigma|_{\KK}=\id$, $\sigma(x)=x+1$ and $\sigma(y_i)=q_i\,y_i$ for $1\leq i\leq v$. Note that $\dfield{\FF}{\sigma}$ is a \pisiE-field over $\KK$. Furthermore take its evaluation function 
$\fct{\bar{\ev}}{\FF\times\NN}{\KK}$ with an $o^{\sigma}$-function $\fct{\bar{L}}{\FF}{\NN}$; see Example~\ref{Exp:MixedDF}.
Suppose that we are given the mixed-hypergeometric products $F_1(n),\dots,F_{r}(n)$ with 
$$F_i(n)=\prod_{k=l_i}^nf_{i}(k,q_1^k,\dots,q_v^k)$$ 
for $1\leq i\leq r$ where\footnote{To fulfill the property $f_i\in\FF^*$, the variables $q_1\dots,q_v$ and $y_1,\dots,y_v$ and the field below $\KK'$ have to be set up accordingly.} $f_i(x,y_1,\dots,y_v)\in\FF^*$ and the numerator and denominator of $f_{i}(k,q_1^k,\dots,q_v^k)$ do not evaluate to zero for all $k\geq l_r\in\NN$.
By construction we can take $\hat{f}_r:=f_r$ and get $\bar{\ev}(\hat{f}_r,k)=f_r(k,q_1^k,\dots,q_v^k)$.
Now take the $P$-extension $\dfield{\EE}{\sigma}$ of $\dfield{\FF}{\sigma}$ with $\EE=\FF\langle\hat{x}_1\rangle\dots\langle\hat{x}_{r}\rangle$ and $\sigma(\hat{x}_i)=\hat{\alpha}_i\,\hat{x}_i$ where
$\hat{\alpha}_i=\sigma(\hat{f}_i)=\hat{f}_i(x+1,q_1\,y_1,\dots,q_{r}\,y_v)$; note that $\ev(\hat{\alpha}_i,n)=\ev(\sigma(\hat{f}_i),n)=\ev(\hat{f}_i,n+1)=f_i(n+1,q_1^{n+1},\dots,q_v^{n+1})$.
By Lemma~\ref{Lemma:LiftEvForShift} we can extend the evaluation function $\bar{\ev}$ from $\dfield{\FF}{\sigma}$ to $\dfield{\EE}{\sigma}$ with $\ev(\hat{x}_i,n)=\prod_{k=l_i}^nf_i(k,q_1^k,\dots,q_v^k)$; similarly one can extend $\bar{L}$ to the $o^{\sigma}$-function $\fct{L}{\EE}{\NN}$. As elaborated in Example~\ref{Exp:MixedDF}, $\fct{\tau}{\EE}{\seqK}$ defined by~\eqref{Equ:tauDef} is a $\KK$-embedding. If we restrict to the case that the subfield $\KK'$ of $\KK$ is a rational function field over an algebraic number field, all ingredients are computable. In particular, $\tau$ turns to a $o^{\sigma}$-computable $\KK$-embedding. 
\end{remark}

\noindent 
Given the above construction, we will consider the following problem.

\medskip

\noindent\textbf{Problem \DR} (Difference ring Representation): \textit{Given} a computable difference field $\dfield{\FF}{\sigma}$ equipped with a computable evaluation function $\bar{ev}$ and computable $o^{\sigma}$-function $\bar{L}$ and \textit{given} a $P$-extension $\dfield{\EE}{\sigma}$ of $\dfield{\FF}{\sigma}$ with a computable evaluation function $\ev$ with $\ev|_{\FF\times\NN}=\bar{\ev}$ and a computable $o^{\sigma}$-function $L$ as described above\footnote{Note that generators $\hat{x}_i$ in $\EE$ represent products with no extra properties: in particular, all algebraic relations induced by their sequence evaluations are ignored.}.\\
\textit{Find} 
\begin{itemize}
\item an $AP$-extension $\dfield{\HH}{\sigma}$ of $\dfield{\FF}{\sigma}$ with
\begin{equation}\label{Equ:HDef}
\HH=\FF\langle t_1\rangle\dots\langle t_s\rangle[z_1]\dots[z_l]
\end{equation}
where the $t_i$ for $1\leq i\leq s$ are $P$-monomials with $\alpha_i=\frac{\sigma(t_i)}{t_i}\in\FF^*$ and the $z_i$ for $1\leq i\leq l$ are $A$-monomials with $\rho_i=\frac{\sigma(z_i)}{z_i}\in\KK^*$ of order $d_i>1$;
\item a computable evaluation function $\fct{\ev'}{\HH\times\NN}{\KK}$ for $\dfield{\HH}{\sigma}$ with a computable $o^{\sigma}$-function $L'$ where $\ev'$ is defined by $\ev'|_{\FF\times\NN}=\ev|_{\FF\times\NN}(=\bar{\ev})$, $\ev'(t_i,n)=\prod_{k=l'_i}^n\ev(\alpha_i,k-1)$ with $l'_i\in\NN$ for $1\leq i\leq s$, and $\ev'(z_i,n)=\rho_i^n$ for $1\leq i\leq l$;
\item a computable \textbf{surjective} difference ring homomorphism $\fct{\lambda}{\EE}{\HH}$
\end{itemize}
such that
\begin{enumerate}
\item $\fct{\tau'}{\EE}{\seqK}$ defined by $\tau'(f)=(\ev'(f,n))_{n\geq0}$ is a $\KK$-embedding;
\item for all $f\in\EE$ we have $\tau(f)=\tau'(\lambda(f))$, i.e., the following diagram commutes:
 \begin{equation}\label{Equ:Diagram}
 \xymatrix@!R=0.7cm@C1.8cm{
\EE\ar[r]^{\lambda}\ar[dr]_{\tau} &\HH\ar@{^{(}->}[d]_{\tau'}\\
 &\seqK.}
 \end{equation} 
\end{enumerate}
We emphasize that a solution of Problem~\DR\ solves also Problem~\RPTB.
More precisely, for $1\leq i\leq r$ we can take $g_i:=\lambda(\hat{x}_i)\in\HH$ and get
$$\ev'(g_i,n)=\ev'(\lambda(\hat{x}_i),n)=\ev(\hat{x}_i,n)=F_i(n)$$
for all $n\geq\max(L(f_i),L'(g_i))$. Thus $F_i(n)$ is modeled by $g_i$ in $\dfield{\EE}{\sigma}$. In particular, the evaluation function provides the full information to obtain an alternative production expression that evaluates to $F_i(n)$
(essentially, one replaces in $g_i$ the $t_j$ for $1\leq j\leq s$ by the products $\prod_{k=l'_j}^n\ev(\alpha_j,k-1)$ and the $z_j$ for $1\leq j\leq l$ by $\rho_j^n$ respectively). This establishes part (i) of Problem~\RPTB. Furthermore,
by Remark~\ref{Remark:Rearrange} we can reorder the generators in $\EE$ and get the $AP$-extension $\dfield{\EE'}{\sigma}$ of $\dfield{\FF}{\sigma}$ with $\EE'=\FF[z_1]\dots[z_l]\langle t_1\rangle\dots\langle t_s\rangle$. 
Furthermore, we observe that the difference ring $\dfield{\tau(\EE)}{\Shift}$ is contained in $\dfield{\seqK}{\Shift}$ (as a subdifference ring) and is isomorphic to $\dfield{\EE}{\sigma}$. 
Thus $\tau(\EE')=\tau(\FF[z_1]\dots[z_l])\langle \tau'(t_1)\rangle\dots\langle \tau'(t_s)\rangle$ forms a ring of Laurent polynomials, i.e., the sequences $\tau'(t_i)=(\prod_{k=l'_i}^n\ev(\alpha_i,k-1))_{n\geq0}$ are algebraically independent over the ring $\tau(\FF[z_1]\dots[z_l])$ (which is a subring of the ring of sequences $\seqK$). Hence also part (ii) of 
Problem~\RPTB\ is tackled.

The solution of Problem~\DR\ is strongly related with the following result which is a special case of~\cite[Thm.~5.14]{DR3}; this result is also connected to~\cite{Singer:97,Schneider:10c,Singer:08}.
\begin{theorem}\label{Thm:RPSImpliesInjective}
 Let $\dfield{\HH}{\sigma}$ with~\eqref{Equ:HDef} be an $AP$-extension of a difference field $\dfield{\FF}{\sigma}$ with $\KK=\const{\FF}{\sigma}$ as supposed in Problem~\DR. Furthermore, let $\fct{\tau}{\HH}{\seqK}$ be a $\KK$-homomorphism where $\tau|_{\FF}$ is injective.  Then $\const{\EE}{\sigma}=\KK$ iff $\tau$ is injective.
\end{theorem}

Consequently property (1) in Problem~\DR\ can be dropped by imposing that $\const{\HH}{\sigma}=\const{\FF}{\sigma}$, i.e., by taking care that $\dfield{\HH}{\sigma}$ is an $R\Pi$-extension of $\dfield{\FF}{\sigma}$. 
Precisely this construction for a solution of Problem~\DR\ has been carried out in~\cite{Schneider:05c,DR2} for hypergeometric products, has been extended to a complete algorithm for mixed-hypergeometric products in~\cite{OS:18}, and has been generalized further for nested products in~\cite{Ocansey:19}. However, these approaches usually find a difference ring homomorphism $\lambda$ that is not surjective.

In this article we will follow the same tactic to solve Problem~\DR\ for single nested products such that $\lambda$ is always surjective. Note that in concrete examples the cases 
$s=0$ or $l=0$ might arise, i.e., no \piE-monomials or no $R$-monomials are needed.
E.g., if $\dfield{\EE}{\sigma}$ itself is a \piE-extension of $\dfield{\FF}{\sigma}$, one can solve Problem~\DR\ by taking $\HH=\EE$, $\lambda=\id$ and $\tau'=\tau$; note that in this special case $\lambda$ is even bijective. Otherwise,
we will show in Theorem~\ref{Thm:MainResult} below that we can solve Problem~\DR\ with a surjective $\lambda$ and an $R\Pi$-extension $\dfield{\HH}{\sigma}$ of $\dfield{\FF}{\sigma}$ 
with~\eqref{Equ:HDef} where $s\geq0$ and $l\in\{0,1\}$ (i.e., at most one $A$-monomial is needed) such that 
\begin{itemize}
\item the number $s$ of \piE-monomials is minimal,
\item the $R$-monomial $z:=z_1$ is introduced only if it is necessary,
\item and if it is necessary, the order $d:=d_1$ of $z$ is minimal\footnote{
Note that by Proposition~\ref{Prop:CharactSeveralRExt} given below a difference ring generated by a finite set of $R$-monomials $z'_1,\dots,z'_l$ is isomorphic to a difference ring generated by only one $R$-monomial $z$ with $\ord(z)=\lcm(\ord(z'_1),\dots,\ord(z'_l))=\ord(z'_1)\dots\ord(z'_l)$. Hence claiming that the order $d=\ord(z)$ is optimal means that among all solutions of Problem~\DR\ ($\lambda$ need not to be surjective) in a difference ring with the $A$-monomials $z'_1,\dots,z'_l$  the order $d$ of $z$ is smaller or equal to $\ord(z'_1)\dots\ord(z'_l)$.}.
\end{itemize}

\begin{example}\label{Exp:IntroExp}
Recall the following naive constructions from Examples~\ref{Exp:RationalDField1}, \ref{Exp:EvMainExp}, \ref{Exp:MainPExt}, \ref{Exp:MainPExtPlusEv} and~\ref{Exp:MainPExtTau} to represent the products~\eqref{Equ:DefineFi} in a difference ring.
We take the rational difference field $\dfield{\KK(x)}{\sigma}$ with $\sigma(x)=x+1$ and $\const{\KK(x)}{\sigma}=\KK=\QQ(\iota)$, which is a \pisiE-field over $\KK$. Furthermore we construct the $P$-extension $\dfield{\EE}{\sigma}$ of $\dfield{\KK(x)}{\sigma}$ with $\EE=\KK(x)\langle\hat{x}_1\rangle\langle\hat{x}_2\rangle\langle\hat{x}_3\rangle\langle \hat{x}_4\rangle$ where $\sigma(\hat{x}_i)=\hat{\alpha}_i\,\hat{x}_i$ for $1\leq i\leq 4$ and~\eqref{Equ:MainAlpha} as introduced in Example~\ref{Equ:MainExample:Start}. Further, we take the
evaluation function $\fct{\hat{\ev}}{\EE\times\NN}{\KK}$ introduced in Example~\ref{Exp:EvMainExp} ($\ev$ replaced by $\bar{\ev}$) and define the 
$\KK$-homomorphism $\fct{\hat{\tau}}{\EE}{\seqK}$ by $\hat{\tau}(f)=(\hat{\ev}(f,n))_{n\geq0}$.
Then based on our main results in Theorem~\ref{Thm:MainResult} below
we can solve Problem~\DR. Namely, as will be carried out in Example~\ref{Exp:MainExampleResult} below we can construct\footnote{In order to fit the specification in Problem~\DR, we set $t_1:=x_3$ and $t_2:=x_4$ and $z_1:=z$.} the $R\Pi$-extension $\dfield{\HH}{\sigma}$ of $\dfield{\KK(x)}{\sigma}$ with $\HH=\KK(x)\langle x_3\rangle\langle x_4\rangle[z]$ and
\begin{align*}
\sigma(x_3)&=-\tfrac{\iota
   (x+4)^6}{9 (x+1) (x+2)^2 (x+3)
   (x+6)}\,x_3,&
\sigma(x_4)&=-\tfrac{162 (x+1) (x+3)}{x+6}\,x_4,&\sigma(z)&=-z.
\end{align*}
Further, we can construct the evaluation function $\fct{\ev'}{\HH\times\NN}{\KK}$ for $\dfield{\HH}{\sigma}$ with 
\begin{align}\label{Eq:MainExpamleEvP}
\ev'(x_3,n)&=\sprod{k=1}n \tfrac{-\iota (3+k)^6}{9 k (1+k)^2 (2+k) (5+k)},&\!\!
\ev'(x_4,n)&=\sprod{k=1}n \tfrac{-162 k (2+k)}{5+k},&\!\!\ev'(z,n)=(-1)^n.
\end{align}
Note that the $\KK$-homomorphism $\fct{\tau'}{\HH}{\seqK}$ defined by $\tau'(f)=(\ev'(f,n))_{n\geq0}$ is injective by Theorem~\ref{Thm:RPSImpliesInjective}. In addition we obtain in Example~\ref{Exp:MainExampleResult} below the surjective difference ring homomorphism $\fct{\hat{\lambda}}{\EE}{\HH}$ defined by $\hat{\lambda}|_{\KK(x)}=\id$ and
\begin{equation}\label{Equ:LambdaMainExp}
 \begin{aligned}
\hat{\lambda}(\hat{x}_1)&=\frac{5 (x+1)^2 (x+2)^5 (x+3)^8x_4
   z}{52488 (x+4) (x+5)x_3^2}&
\hat{\lambda}(\hat{x}_2)&=\frac{1}{400} (x+4)^2 (x+5)^2
   x_4^2\\
\hat{\lambda}(\hat{x}_3)&=\frac{2754990144 (x+4)^2 (x+5)^2
   x_3^3}{25 (x+1)^4 (x+2)^{10}(x+3)^{16}}&
\hat{\lambda}(\hat{x}_4)&=x_4
\end{aligned}
\end{equation}
such that $\hat{\tau}=\tau'\circ\hat{\lambda}$ holds.
Due to our main result stated in Theorem~\ref{Thm:MainResult} below it will follow that the number $s=2$ of the \piE-monomials and the order $\lambda=2$ of the $R$-monomial $z$ are optimal among all possible solutions of Problem~\DR\ (where $\lambda$ is not necessarily surjective).
The construction implies that $\hat{\ev}$ can be given in the following alternative form:
\begin{equation}\label{Equ:AlternativeEv}
\begin{aligned}
\hat{\ev}(\hat{x}_1,n)&=\tfrac{5 (1+n)^2 (2+n)^5 (3+n)^8}{52488 (4+n) (5+n)} (-1)^n
        \prod_{k=1}^n -\tfrac{162 k (2+k)}{5+k}
\Big(\prod_{k=1}^n -\tfrac{i (3+k)^6}{9 k (1+k)^2 (2+k) (5+k)} \Big)^{-2},\\
\hat{\ev}(\hat{x}_2,n)&=\tfrac{(4+n)^2 (5+n)^2 }{400}\Big(
        \prod_{k=1}^n -\tfrac{162 k (2+k)}{5+k}\Big)^2,\\
\hat{\ev}(\hat{x}_3,n)&=\tfrac{2754990144 (4+n)^2 (5+n)^2}{25 (1+n)^4 (2+n)^{10} (3+n)^{16}} \Big(
        \prod_{k=1}^n -\tfrac{\iota(3+k)^6}{9 k (1+k)^2 (2+k) (5+k)}\Big)^3,\\
\hat{\ev}(\hat{x}_4,n)&=\prod_{k=1}^n -\tfrac{162 k (2+k)}{5+k} 
\end{aligned}
\end{equation}
which is precisely~\eqref{Equ:FiEvalNew}. Furthermore, $\tau'$ is a difference ring embedding which implies that 
$\tau'(\KK[z])\langle (\ev'(x_3,n))_{n\geq0}\rangle\langle (\ev'(x_4,n))_{n\geq0}\rangle$ with $\ev'$ defined in~\eqref{Eq:MainExpamleEvP} forms a Laurent polynomial ring with coefficients from the subring $\tau'(\FF[z])$ of $\seqK$. 
\end{example}


But even more will be derived.
By the first isomorphism theorem 
we get the ring isomorphism $\fct{\mu}{\EE/\ker(\lambda)}{\HH}$ defined by $\mu(f+I)=\lambda(f)$.
Since $\ker(\lambda)$ is a reflexive difference ideal (i.e., $\sigma^k(f)\in\ker(\lambda)$ for all $k\in\ZZ$ and $f\in\ker(\lambda)$) it follows that $\fct{\sigma'}{\EE/\ker(\lambda)}{\EE/\ker(\lambda)}$ with $\sigma'(f+I):=\sigma(f)+I$ forms a ring automorphism; compare~\cite{Cohn:65}. In particular, $\mu$ turns into a difference ring isomorphism between $\dfield{\EE/\ker(\lambda)}{\sigma'}$ and $\dfield{\HH}{\sigma}$. 
As a bonus we have 
\begin{equation}\label{Equ:DiffIdealI}
I:=\ker(\lambda)=\ker(\tau)
\end{equation}
which follows by the following lemma (by setting $\SA=\seqK$).
\begin{lemma}\label{Lemma:ConnectLemmaToTu}
Let $\EE$, $\HH$ and $\SA$ be rings with a ring homomorphism $\fct{\tau}{\EE}{\SA}$ and a a ring embedding $\fct{\tau'}{\HH}{\SA}$. If $\fct{\lambda}{\EE}{\HH}$ is a ring homomorphism with $\tau'(\lambda(a))=\tau(a)$ for all $a\in\EE$, then $\ker(\lambda)=\ker(\tau)$.
\end{lemma}
\begin{proof}
Suppose that $\lambda$ is a ring homomorphism as claimed in the lemma and let $a\in\EE$. If $\lambda(a)=0_{\HH}$, then $0_{\SA}=\tau'(0_{\HH})=\tau'(\lambda(a))=\tau(a)$. Conversely, if $\tau(a)=0_{\SA}$, then $\tau'(\lambda(a))=0_{\SA}$. Since $\tau'$ is injective and $\tau'(0_{\HH})=0_{\SA}$, $\lambda(a)=0_{\HH}$.
\end{proof}
Since $\tau'$ is injective by construction, we finally obtain the following difference ring isomorphisms:
\begin{equation}\label{Equ:isomorphDR}
\dfield{\EE/\ker(\tau)}{\sigma'}=\dfield{\EE/\ker(\lambda)}{\sigma'}\simeq\dfield{\HH}{\sigma}\simeq\dfield{\tau'(\HH)}{\shiftS}.
\end{equation}

In addition, we will elaborate in Theorem~\ref{Thm:MainResult} that one can compute explicitly a finite set of generators that span the difference ideal~\eqref{Equ:DiffIdealI}. This means that we obtain the full information of all the algebraic relations of the sequences $\tau(\hat{x}_1),\dots,\tau(\hat{x}_r)$. In particular, the corresponding mappings in~\eqref{Equ:isomorphDR} can be carried out explicitly.

\begin{example}[Cont. Ex.~\ref{Exp:IntroExp}]
In Example~\ref{Exp:MainExampleResult} (based on Theorem~\ref{Thm:MainResult}) we will get
\begin{equation}\label{Equ:kertauExp}
\begin{aligned}
 \ker(\hat{\tau})=\ker(\hat{\lambda})=
\langle&\frac{\hat{x}_2}{\hat{x}_4^2}-\frac1{400} (4+x)^2 (5+x)^2,\\
&\frac{\hat{x}_2^2 \hat{x}_4^2}{\hat{x}_1^6 \hat{x}_3^4}
-\tfrac1{4199040^2}(1+x)^4 (2+x)^{10} (3+x)^{16} (4+x)^2 (5+x)^2\rangle_{\EE}.
\end{aligned}
\end{equation}
Note that this yields~\eqref{Equ:IdealSetZ} with~\eqref{Equ:riDef}.
\end{example}

As already indicated in Remark~\ref{Remark:MixedCaseWorks}, this algorithmic toolbox is applicable for the mixed-rational difference field $\dfield{\FF}{\sigma}$ given in Example~\ref{Exp:MixedDF}. More generally, we will show in Theorem~\ref{Thm:MainResult} below that we obtain a  complete algorithm for Problem~DR\ if the ground field $\dfield{\FF}{\sigma}$ with $\KK=\const{\FF}{\sigma}$ satisfies various (algorithmic) properties. These properties (see Assumption~\ref{Assum:AlgProp}) will be explored further in the next section.

\section{Required properties of the underlying difference field}\label{Sec:BasicProperties}

For our proposed strategy to solve Problem~\DR\ (see Subsection~\ref{Subsec:ProblemDescription}) we will rely on the following properties of the ground field $\dfield{\FF}{\sigma}$.¸

\begin{assum}\label{Assum:AlgProp}
({Required properties of the ground field $\dfield{\FF}{\sigma}$ for Problem~\DR}):
\begin{enumerate}
 \item $\dfield{\FF}{\sigma}$ is computable: the addition, multiplication, the inversion of elements in $\FF$, and the automorphism $\sigma$ are computable.
 \item There is an $o^{\sigma}$-computable difference ring embedding $\fct{\tau}{\FF}{\seqK}$ together with a computable $z$-function; see Definitions~~\ref{Def:EvZ} and \ref{Def:ComputableTau} from above.
 \item $\dfield{\FF}{\sigma}$ is radical-stable; see Definition~\ref{Def:RacialStable} below.
 \item One can solve homogeneous first-order difference equations in $\dfield{\FF}{\sigma}$: given $w\in\FF^*$, one can compute, if it exists, a $g\in\FF^*$ with $\sigma(g)-w\,g=0$. 
 \item For $\alpha_1,\dots,\alpha_r\in\FF^*$ one can compute a $\ZZ$ basis of
$M((\alpha_1,\dots,\alpha_r),\FF)$; see Definition~\ref{Def:MSet} below.
\end{enumerate}
\end{assum}

In the following we will elaborate further details concerning these properties. In particular, we will show  the following result in Subsection~\ref{Sec:MixedRatDF}.

\begin{theorem}\label{Thm:MixedIsGood}
Let $\dfield{\FF}{\sigma}$ be a mixed rational difference field where the constant field is built by a rational function field over an algebraic number field. Then $\dfield{\FF}{\sigma}$ satisfies the properties listed in Assumption~\ref{Assum:AlgProp}.  
\end{theorem}

\subsection{Constant stability}\label{Subsec:constantstable}

Note that for any $k\in\NN$ we have $\const{\FF}{\sigma}\subseteq\const{\FF}{\sigma^k}$. In the following we will often assume that the two sets are equal; compare~\cite{DR2,DR3}.

\begin{definition}\label{Def:ConstantStable}
A difference ring (resp.\ field) $\dfield{\AR}{\sigma}$ is called \notion{constant-stable} if $\const{\AR}{\sigma^k}=\const{\AR}{\sigma}$ for all $k\in\NN$.
\end{definition}

We will show in Proposition~\ref{Lemma:ProofForConstantStable} below that a nested \pisiE-extension is constant-stable if the ground field is constant-stable. Here we will utilize

\begin{lemma}\label{Lemma:ProofForConstantStable}
Let $\dfield{\FF(t)}{\sigma}$ be a \pisiE-field extension of $\dfield{\FF}{\sigma}$ with $\sigma(t)=\alpha\,t+\beta$ ($\alpha=1$, or $\alpha\in\FF^*$ and $\beta=0$). Then the following holds.
\begin{enumerate}
\item If $a\in\FF[t]$ is monic and irreducible, then: $\gcd(\sigma^m(a),a)=1$ for all $m\geq1$ if and only if $\beta\neq0$ or $a\neq t$.
\item Suppose that $a\in\FF[t]$ is monic and irreducible with $\beta\neq0$ or $a\neq t$. Then for any irreducible $b\in\FF[t]$ there is no $m\in\ZZ$ with $\gcd(\sigma^m(a),b)\neq1$ or there is precisely one $m\in\ZZ$ with $\gcd(\sigma^m(a),b)=1$.
\item If $a\in\FF(t)\setminus\FF$ and $m\geq1$, then $a\,\sigma(a)\dots\sigma^{m}(a)\notin\FF$
\item If $\dfield{\FF}{\sigma}$ is constant-stable, then $\dfield{\FF(t)}{\sigma}$ is constant-stable
\end{enumerate}
\end{lemma}
\begin{proof}
(1) This follows by~\cite{Karr:81,Bron:00}; see also~\cite[Theorem~2.2.4]{Schneider:01}.\\
(2) Let $a$ be monic and irreducible as stated above and suppose that there is an irreducible $b\in\FF[t]$ with $\gcd(\sigma^m(a),b)\neq1$ and $\gcd(\sigma^n(a),b)\neq1$ for some $0\leq m<n$. Then $b=u\,\sigma^m(a)$ and $b=v\,\sigma^n(a)$ for some $u,v\in\FF^*$. Thus $\sigma^m(a)=\frac{v}{u}\sigma^n(a)$ and hence $a=\sigma^{-m}(\frac{v}{u})\sigma^{n-m}(a)$. By part (1), it follows that $n=m$.\\  
(3) Define $h=a\,\sigma(a)\dots\sigma^{m}(a)$. If $t$ is a \piE-monomial (i.e., $\beta=0$) and $a=b\,t^r$ for some $b\in\FF$ and $r\in\ZZ$, we get $h=q\,t^{m\,r}$ for some $q\in\FF^*$ and thus $h\notin\FF$.
Otherwise, by part (2) of our lemma we can take a monic irreducible $f\in\FF[t]$ (which is not $t$ if $\beta=0$) that arise in $a$ such that for any irreducible polynomial $g$ in $a$ there is no $m\geq1$ with $\gcd(\sigma^m(f),g)\neq1$ (If there are several irreducible factors $f_1,f_2\in\FF[t]$ in $a$ with $\sigma^m(f_1)=u\,f_2$ where $u\in\FF^*$ and $m\in\ZZ\setminus\{0\}$, we take that one which is related to the others only by positive shifts $m$). Thus $f$ occurs in $a$, but not in the elements $\sigma(a),\sigma^2(a),\dots,\sigma^m(a)$. Therefore $f$ cannot cancel in $h$ and it follows that $h\notin\FF$.\\
(4) Let $a\in\FF(t)$ and suppose that $\sigma^m(a)=a$ for some $m\geq2$. If $a\in\FF$, then $a\in\const{\FF}{\sigma}$ by assumption. Otherwise if $a\notin\FF$, define $h:=a\,\sigma(a)\dots\sigma^{m-1}(a)$. By part~(3) it follows that $h\notin\FF$. However,
$\frac{\sigma(h)}{h}=\frac{\sigma^m(a)}{a}=1$ and hence $h\in\const{\FF(t)}{\sigma}=\const{\FF}{\sigma}\subseteq\FF$, a contradiction. Summarizing, $\const{\FF(t)}{\sigma^m}\subseteq\const{\FF(t)}{\sigma}$ and thus equality holds.
 \end{proof}

Applying part~(4) of Lemma~\ref{Lemma:ProofForConstantStable} iteratively we obtain 
\begin{proposition}\label{Prop:ConstantStableFieldVersion}
Let $\dfield{\EE}{\sigma}$ be a (nested) \pisiE-field extension of $\dfield{\FF}{\sigma}$. If $\dfield{\FF}{\sigma}$ is constant-stable, then $\dfield{\EE}{\sigma}$ is constant-stable. 
\end{proposition}

We note that a difference field $\dfield{\KK}{\sigma}$ with $\sigma=\id$ is trivially constant-stable. Thus we rediscover the following result of~\cite{Karr:81}.

\begin{corollary}\label{Cor:PisiFieldIsConstantStable}
A \pisiE-field is constant-stable. 
\end{corollary}

In particular, one can enhance Proposition~\ref{Prop:ConstantStableFieldVersion} to the following ring-version.

\begin{corollary}\label{Cor:NestedPIsPiAndCS}
Let $\dfield{\EE}{\sigma}$ be a \pisiE-extension of a difference field $\dfield{\FF}{\sigma}$. If $\dfield{\FF}{\sigma}$ is constant-stable, $\dfield{\EE}{\sigma}$ is constant-stable.
\end{corollary}
\begin{proof}
Let $\EE=\FF\langle t_1\rangle\dots\langle t_e\rangle[s_1]\dots[s_l]$ where the $t_i$ are \piE-monomials and the $s_i$ are \sigmaE-monomials. Consider the $PS$-field extension $\dfield{\HH}{\sigma}$ of $\dfield{\FF}{\sigma}$ with $\HH=\FF(t_1)\dots(t_e)(s_1)\dots(s_l)$ in which $\dfield{\EE}{\sigma}$ is contained as sub-difference ring of $\dfield{\HH}{\sigma}$. By Prop.~\ref{Prop:PiEquivalences} $\const{\HH}{\sigma}=\const{\FF}{\sigma}$
and by Prop.~\ref{Prop:ConstantStableFieldVersion} $\dfield{\HH}{\sigma}$ is constant-stable. Thus for any integer $m>0$, $\const{\FF}{\sigma}=\const{\HH}{\sigma}=\const{\HH}{\sigma^m}\supseteq\const{\EE}{\sigma^m}\supseteq\const{\EE}{\sigma}=\const{\FF}{\sigma}$ and hence $\const{\EE}{\sigma^m}=\const{\EE}{\sigma}(=\const{\FF}{\sigma})$.
\end{proof}

Later we assume that there is a $\KK$-embedding $\fct{\tau}{\FF}{\seqK}$ of a difference field $\dfield{\FF}{\sigma}$ with constant field $\KK$. In this case we can specialize~\cite[Lemma~5.12]{DR3} to 

\begin{lemma}\label{Lemma:TauImpliesConstantStable}
Let $\dfield{\FF}{\sigma}$ be a difference field with $\KK=\const{\FF}{\sigma}$. If there is a $\KK$-embedding $\fct{\tau}{\FF}{\seqK}$, then $\dfield{\FF}{\sigma}$ is constant-stable. 
\end{lemma}

\subsection{The shape of solutions of first-order homogeneous equations}

Suppose we are given a difference field $\dfield{\FF}{\sigma}$ with $w\in\FF^*$. 
Then at various places in the article we will use the fact that one can predict the shape of the solution $g$ of a difference equation $\sigma(g)-w\,g=0$ if $g$ is from a certain class of $R\Pi$-extensions. More precisely, applying~\cite[Thm.~2.22]{DR1} with~\cite[Cor.~4.15]{DR1} yields
\begin{proposition}\label{Prop:StructureofProductSol}
Let $\dfield{\FF}{\sigma}$ be a difference field with constant field $\KK$, and let
$\dfield{\EE}{\sigma}$ be an \rpiE-extension of $\dfield{\FF}{\sigma}$ with $\EE=\FF\langle x_1\rangle\dots\langle x_r\rangle[z_1]\dots[z_l]$ where for $1\leq i\leq r$ the $x_i$ with $\frac{\sigma(x_i)}{x_i}\in\FF\langle x_1\rangle\dots\langle x_{i-1}\rangle^*$ are \piE-monomials and for $1\leq i\leq l$ the $z_i$ with $\frac{\sigma(z_i)}{z_i}\in\KK^*$ are \rE-monomials. If there are $g\in\EE\setminus\{0\}$ and $w\in\FF^*$ with $\sigma(g)=w\,g$, then 
$$g=h\,x_1^{m_1}\dots x_r^{m_r}z_1^{n_1}\dots z_l^{n_l}$$ 
with $m_1,\dots,m_r\in\ZZ$ and $n_1,\dots,n_l\in\NN$ with $0\leq n_i<\ord(z_i)$ for $1\leq i\leq l$.   
\end{proposition}

\subsection{Characterizations of $R\Pi$-extensions}

We start with the following
\begin{definition}[\cite{Karr:81}]\label{Def:MSet}
Let $\dfield{\FF}{\sigma}$ be a difference field. For $\alpha_1,\dots,\alpha_r\in\FF^*$,
$$M((\alpha_1,\dots,\alpha_r),\FF)=\{(n_1,\dots,n_r)\in\ZZ^r\mid\exists g\in\FF^*:\, \alpha_1^{n_1}\dots\alpha_r^{n_r}=\tfrac{\sigma(g)}{g}\}.$$
\end{definition}

\noindent Note that $M((\alpha_1,\dots,\alpha_r),\FF)$ is a submodule of $\ZZ^r$ over $\ZZ$. As a consequence, it is finitely generated and has a basis of rank $\leq r$. Due to~\cite[Thm.~3.2 and 3.5]{Schneider:05c} (based on~\cite{Karr:81,Ge:93}) such a basis can be calculated for \pisiE-fields.

\begin{proposition}\label{Prop:ComputeBasisForM}
Let $\dfield{\FF}{\sigma}$ be a \pisiE-field over $\KK$ where $\KK$ is built by a rational function field over an algebraic number field. Then a basis of $M((\alpha_1,\dots,\alpha_r),\FF)$ for $\alpha_1,\dots,\alpha_r\in\FF^*$ can be computed.
\end{proposition}

In the following we will elaborate a characterization of single-nested $R\Pi$-exten\-sions using the notion given in Definition~\ref{Def:MSet}.
For one \piE-monomial we start with
\begin{proposition}(\cite[Thm.~2.12.(2)]{DR1})\label{Prop:PiChar}
Let $\dfield{\AR\langle t\rangle}{\sigma}$ be a $P$-extension of a difference ring $\dfield{\AR}{\sigma}$ with $\sigma(t)=a\,t$. Then this is a \piE-extension (i.e., $\const{\AR\langle t\rangle}{\sigma}=\const{\AR}{\sigma}$) iff there are no $n\in\ZZ\setminus\{0\}$ and $g\in\AR\setminus\{0\}$ with $\sigma(g)=a^n\,t$. 
\end{proposition}

\noindent For nested \piE-extensions defined over a difference field we will need in addition

\begin{proposition}\label{Prop:PiEquivalences}(\cite[Lemma~5.1]{OS:18})
	Let $\dfield{\FF}{\sigma}$ be a difference field with $\alpha_1,\dots,\alpha_r\in\FF^*$. Let $\dfield{\EE}{\sigma}$ be the $P$-extension of $\dfield{\FF}{\sigma}$ with $\EE=\FF(x_{1})\dots(x_{r})$ and $\sigma(x_i)=\alpha_i\,x_i$ for $1\leq i \leq r$ and let\footnote{Note that the quotient field of $\EE$ is $\HH$. In particular, $\dfield{\HH}{\sigma}$ is a sub-difference ring of $\dfield{\EE}{\sigma}$.} $\dfield{\HH}{\sigma}$ be the $P$-ring extension of $\dfield{\FF}{\sigma}$ with $\HH=\FF\langle x_1\rangle\dots\langle x_r\rangle$ and $\sigma(x_i)=\alpha_i\,x_i$ for $1\leq i \leq r$. Then the following statements are equivalent. 
	\begin{enumerate} 
		\item $M((\alpha_1,\dots,\alpha_r),\FF)=\{\vect{0}\}$.
		\item $\const{\EE}{\sigma}=\dfield{\FF}{\sigma}$, i.e., $\dfield{\EE}{\sigma}$ is a \piE-field extension of $\dfield{\FF}{\sigma}$.
		\item $\const{\HH}{\sigma}=\dfield{\FF}{\sigma}$, i.e., $\dfield{\HH}{\sigma}$ is a \piE-extension of $\dfield{\FF}{\sigma}$.
	\end{enumerate}
\end{proposition}

\noindent Next, we will consider the case for several $R$-monomials. We start with the following simple observations.

\begin{lemma}\label{Lemma:NoRootElement}
Let $\dfield{\AR}{\sigma}$ be a constant-stable difference ring with constant field $\KK=\const{\AR}{\sigma}$. Then there are no $\gamma\in\AR\setminus\{0\}$ and no root of unity $w\in\KK^*$ with $\ord(w)>1$ and $\sigma(\gamma)=w\,\gamma$. 
\end{lemma}
\begin{proof}
Assume that there are such $\gamma$ and $w$ with $k=\ord(w)>1$. Then
$\sigma^k(\gamma)=\sigma^{k-1}(w\,\gamma)=w\sigma^{k-1}(\gamma)=\dots=w^k\,\gamma=\gamma$
and thus $\gamma\in\const{\AR}{\sigma^k}$. Since $\dfield{\AR}{\sigma}$ is constant-stable, $\gamma\in\KK^*$. Hence $\gamma=w\,\gamma$ and thus $w=1$, a contradiction.
\end{proof}

\begin{proposition}\label{Prop:generalRChar}
Let $\dfield{\AR[z]}{\sigma}$ be an $A$-extension of $\dfield{\AR}{\sigma}$ of order $n$ with $\sigma(z)=a\,z$ where $a\in\AR^*$. Then the following holds.
\begin{enumerate}
\item $z$ is an $R$-monomial (i.e., $\const{\AR[z]}{\sigma}=\const{\AR}{\sigma}$) iff there is no $g\in\AR^*$ and $\lambda\in\NN$ with $1\leq\lambda<n$ such that $\sigma(g)=a^{\lambda}\,g$ holds
\item If $z$ is an \rE-monomial, $a$ is a $\lambda$th primitive root of unity and $\ord(z^k)=\ord(a^k)$ for all $0\leq k\leq\lambda$.
\item If $\dfield{\AR}{\sigma}$ is constant-stable, $\KK=\const{\AR}{\sigma}$ is a field and $a\in\KK^*$ is a $\lambda$th primitive root of unity, then $z$ is an \rE-monomial.
\end{enumerate}
\end{proposition}
\begin{proof}
(1) This equivalence follows from~\cite[Thm.~2.12]{DR1}.\\ 
(2) In addition, if $z$ is an \rE-monomial, $a$ is a $\lambda$th primitive root of unity by~\cite[Thm.~2.12]{DR1}.
In particular, since $z^{\lambda}=1$ is the defining relation, we get $\ord(z)=\ord(a)=\lambda$. Trivially, we have $1=\ord(a^0)=\ord(z^0)$. 
Now suppose that $\ord(z^k)\neq\ord(\alpha^k)$ for some $1<k<\lambda$. If
$\ord(z^k)>\ord(\alpha^k)=:l$, then $\sigma(z^{k\,l})=\alpha^{k\,l}z^{k\,l}=z^{k\,l}$ with $z^{k\,l}=(z^{k})^l\neq1$ and thus $z^{k\,l}\in\const{\AR[z]}{\sigma}\setminus\AR$, a contradiction to the assumption that $z$ is an \rE-monomial.
Otherwise, if $l:=\ord(z^k)<\ord(\alpha^k)$, then $\alpha^{k\,l}=\alpha^{k\,l}\,z^{k\,l}=\sigma(z^{k\,l})=\sigma(1)=1$, i.e., $\ord(\alpha^k)\leq l$, a contradiction.\\
(4) Suppose that $\dfield{\AR}{\sigma}$ is constant-stable and $a\in\KK^*$ with $\ord(a)=\lambda$. Since $\ord(a^k)\neq1$ for $1\leq k<\lambda$, there are no $\gamma\in\AR\setminus\{0\}$ and $k\in\NN$ with $1\leq k<\lambda$ and $\sigma(\gamma)=a^k\,\gamma$ by Lemma~\ref{Lemma:NoRootElement}. Thus $z$ is an \rE-monomial by statement~(1).
\end{proof}

In particular, we will require the following result~\cite[Prop.~2.23]{DR3}.
\begin{proposition}\label{Prop:CharactSeveralRExt}
Let $\dfield{\FF}{\sigma}$ be a constant-stable
difference field with constant field $\KK$. Let $\dfield{\EE}{\sigma}$ with $\EE=\FF[z_1]\dots[z_l]$ be an $A$-extension of $\dfield{\FF}{\sigma}$ with $a_i=\frac{\sigma(z_i)}{z_i}\in\KK^*$ for $1\leq i\leq l$ of orders $\lambda_1,\dots,\lambda_l$, respectively. Then this is an \rE-extension iff for $1\leq i\leq l$, the $a_i$ are primitive $\lambda_i$-th roots of unity and $\gcd(\lambda_i,\lambda_j)=1$ for pairwise distinct $i,j$.
\end{proposition}

Finally, we are in the position to present a characterization that combines single-nested $R$-extensions and \piE-extensions.

\begin{proposition}\label{Prop:FromAlphaToRPi}
	Let $\dfield{\FF}{\sigma}$ be a constant-stable difference field with constant field $\KK$. Furthermore let $\dfield{\EE}{\sigma}$ be $AP$-extension of $\dfield{\FF}{\sigma}$ with $\FF(x_1)\dots(x_r)[z_1]\dots[z_l]$ where the $x_i$ are $P$-monomials with $\alpha_i=\frac{\sigma(x_i)}{x_i}\in\FF^*$ for $1\leq i \leq r$ and the $z_i$ are $A$-monomials with $a_i=\frac{\sigma(z_i)}{z_i}\in\KK^*$ of order $\lambda_i$ for $1\leq i \leq l$. Consider the sub-difference ring $\dfield{\HH}{\sigma}$ of $\dfield{\EE}{\sigma}$ with $\HH=\FF\langle x_1\rangle\dots\langle x_r\rangle[z_1]\dots[z_l]$ which forms an $AP$-extension of $\dfield{\FF}{\sigma}$.
	Then the following statements are equivalent. 
	\begin{enumerate} 
		\item $M((\alpha_1,\dots,\alpha_r),\FF)=\{\vect{0}\}$, the $a_i$ primitive roots of unity and $\gcd(\lambda_i,\lambda_j)=1$ for pairwise distinct $i,j$.
		\item $\const{\EE}{\sigma}=\KK$, i.e., $\dfield{\EE}{\sigma}$ is an \rpiE-extension of $\dfield{\FF}{\sigma}$.
		\item $\const{\HH}{\sigma}=\KK$, i.e., $\dfield{\HH}{\sigma}$ is an \rpiE-extension of $\dfield{\FF}{\sigma}$.
	\end{enumerate}
\end{proposition}
\begin{proof}
(1)$\Rightarrow$(2): By Proposition~\ref{Prop:PiEquivalences} $\dfield{\FF(x_1)\dots(x_r)}{\sigma}$ is a \piE-extension of $\dfield{\FF}{\sigma}$. In particular, by Proposition~\ref{Prop:ConstantStableFieldVersion} it is constant-stable. Thus Prop.~\ref{Prop:CharactSeveralRExt} is applicable which shows that $\dfield{\FF(x_1)\dots(x_r)[z_1]\dots[z_l]}{\sigma}$ is an \rE-extension of $\dfield{\FF(x_1)\dots(x_r)}{\sigma}$.\\
(2)$\Rightarrow$(1): This follows immediately by Propositions~\ref{Prop:PiEquivalences} and~\ref{Prop:CharactSeveralRExt}.\\
(2)$\Rightarrow$(3): This follows by $\KK\subseteq\const{\HH}{\sigma}\subseteq\const{\EE}{\sigma}$.\\
(3)$\Rightarrow$(2): Suppose that statement (3) holds, but (2) does not hold. With Proposition~\ref{Prop:PiEquivalences} it follows that $\dfield{\GG}{\sigma}$ with $\GG=\FF(x_1)\dots(x_r)$ is a \piE-field extension of $\dfield{\FF}{\sigma}$. Consequently, there is a $j$ with $1\leq j< l$ such that $\dfield{\GG[z_1]\dots[z_j]}{\sigma}$ is an \rE-extension of $\dfield{\GG}{\sigma}$ and $z_{j+1}$ is not an $R$-monomial. By Proposition~\ref{Prop:generalRChar} there are a $\nu_{j+1}\in\NN$ with $1\leq\nu_{j+1}<\lambda_{j+1}$ and $g\in\GG[z_1]\dots[z_j]\setminus\{0\}$ with $\sigma(g)=a_{j+1}^{\nu_{j+1}}\,g$. 
Furthermore, by Proposition~\ref{Prop:StructureofProductSol} (with $\FF=\GG$ and $r=0$) it follows that $g=h\,z_1^{\nu_1}\dots z_j^{\nu_j}$ for some $h\in\GG^*$. Thus $\sigma(h)=\gamma\,h$ where $\gamma=a_1^{-\nu_1}\dots a_j^{-\nu_j}a_{j+1}^{\nu_{j+1}}\in\KK^*$ is a root of unity, and hence $\gamma^{\lambda}=1$ for some $\lambda\in\NN\setminus\{0\}$. This implies that $\sigma^{\lambda}(h)=\gamma^{\lambda}\,h=1$. Since $\dfield{\GG}{\sigma}$ is constant-stable by Proposition~\ref{Prop:ConstantStableFieldVersion}, $h\in\const{\GG}{\sigma}=\KK$. Consequently $g\in\FF[z_1]\dots[z_j]\setminus\{0\}$ with $\sigma(g)=a_{j+1}^{\nu_{j+1}}\,g$. However, $\KK\subseteq\const{\FF[z_1]\dots[z_l]}{\sigma}\subseteq\const{\HH}{\sigma}=\KK$.
Hence
$\dfield{\FF[z_1]\dots[z_j][z_{j+1}]}{\sigma}$ is an \rE-extension of $\dfield{\FF[z_1]\dots[z_j]}{\sigma}$, a contradiction to part~(1) of Proposition~\ref{Prop:generalRChar}.
\end{proof}

\subsection{Radical stability}

In order to solve Problem~\DR\ as described in Section~\ref{Subsec:ProblemDescription}, the following properties are crucial.

\begin{definition}\label{Def:RacialStable}
Let $\dfield{\FF}{\sigma}$ be a difference field with $\KK=\const{\FF}{\sigma}$. Following~\cite{Karr:81} we define the \notion{homogeneous group} by 
$$H_{\dfield{\FF}{\sigma}}:=\{\frac{\sigma(g)}{g}\mid g\in\FF^*\}.$$
We call $\dfield{\FF}{\sigma}$ \notion{radical-stable} if for any $a\in\FF^*$ and $m\in\NN\setminus\{0\}$ with $a^m\in H_{\dfield{\FF}{\sigma}}$ there is a $\rho\in\KK^*$ with 
$\rho^m=1$ such that $$a\,\rho\in H_{\dfield{\FF}{\sigma}}$$ 
holds. $\dfield{\FF}{\sigma}$ is called \notion{radical-solvable} if it is radical-stable and for given $a\in\FF^*$ and $m\in\NN\setminus\{0\}$ with
$a^m\in H_{\dfield{\FF}{\sigma}}$ one can compute $\rho\in\KK^*$ with $\rho^m=1$ and $g'\in\FF^*$ such that $\sigma(g')=\gamma\,a\,g'$ holds.
\end{definition}

\noindent First we show that one can reduce the property of being radical-solvable to the property of being radical-stable and solving first-order homogeneous equations.

\begin{lemma}\label{Lemma:RadicalStableImpliesRadicalComputable}
If a difference field $\dfield{\FF}{\sigma}$ is radical stable and one can solve first-order homogeneous linear difference equations in $\dfield{\FF}{\sigma}$, then it is radical-solvable.
\end{lemma}
\begin{proof}
Let $\dfield{\FF}{\sigma}$ be radical-stable and take $a\in\FF^*$ and $m\in\NN\setminus\{0\}$ with $a^m\in H_{\dfield{\FF}{\sigma}}$.
Then we can conclude that there are a $\rho\in\KK^*$ with $\rho^m=1$ and a $\hat{g}\in\FF^*$ with $\frac{\sigma(\hat{g})}{\hat{g}}=\rho\,a$.
Take a primitive $m$th root of unity $\lambda\in\KK^*$ and loop through $i=0,1,2,\dots,m-1$ until one finds a $\hat{g}\in\FF^*$ with $\frac{\sigma(\hat{g})}{\hat{g}}=\lambda^i\,a$. In this way one considers all\footnote{If $m$ is minimal such that $a^m\in H_{\dfield{\FF}{\sigma}}$ holds, we only have to consider all $i$ with $\gcd(i,m)=1$, i.e., all cases where $\lambda^i$ is again a primitive $m$th root of unity.} $m$th roots of unity, in particular $\rho=\lambda^i$ will arise for some $i$.
\end{proof}


In Subsection~\ref{Subsec:constantstable} we succeeded in lifting the property of being constant-stable from a smaller field to a larger field. The main goal of this subsection is to obtain a similar result for the property of being radical-stable. For \sigmaE-monomials this lifting process will work in full generality (see Prop.~\ref{Prop:GeneralSigmaLifting}). To obtain such a result for \piE-monomials, we have to require further properties on $\dfield{\FF}{\sigma}$ (see Prop.~\ref{Prop:SpecialPiLifting}). This finally enables us to show in Corollary~\ref{Cor:RadicalComputable} that the mixed-rational difference field is radical-stable. We expect that this lifting-machinery can be applied also for other types of \pisiE-field extensions to show that they are radical-stable.

In this subsection we will use the following convention. Let a multivariate rational function $\gamma=\frac{p}{q}\in\FF(t_1,\dots,t_e)\setminus\{0\}$ be in reduced representation, i.e., $p,q\in\FF[t_1,\dots,t_e]\setminus\{0\}$ are polynomials that are co-prime. Then for an irreducible polynomial $h\in\FF[t_1,\dots,t_e]$ we write $h\nmid\gamma$ if $h\nmid p$ and $h\nmid q$ holds.

\noindent The following lemma elaborates the main complication of our desired lifting process.
\begin{lemma}\label{Lemma:RadicalLift}
Let $\dfield{\FF(t)}{\sigma}$ be a \pisiE-field extension of $\dfield{\FF}{\sigma}$ and let $a,g\in\FF(t)^*$ and $m\in\NN\setminus\{0\}$ with $\frac{\sigma(g)}{g}=a^m$. Then there is a $\gamma\in\FF(t)^*$ and $u\in\FF^*$ with 
\begin{equation}\label{Equ:MainProblemForRadicalStable}
g=\gamma^m\,u\,t^n
\end{equation}
for some $n\in\ZZ$. If $t$ is a \sigmaE-monomial, $n=0$; if $t$ is a \piE-monomial, $t\nmid\gamma$. 
\end{lemma}
\begin{proof}
If $g\in\FF^*$, we can set $\gamma=1$, $n=0$ and $u=g$. Otherwise suppose that $g\notin\FF$. 
By~\cite[Thm.~7]{Karr:81} (see also~\cite[Sec.~2.3]{Schneider:05c}) we can write $g$ in the following form\footnote{In~\cite{Karr:81} this representation is also called $\sigma$-representation. Its existence can be derived from the statements~(1) and~(2) of Lemma~\ref{Lemma:ProofForConstantStable}.}:
$g=u\,t^n\,g_1\cdots g_k$ where $n\in\ZZ$, $u\in\FF^*$,
$\gcd(g_i,\sigma^l(g_j))=1$ for all $i\neq j$ and $l\in\ZZ$, and for $1\leq i\leq k$, 
\begin{equation}\label{Equ:SigmaFacForHyper}
g_i=\prod_{j=0}^{r_i}\sigma^j(h_i)^{m_{ij}}\neq1
\end{equation}
where $h_i\in\FF[t]\setminus\FF$ are irreducible, 
$m_{ij}\in\ZZ$ and $r_i\geq0$. In particular, $n=0$ if $t$ is a \sigmaE-monomial, and $t\nmid g_i$ for all $i$ if $t$ is a \piE-monomial.
Then by $\frac{\sigma(g)}{g}=\alpha^m$ it follows
that for all $1\leq i\leq k$ and all $0\leq j\leq r_i-1$ we have
$m\mid (m_{i,j+1}-m_{i,j})\text{ and } m\mid m_{i,r_i}$.
Because of $m\mid m_{i,r_i}$ and $m\mid
(m_{i,r_i}-m_{i,r_{i}-1})$, it follows that $m\mid m_{i,r_{i}-1}$.
Applying this argument $r_i$ times proves that $m\mid m_{ij}$ for
all $1\leq i\leq k$ and all $1\leq j\leq r_i$. Hence $g_i=g_i'^m$
for $g_i':=\prod_{j=0}^{r_i}\sigma^j(h_i)^{m_{ij}/m}\in\FF(t)$ for
all $1\leq i\leq k$. With $\gamma=g'_1\dots g'_k$ we get $g=u\,\gamma^m\,t^n$ which completes the proof.
\end{proof}

More precisely, in~\eqref{Equ:MainProblemForRadicalStable} one obtains a solution which is close to derive radical-stability for $\dfield{\FF(t)}{\sigma}$: the only troublemaker is the possible factor $t^n$ with $n\in\ZZ$ if $t$ is a \piE-monomial. This property can be carried over for several \piE-extensions.

\begin{corollary}\label{Cor:MultiPigForm}
Let $\dfield{\EE}{\sigma}$ with $\EE=\FF(t_1)\dots(t_e)$ be a \piE-field extension of $\dfield{\FF}{\sigma}$ with $\alpha_i=\frac{\sigma(t_i)}{t_i}\in\FF^*$ for $1\leq i\leq e$. Let $a,g\in\EE^*$ and $m\in\NN\setminus\{0\}$ with $\frac{\sigma(g)}{g}=a^m$. Then there are $m_1,\dots,m_e\in\ZZ$, $\gamma\in\EE^*$ with $t_i\nmid\gamma$ for $1\leq i\leq e$ and $u\in\FF^*$ such that $g=u\,\gamma^m\,t_1^{m_1}\dots t_e^{m_e}$ holds.
\end{corollary}
\begin{proof}
For all $1\leq j\leq e$ we get the following construction: Moving $t_j$ for $1\leq j\leq e$ on top and applying Lemma~\ref{Lemma:RadicalLift} we get $u_j\in\EE_j=\FF(t_1)\dots (t_{j-1})(t_{j+1})\dots(t_e)^*$, $\gamma_j\in\EE_j(t_j)^*$, $t_j\nmid \gamma_j$ and $m_j\in\ZZ$ with 
\begin{equation}\label{Equ:gjStep}
g=u_j\,\gamma_j^m\,t_j^{m_j}.
\end{equation}
Write $\gamma_j=\frac{p_j}{q_j}$ with $p_j,q_j\in\FF[t_1,\dots,t_e]$ being coprime; we can assume that the polynomials $p_j$ and $q_j$ contain no polynomial factors in $\FF[t_1,\dots,t_e]$ that that are free of $t_j$ by moving them into $u_j$. Define $\gamma:=\frac{\lcm(p_1,\dots,p_e)}{\lcm(q_1,\dots,q_e)}\in\EE^*$ and 
\begin{equation}\label{Equ:Defineu}
u:=\frac{g}{\gamma^m t_1^{m_1}\dots t_e^{m_e}}\in\EE^*. 
\end{equation}
By construction, $g=u\,\gamma^m t_1^{m_1}\dots t_e^{m_e}$. The corollary follows if $u\in\FF$. 
Suppose that $u\notin\FF$. Then we can take an irreducible factor $f$ from $u$ such such that for some $1\leq j\leq e$ the variable $t_j$ depends on $f$.\\
\textbf{Case 1:} $f=t_j$. By construction $\gamma_j$ is free of the factor $t_j$ and by~\eqref{Equ:gjStep} also $g/t_j^{m_j}$ is free of the factor $t_j$. Furthermore $\gamma_i$ for all $i$ with $i\neq i$ is free of the factor $t_j$ (it is collected in the content $u_i\in\EE_i$). Thus also $t_j\nmid\gamma$. With~\eqref{Equ:Defineu}, $t_j$ cannot occur as a factor in $u$, a contradiction.\\ 
\textbf{Case 2:} $f\neq t_j$. Since $f$ occurs in $u$, it must occur in $g$ or in $\gamma$. Suppose that it does not occur in $g$. Then for $1\leq k\leq e$ it does not occur in $\gamma_k$ as a factor by~\eqref{Equ:gjStep} and thus it cannot occur in $\gamma$, a contradiction. Thus $f$ must occur in $g$. Let $n$ be maximal such that $f^n$ is a factor in $g$. Then $n$ is maximal such that $f^n$ occurs in $\gamma_j$ by~\eqref{Equ:gjStep}. 
In particular, $n$ is maximal such that $f^n$ occurs in $p_j^m$ or $q_j^m$, thus in $\lcm(p_1,\dots,p_e)^m$ or in $\lcm(q_1,\dots,q_e)^m$ and therefore in $\gamma^m$. In conclusion, $f$ cannot occur in $\frac{g}{\gamma^m}$ and thus not in~\eqref{Equ:Defineu}, again a contradiction.
\end{proof}

The following lemma is basic, but for completeness we state it here.

\begin{lemma}\label{Lemma:PiSiPowerProp}
Let $\dfield{\EE}{\sigma}$ be a $PS$-field extension of $\dfield{\FF}{\sigma}$, $a\in\EE^*$ and $m\in\NN\setminus\{0\}$. If $a^m\in\FF$ then $a\in\FF$. 
\end{lemma}
\begin{proof}
Let $\EE=\FF(t_1)\dots(t_e)$. Suppose that $a^m\in\FF$ and let $i$ be maximal such that $a\in\FF(t_1)\dots(t_i)\setminus\FF(t_1)\dots(t_{i-1})$. 
Write $a=\frac{p}{q}$ with $p,q\in\FF(t_1)\dots(t_{i-1})[t_i]$ where $\gcd(p,q)=1$ and either $p$ or $q$ depend on $t_i$. 
We get $p^m=a^m\,q^m$ with $a^m\in\FF^*$ and thus $\gcd(p^m,q^m)=p^m$ and $\gcd(p^m,q^m)=q^m$. 
However, $\gcd(p,q)=1$ implies that $\gcd(p^m,q^m)=1$ and thus 
$p^m\in\FF^*$ and $q^m\in\FF^*$, a contradiction that one of the polynomials $p$ or $q$ depend on $t_i$.
\end{proof}

Finally, we can give a recipe how one can lift the property of being radical-stable within a $PS$-extension. Here the crucial assumption is that the factor $t^n$ in~\eqref{Equ:MainProblemForRadicalStable} (and more generally, several such \piE-monomials) does not appear.

\begin{lemma}\label{Lemma:BasicListRadicalStable}
Let $\dfield{\EE}{\sigma}$ be a \pisiE-field extension of $\dfield{\FF}{\sigma}$ and let $m\in\NN\setminus\{0\}$.
Take
$g=u\,\gamma^m$ with $u\in\FF^*$ and $\gamma\in\EE^*$ such that $\frac{\sigma(g)}{g}=a^m$.
If $\dfield{\FF}{\sigma}$ is radical-stable, there are a $\tilde{\gamma}\in\EE^*$ and a $\rho\in(\const{\FF}{\sigma})^*$ with $\rho^m=1$ such that $\frac{\sigma(\tilde{\gamma})}{\tilde{\gamma}}=a\,\rho$ holds.  
\end{lemma}
\begin{proof}
Suppose that $\dfield{\FF}{\sigma}$ is radical-stable and let $g=u\,\gamma^m$ as claimed above.
Thus $a^m=\frac{\sigma(g)}{g}=\frac{\sigma(\gamma^m)}{\gamma^m}\,\frac{\sigma(u)}{u}$ and consequently
$$\frac{\sigma(u)}{u}=\left(a\,\frac{\gamma}{\sigma(\gamma)}\right)^m.$$
Since the left-hand side is in $\FF^*$, also the right-hand side is in $\FF^*$ and therefore also $a':=a\,\frac{\gamma}{\sigma(\gamma)}\in\FF^*$ by Lemma~\ref{Lemma:PiSiPowerProp}.
Since $\dfield{\FF}{\sigma}$ is radical-stable, it follows that there are a $u'\in\FF$ and a root of unity $\rho\in\const{\FF}{\sigma}$ with $\rho^m=1$ with $\frac{\sigma(u')}{u'}=\rho\,a'$. 
Therefore 
$\frac{\sigma(u')}{u'}=\rho\,a\,\frac{\gamma}{\sigma(\gamma)}$ and thus
$\frac{\sigma(\tilde{\gamma})}{\tilde{\gamma}}=a\,\rho$ with $\tilde{\gamma}=u'\,\gamma\in\EE^*$. 
\end{proof}

Using this lemma, we can derive our two main statements to lift the property of being radical-stable for \sigmaE-monomials and certain types of \piE-monomials.

\begin{proposition}\label{Prop:GeneralSigmaLifting}
Let $\dfield{\FF(t)}{\sigma}$ be a \sigmaE-field extension of $\dfield{\FF}{\sigma}$.
If $\dfield{\FF}{\sigma}$ is radical-stable, then $\dfield{\FF}{\sigma}$ is radical-stable. 
\end{proposition}
\begin{proof}
Suppose that $\dfield{\FF}{\sigma}$ is radical-stable.
Let $a,g\in\FF(t)^*$ and $m\in\NN\setminus\{0\}$ with $\frac{\sigma(g)}{g}=a^m$. By Lemma~\ref{Lemma:RadicalLift}
there are a $\gamma\in\FF(t)^*$ and $u\in\FF^*$ with $g=\gamma^m\,u$. 
Thus we can activate Lemma~\ref{Lemma:BasicListRadicalStable} with $\EE=\FF(t)$ and it follows that there is a $\tilde{\gamma}\in\FF(t)^*$ and a root of unity $\rho\in(\const{\FF}{\sigma})^*$ with $\rho^m=1$ 
such that $\frac{\sigma(\tilde{\gamma})}{\tilde{\gamma}}=a\,\rho$ holds.
Consequently, $\dfield{\FF(t)}{\sigma}$ is radical-stable.
\end{proof}

\begin{proposition}\label{Prop:SpecialPiLifting}
Let $\dfield{\EE}{\sigma}$ with $\EE=\FF(t_1)\dots(t_e)$ be a \piE-field extension of $\dfield{\FF}{\sigma}$ with $\alpha_i=\frac{\sigma(t_i)}{t_i}\in\FF^*$ for $1\leq i\leq e$. Suppose that $\dfield{\FF}{\sigma}$ is radical-stable and the following property holds:
\begin{multline}\label{Equ:PropertyForPiLift}
\forall m\in\NN\setminus\{0\}\,\forall (m_1,\dots,m_e)\in\ZZ^e\,\forall a,w\in\FF^*:\frac{\sigma(w)}{w}=a^m\alpha_1^{m_1}\dots\alpha_e^{m_e}\\
\Rightarrow m\mid m_1\wedge\dots\wedge m\mid m_e.
\end{multline}
Then $\dfield{\EE}{\sigma}$ is radical-stable.
\end{proposition}
\begin{proof}
Suppose that $\dfield{\FF}{\sigma}$ is radical-stable and that~\eqref{Equ:PropertyForPiLift} holds.
Let $a,g\in\EE^*$ and $m\in\NN\setminus\{0\}$ with $\frac{\sigma(g)}{g}=a^m$. 
By Corollary~\ref{Cor:MultiPigForm}
there are $m_1,\dots,m_e\in\ZZ$, $\gamma\in\EE^*$ with $t_i\nmid\gamma$ for $1\leq i\leq e$ and $u\in\FF^*$ such that $g=u\,\gamma^m\,t_1^{m_1}\dots t_e^{m_e}$ holds. Hence $\frac{\sigma(w)}{w}=a^m\alpha_1^{m_1}\dots\alpha_e^{m_e}$ with $a=\frac{\sigma(\gamma)}{\gamma}\in\EE^*$ and $w=u^{-1}\in\FF^*$. Since $\frac{\sigma(w)}{w}\in\FF^*$ and $\alpha_1^{m_1}\dots\alpha_e^{m_e}\in\FF^*$, $a^m\in\FF^*$. Thus by Lemma~\ref{Lemma:PiSiPowerProp} it follows $a\in\FF^*$.
By property~\eqref{Equ:PropertyForPiLift} we conclude that there are $n_i\in\ZZ$ with $m\,n_i=m_i$ for $1\leq i\leq e$. Thus $g=u\,\gamma'^m$ with
$\gamma'=\gamma\,t_1^{n_1}\dots t_e^{n_e}\in\EE$.
As in the proof of Proposition~\ref{Prop:GeneralSigmaLifting} we can apply Lemma~\ref{Lemma:BasicListRadicalStable} with $\EE=\FF(t_1)\dots(t_e)$ and it follows that there are a $\tilde{\gamma}\in\EE^*$ and a root of unity $\rho\in(\const{\FF}{\sigma})^*$ with $\rho^m=1$ such that $\frac{\sigma(\tilde{\gamma})}{\tilde{\gamma}}=a\,\rho$ holds.
Consequently, $\dfield{\EE}{\sigma}$ is radical-stable.
\end{proof}

\noindent Finally, we show that the mixed-rational difference field is radical-stable and radical-solvable.

\begin{corollary}\label{Cor:RadicalComputable}
Let $\KK=\KK'(q_1,\dots,q_v)$ be a rational function field with coefficients from a field $\KK'$.
The mixed-rational difference field $\dfield{\FF}{\sigma}$ with $\FF=\KK(x)(y_1)\dots(y_v)$ where $\const{\KK}{\sigma}=\KK$, $\sigma(x)=x+1$ and $\sigma(y_i)=q_i\,y_i$ for $1\leq i\leq v$ is radical-stable.
If $\KK'$ is a rational function field over an algebraic number field, one can solve homogeneous first-order linear difference equations in $\dfield{\FF}{\sigma}$ and $\dfield{\FF}{\sigma}$ is radical-solvable.
\end{corollary}
\begin{proof}
Note that the constant field $\dfield{\KK}{\sigma}$ is trivially radical-stable: If there are a $g\in\KK^*$ and $a\in\KK^*$ with $a^m=\frac{\sigma(g)}{g}=1$ for some $m\in\NN\setminus\{0\}$ then $a$ is a root of unity. Thus we can set $\rho=a^{-1}\in\KK^*$ with $\rho^m=1$ and $\gamma=1\in\KK^*$ and get 
$\frac{\sigma(\tilde{\gamma})}{\tilde{\gamma}}=a\,\rho$. Furthermore, property~\eqref{Equ:PropertyForPiLift} (with $\EE=\KK(y_1)\dots(y_v)$ and $\alpha_i=q_i$ for $1\leq i\leq v=e$) holds: 
Let $a,g\in\KK^*$, $m\in\NN\setminus\{0\}$ and $(m_1,\dots,m_v)\in\ZZ^v$ such that
$1=\frac{\sigma(g)}{g}=a^m q_1^{m_1}\dots q_v^{m_v}$
holds. Suppose that $m_i\neq0$ for some $1\leq i\leq v$. Then $q_i$ must occur as factor in $a$, say with multiplicity $l_i\in\ZZ\setminus\{0\}$. More precisely, we must have $m_i=l_i\,m$, i.e., $m\mid m_i$. Since property~\eqref{Equ:PropertyForPiLift} holds, we can apply Proposition~\ref{Prop:SpecialPiLifting} and it follows that $\dfield{\KK(y_1)\dots(y_v)}{\sigma}$ is radical-stable. By Proposition~\ref{Prop:GeneralSigmaLifting} we conclude that also $\dfield{\KK(y_1)\dots(y_v)(x)}{\sigma}$ and thus by reordering of the generators also $\dfield{\KK(x)(y_1)\dots(y_v)}{\sigma}$ is radical-stable. In particular, if $\KK'$ is a rational function field over an algebraic number field, one can solve linear first-order homogeneous difference equations in $\dfield{\FF}{\sigma}$ by\footnote{The proof is based on Karr's summation algorithm~\cite{Karr:81} and Ge's algorithm~\cite{Ge:93}. For the rational and $q$-rational case we refer also to~\cite{Abramov:89a}.}~\cite[Thm.~3.2 and 3.5]{Schneider:05c}. As a consequence, 
$\dfield{\FF}{\sigma}$ is also radical-solvable by Lemma~\ref{Lemma:RadicalStableImpliesRadicalComputable}.
\end{proof}

As a consequence also the rational difference ring (i.e., $v=0$), the $q$-rational difference ring (i.e., $v=1$ and $\FF=\KK(y_1)$), and the multi-basic difference ring (i.e., $\FF=\KK(y_1)\dots(y_v)$) are radical-stable.
We remark that the property of being radical-stable has been shown already earlier for the rational case in~\cite[Lemma~5.3]{Schneider:05c} and for the $q$-rational case in~\cite[Lemma~5.5 and Lemma~5.6]{Schneider:05c}.


\begin{example}\label{Exp:Computeg2}
Consider the rational difference field $\dfield{\FF}{\sigma}$ with $\FF=\KK(x)$, $\sigma(x)=x+1$ and 
$\alpha_2$ given in~\eqref{Equ:Exp:DefineSpecialAlpha}. It will turn out that $m=2$ is minimal such that  
$\alpha_2^m\in H_{\dfield{\FF}{\sigma}}$ holds. Since $\dfield{\FF}{\sigma}$ is radical-stable by Corollary~\ref{Cor:RadicalComputable}, 
there is a $\rho\in\KK^*$ with $\rho^2=1$ and a $\bar{g}_2\in\FF$ such that $\sigma(\bar{g}_2)=\rho\,\alpha_2\,\bar{g}_2$ holds. We can calculate $\bar{g}_2$ and $\rho$ following the proof of Lemma~\ref{Lemma:RadicalStableImpliesRadicalComputable}.
Take the primitive 2nd root of unity $\rho=-1$.
Then we check for $i=0,1$ if there is a $\bar{g}_2\in\FF^*$ with $\sigma(\bar{g}_2)=(-1)^i\alpha_2\bar{g}_2$. Since $m=2$ is minimal with $\alpha_2^m\in H_{\dfield{\FF}{\sigma}}$, it suffices to look at $i=1$; see the footnote in the proof of Corollary~\ref{Cor:RadicalComputable}). Solving this first-order homogeneous difference equation we obtain
\begin{equation}\label{Equ:g'2Explicit}
\bar{g}_2=(1 + x)^2 (2 + x)^5 (3 + x)^8 (4 + x) (5 + x).
\end{equation}
\end{example}

\subsection{Properties of the mixed-rational difference field (Theorem~\ref{Thm:MixedIsGood})}\label{Sec:MixedRatDF}
Collecting results of the previous subsections yields a

\medskip

\noindent\textbf{Proof of Theorem~\ref{Thm:MixedIsGood}.}
Let $\dfield{\FF}{\sigma}$ be a mixed-rational difference field with constant field $\KK$ where $\KK$
is a rational function field over an algebraic number field.
Note that $\dfield{\FF}{\sigma}$ is a \pisiE-field over $\KK$.
\begin{enumerate}
\item Since the standard operations in a rational function field defined over $\KK$ are computable, $\dfield{\FF}{\sigma}$ is computable.

\item As elaborated in Example~\ref{Exp:MixedDF} there is a $o^{\sigma}$-computable $\KK$-embedding into the ring of sequences.

\item By Corollary~\ref{Cor:RadicalComputable} the difference field $\dfield{\FF}{\sigma}$ is radical-stable.

\item By Corollary~\ref{Cor:RadicalComputable} one can solve linear first-order homogeneous difference equations in $\dfield{\FF}{\sigma}$. 

\item Finally, by Proposition~\ref{Prop:ComputeBasisForM} one can compute a basis of $M((\alpha_1,\dots,\alpha_r),\FF)$ with $\alpha_1,\dots,\alpha_r\in\FF^*$.\qed
\end{enumerate}

We remark that the mixed-rational difference field $\dfield{\FF}{\sigma}$ introduced in Example~\ref{Exp:MixedDF} is also constant-stable. This follows either by Corollary~\ref{Cor:PisiFieldIsConstantStable} ($\dfield{\FF}{\sigma}$ is a \pisiE-field) or by Lemma~\ref{Lemma:TauImpliesConstantStable} ($\dfield{\FF}{\sigma}$ can be embedded into the ring of sequences).

\section{The solution of Problem~\DR\ for a special case}\label{Sec:SpecialCase}

In the following we will consider a $P$-extension $\dfield{\FF\langle x_1\rangle\dots\langle x_r\rangle}{\sigma}$ of a difference field $\dfield{\FF}{\sigma}$ with $\alpha_i=\frac{\sigma(x_i)}{x_i}\in\FF^*$ for $1\leq i \leq r$ with the following special property: the $\ZZ$-submodule $M((\alpha_1,\dots,\alpha_r),\FF)$ of $\ZZ^r$ has rank $u\geq1$ and there is a basis of the form
\begin{equation}\label{Equ:SmithBasisforM}
\{(d_1,0,\dots,0),(0,d_2,0,\dots,0),\dots,(0,\dots,0,d_u,0,\dots,0)\}
\end{equation}
with $d_1\mid\dots\mid d_u$.

\begin{example}\label{Exp:SpecialCase1}
As running example we will start with the difference field $\dfield{\FF}{\sigma}$ with $\FF=\KK(x)$ where $\KK=\QQ(\iota)$ and $\sigma(x)=x+1$ equipped with $\fct{\bar{\ev}}{\FF\times\NN}{\KK}$ defined in Example~\ref{Exp:RationalDField1} ($\ev$ replaced by $\bar{\ev}$). Furthermore, we consider the $P$-extension
$\dfield{\HH}{\sigma}$ of $\dfield{\FF}{\sigma}$ with $\HH=\FF\langle x_1\rangle\langle x_2\rangle\langle x_3\rangle\langle x_4\rangle$ where $\sigma(x_i)=\alpha_i\,x_i$ for $1\leq i\leq 4$ with
\begin{align}\label{Equ:Exp:DefineSpecialAlpha}
\alpha_1&=\tfrac{(x+6)^2}{(x+4)^2},&
\alpha_2&=-\tfrac{(x+4)^7
   (x+6)}{(x+1)^2 (x+2)^3 (x+3)^3},\\
\alpha_3&=-\tfrac{\iota
   (x+4)^6}{9 (x+1) (x+2)^2 (x+3)
   (x+6)},&
\alpha_4&=-\tfrac{162 (x+1) (x+3)}{x+6}.\nonumber
\end{align}
Further, we extend the evaluation function $\fct{\bar{\ev}}{\FF\times\NN}{\KK}$ to $\fct{\ev}{\HH\times\NN}{\KK}$ by
\begin{equation}\label{Equ:Exp:DefineSpecialEv}
\begin{aligned}
\ev(x_1,n)&=\prod_{k=1}^n \tfrac{(5+k)^2}{(3+k)^2}&
\ev(x_2,n)&=\prod_{k=1}^n \tfrac{-(3+k)^7 (5+k)}{k^2 (1+k)^3 (2+k)^3},\\
\ev(x_3,n)&=\prod_{k=1}^n \tfrac{-\iota (3+k)^6}{9 k (1+k)^2 (2+k) (5+k)},&
\ev(x_4,n)&=\prod_{k=1}^n \tfrac{-162 k (2+k)}{5+k}.
\end{aligned}
\end{equation}
Using, e.g., the algorithm from~\cite{Karr:81} we obtain the basis 
\begin{equation}\label{Equ:MBasisSimple}
\{(1, 0, 0, 0), (0, 2, 0, 0)\}
\end{equation}
of $M((\alpha_1,\alpha_2,\alpha_3,\alpha_4),\FF)$, 
i.e., $u=2$ with $d_1=1$ and $d_2=2$.
\end{example}
For such an extension, we will solve Problem~\DR\ as described in Subsection~\ref{Subsec:ProblemDescription}. 
In order to derive this result, we will first treat a more general situation in  Lemma~\ref{Lemma:LambdaConstruction} that does not require that there is a $\KK$-embedding $\fct{\tau}{\FF}{\seqK}$. Afterwards we will specialize this result to Theorem~\ref{Thm:LambdaTauConstruction} for a given $\KK$-embedding.

Before we can proceed with this construction, we will elaborate several lemmas. In their proofs we will use the following definition.
For a Laurent polynomial $f\in\AR[t,t^{-1}]\setminus\{0\}$ we define
$$\xdeg(f)=\deg(f)-\ord(f).$$
This means that for $f=\sum_{k=l}^rf_k\,t^k$ with $l\leq r$ where $f_l\neq0\neq f_r$ (i.e., $\ord(f):=l$ and $\deg(f):=r$) we have $\xdeg(f)=r-s$.

\begin{lemma}\label{Lemma:ConstructDRHomo}
Let $\dfield{\EE}{\sigma}$ with $\EE=\AR\langle x_1\rangle\dots\langle x_r\rangle$ be a $P$-extension of a difference ring $\dfield{\AR}{\sigma}$ with $\alpha_i=\frac{\sigma(x_i)}{x_i}\in\AR^*$ for $1\leq i\leq r$. Let $\dfield{\HH}{\sigma}$ be a difference ring extension of $\dfield{\AR}{\sigma}$ and take the ring homomorphism $\fct{\lambda}{\EE}{\HH}$ with $\lambda|_{\AR}=\id$ and $\lambda(x_i)=g_i$ for some $g_i\in\HH^*$ with $1\leq i\leq r$. Then $\lambda$ is a difference ring homomorphism iff $\alpha_i=\frac{\sigma(g_i)}{g_i}$ for $1\leq i\leq r$. 
\end{lemma}
\begin{proof}
Suppose that $\lambda$ is a difference ring homomorphism. Then for $1\leq i\leq r$ we have $\frac{\sigma(g_i)}{g_i}=\frac{\sigma(\lambda(x_i))}{\lambda(x_i)}=\lambda(\frac{\sigma(x_i)}{x_i})=\lambda(\alpha_i)=\alpha_i$. Conversely, if $\alpha_i=\frac{\sigma(g_i)}{g_i}$ for $1\leq i \leq r$, then
$\lambda(\sigma(x_i))=\lambda(\alpha_i\,x_i)=\alpha_i\,g_i=\sigma(g_i)=\sigma(\lambda(x_i))$ which implies that $\sigma(\lambda(f))=\lambda(\sigma(f)$ for all $f\in\EE$.
\end{proof}

\begin{lemma}\label{Lemma:PropertyInOneVar}
Let $\dfield{\FF\langle x\rangle}{\sigma}$ be a $P$-extension of a difference field $\dfield{\FF}{\sigma}$ with $\sigma(x)=\alpha\,x$ and $\KK=\const{\FF}{\sigma}$.
Let $\dfield{\HH}{\sigma}$ be a difference ring extension of $\dfield{\FF}{\sigma}$ with $\const{\HH}{\sigma}=\KK$ equipped with a difference ring homomorphism $\fct{\lambda}{\FF\langle x\rangle}{\HH}$ with 
$\lambda|_{\KK}=\id$. 
Then the following holds. 
\begin{enumerate}
\item $M((\alpha),\FF)\neq\{0\}$ if and only if $\ker(\lambda)\neq\{0\}$.
\item If $M((\alpha),\FF)=\langle m\rangle_{\ZZ}$ with $m>0$, then $\ker(\lambda)=\langle \mu\rangle_{\FF\langle x\rangle}$ 
with $\mu=x^m+g$ for some $g\in\FF^*$ with $\sigma(g)=\alpha^m\,g$.
\end{enumerate}
\end{lemma}
\begin{proof}
(1) Suppose that $M((\alpha),\FF)\neq\{0\}$. Then there is a $\gamma\in\FF^*$ with $\sigma(\gamma)=\alpha^m\,\gamma$ for some $m>0$ (if $m<0$, we can take $\gamma'=1/\gamma\in\FF^*$ with $\sigma(\gamma')=\alpha^{-m}\gamma'$). Hence $\sigma(\frac{x^m}{\gamma})=\frac{x^m}{\gamma}$, therefore $\sigma(\lambda(\frac{x^m}{\gamma}))=\lambda(\frac{x^m}{\gamma})$, and thus $\lambda(\frac{x^m}{\gamma})\in\const{\HH}{\sigma}=\const{\FF}{\sigma}=\KK$. 
Since $\frac{x^m}{\gamma}$ is a unit, $\lambda(\frac{x^m}{\gamma})\neq0$.
Hence $\lambda(\frac{x^m}{\gamma})=c$ for some $c\in\KK^*$ and consequently $\lambda(\frac{x^m}{\gamma}-c)=0$ or equivalently $\lambda(x^m-g)=0$ with $g=c\,\gamma\in\FF^*$; obviously we have that $\sigma(g)/g=\alpha^m$. 
Conversely, suppose that there is a $\mu\in\FF\langle x\rangle\setminus\{0\}$ with $\lambda(\mu)=0$. Take such a $\mu$ such that $n=\xdeg(\mu)$ is minimal. W.l.o.g. we may assume that $\mu=x^n+b$ with $b\in\FF[x]$ and $\deg(b)<n$ (otherwise we may take $\mu'=\text{lc}(\mu)^{-1}x^{-\ord(n)}\mu\in\FF[x]$ with $\lambda(\mu')=\lambda(\text{lc}(\mu)^{-1}x^{-\ord(n)})\lambda(\mu)=0$ and $\xdeg(\mu')=\deg(\mu')=m$). Define $h:=\sigma(\mu)-\alpha^n\,\mu$. By construction, $\xdeg(h)=\deg(h)<n$. Furthermore, since $\lambda(\sigma(\mu))=\sigma(\lambda(\mu))=\sigma(0)=0$, we get 
$\lambda(h)=\lambda(\sigma(\mu)-\alpha^n\,\mu)=\lambda(\sigma(\mu))-\lambda(\alpha^n)\lambda(\mu)=0.$
Because of the minimality of $\mu$, it follows that $h=0$ and thus
\begin{equation}\label{Equ:muRel}
\sigma(\mu)=\alpha^n\,\mu.
\end{equation}
Suppose that $b=0$. Then $0=\lambda(\mu)=\lambda(x)^n$. But $\lambda(x)\,\lambda(x^{-1})=1$ and therefore $\lambda(x)\neq0$, a contradiction. Thus $b=g_l\,x^l+d$ with $g_l\in\FF^*$, $0\leq l<n$ and $d\in\FF[x]$ with $\deg(d)<l$. By coefficient comparison w.r.t.\ $x^l$ in~\eqref{Equ:muRel} we conclude that $\sigma(g_l)=\alpha^{n-l}\,g_l$. Therefore $M((\alpha),\FF)\neq\{0\}$.\\
(2) Suppose that $M((\alpha),\FF)=\langle m\rangle_{\ZZ}$ with $m>0$ . By part~(1) there is a polynomial $h=x^m-g\in\FF[x]\setminus\FF$ with $\lambda(\mu)=0$ where $g\in\FF^*$ with $\sigma(g)=\alpha^m\,g$. Moreover, looking at the proof of part (1), among the $\mu\in\FF\langle x\rangle^*$ with $\lambda(\mu)=0$ where $n=\xdeg(\mu)$ is minimal, we can take 
$\mu=x^{n}+g_l\,x^l+d$ with $g_l\in\FF^*$ and $d\in\FF[t]$ with $\deg(d)<l$ where $\sigma(g_l)=\alpha^{n-l}\,g_l$. Because of $n\leq m$ and $M((\alpha),\FF)=\langle m\rangle_{\ZZ}$, it follows that $m=n$ and $l=0$. Consequently, $\mu=x^m+g_l$. 
Finally, we show that\footnote{Note that $\FF\langle x\rangle$ is a p.i.d.\ which implies this statement. For completeness we carry out the proof explicitly.} $\ker(\lambda)=\langle \mu\rangle_{\FF\langle x\rangle}$. $\ker(\lambda)\supseteq\langle\mu\rangle_{\FF\langle x\rangle}$ holds trivially. Let $f\in\ker(\lambda)$.  By polynomial reductions (polynomial divisions) in $\FF\langle x\rangle$ we remove all terms whose degrees are larger than $x^m$ or smaller than $0$. Thus we get $f=r+\mu\,g$ with $g\in\FF\langle x\rangle$ and $r\in\FF[x]$ where $\deg(r)<m$. Since $f,\mu\in\ker(\lambda)$, $r\in\ker(\lambda)$, and by the minimality of $\mu$ it follows that $r=0$. Thus $f=\mu\,g\in\langle\mu\rangle_{\FF\langle x\rangle}$ which completes the proof.
\end{proof}

\begin{lemma}\label{Lemma:LambdaConstruction}
Let $\dfield{\FF}{\sigma}$ be a radical-stable difference field with $\KK=\const{\FF}{\sigma}$, and 
let $\dfield{\FF\langle x_1\rangle\dots\langle x_r\rangle}{\sigma}$ be a $P$-extension of $\dfield{\FF}{\sigma}$ with $\alpha_i=\frac{\sigma(x_i)}{x_i}\in\FF^*$. 
Suppose that~\eqref{Equ:SmithBasisforM} with $u\geq1$
and $d_1\mid\dots\mid d_u$  is a basis of $M((\alpha_1,\dots,\alpha_r),\FF)$. If $d_u>0$, suppose in addition that $\dfield{\FF}{\sigma}$ is constant-stable.
Then the following holds:
\begin{enumerate}
 \item[(1)] $\dfield{\EE}{\sigma}$ with $\EE=\FF\langle x_{u+1}\rangle\dots\langle x_r\rangle$ is a \piE-extension of $\dfield{\FF}{\sigma}$ (i.e., $\const{\EE}{\sigma}=\const{\FF}{\sigma}=\KK$).
 \item[2a)] If $d_u=1$, there are $\bar{g}_i\in\FF^*$ with $\sigma(\bar{g}_i)=\alpha_i\,\bar{g}_i$ for $1\leq i\leq u$ and for any $c_1,\dots,c_u\in\KK^*$ the surjective ring homomorphism $\fct{\lambda}{\FF\langle x_1\rangle\dots\langle x_r\rangle}{\EE}$ defined by $\lambda|_{\EE}=\id$ and 
 \begin{equation}\label{Equ:MapXiToEE}
 \lambda(x_i)=c_i\,\bar{g}_i\quad 1\leq i\leq u
 \end{equation}
is a difference ring homomorphism. 
 \item[2b)] Otherwise, if $d_u>1$, there is an $R$-extension $\dfield{\EE[z]}{\sigma}$ of $\dfield{\EE}{\sigma}$ of order $d_u$ with $\frac{\sigma(z)}{z}=\rho$. In addition, one can take $\bar{g}_i\in\FF^*$ and $\nu_i\in\NN$ for $1\leq i\leq u$ with $\sigma(\bar{g}_i)=\rho^{\nu_i}\,\alpha_i\,\bar{g}_i$
 such that for any choice $c_1,\dots,c_u\in\KK^*$  
  the surjective ring homomorphism $\fct{\lambda}{\FF\langle x_1\rangle\dots\langle x_r\rangle}{\EE[z]}$ defined by $\lambda|_{\EE}=\id$ and
 \begin{equation}\label{Equ:LambdaMap}
 \lambda(x_i)=c_i\,z^{d_u-\nu_i}\bar{g}_i\quad 1\leq i\leq u
 \end{equation}
 is a difference ring homomorphism; further\footnote{Since $\ord(z^{d_u-\nu_i})=d_i$, it follows that $z^{d_u-\nu_i}=z^{n_i\frac{d_u}{d_i}}$ for some $n_i\in\NN$ with $1\leq n_i<d_i$ and $\gcd(n_i,d_i)=1$.}, $\ord(z^{\nu_i})=\ord(z^{d_u-\nu_i})=\ord(\rho^{\nu_i})=\ord(\rho^{d_u-\nu_i})=d_i$.
 \item[(3)] $\ker(\lambda)=\langle x_1^{d_1}-(c_1\,\bar{g}_1)^{d_1},\dots,x_u^{d_u}-(c_u\,\bar{g}_u)^{d_u}\rangle_{\FF\langle x_1\rangle\dots\langle x_r\rangle}$ where the $\bar{g}_i$ are given by (2a) or (2b), respectively.
\end{enumerate} 
If $\dfield{\FF}{\sigma}$ is in addition computable and one can solve first-order homogeneous difference equations in $\dfield{\FF}{\sigma}$, then
one can compute $\dfield{\EE}{\sigma}$, $\lambda$ is computable and the generators of $\ker(\lambda)$ can be given explicitly.
\end{lemma}

\begin{proof}
(1) Since $M((\alpha_{u+1},\dots,\alpha_r),\FF)=\{\vect{0}\}$, we can activate Proposition~\ref{Prop:PiEquivalences} and it follows that $\dfield{\EE}{\sigma}$ is a \piE-extension of $\dfield{\FF}{\sigma}$.\\
(2a) Suppose that $d_1=\dots=d_u=1$. 
Then for $1\leq i\leq u$ there are $\bar{g}_i\in\FF^*$ with $\sigma(\bar{g}_i)=\alpha_i\,\bar{g}_i$. Therefore we can define for any $c_i\in\KK^*$ with $u<i\leq r$ the ring homomorphism $\fct{\lambda}{\FF\langle x_1\rangle\dots\langle x_r\rangle}{\EE}$ with $\lambda|_{\EE}=\id$ and $\lambda(x_i)=c_i\,\bar{g}_i$ for $1\leq i\leq u$. 
Furthermore, $\lambda$ is a difference ring homomorphism by Lemma~\ref{Lemma:ConstructDRHomo}.\\
(2b) Suppose that $d_u>1$. For $1\leq i\leq u$ take $g_i\in\FF^*$ with $\sigma(g_i)=\alpha_i^{d_i}\,g_i$. 
Since $\dfield{\FF}{\sigma}$ is radical-stable, there are a $d_i$th root of unity $\rho_i\in\KK^*$ and $\bar{g}_i\in\FF^*$ with  
$\sigma(\bar{g}_i)=\rho_i\,\alpha_i\,\bar{g}_i$ for $1\leq i\leq u$. Note that
\begin{equation}\label{Equ:OrdRhoiEqdi}
\ord(\rho_i)=d_i.
\end{equation}
Otherwise we would get $\rho_i^s=1$ for some $s$ with $1<s<d_i$.
Therefore $\sigma(\bar{g}_i^s)=\rho_i^{s}\,\alpha_i^s\,\bar{g}_i^s=\alpha_i^s\,\bar{g}_i^s$. Since $s<d_i$, $(0,\dots,0,d_i,0,\dots,0)$ cannot be a basis element of $M((\alpha_1,\dots,\alpha_r),\FF)$, a contradiction. Thus the $\rho_i$ are primitive $d_i$th roots of unity for $1\leq i\leq u$.
Let $\rho=\rho_u$ and 
take the $A$-extension $\dfield{\EE[z]}{\sigma}$ of $\dfield{\EE}{\sigma}$ of order $d_u$ with $\sigma(z)=\rho\,z$.
Since $\dfield{\FF}{\sigma}$ is constant-stable, $\dfield{\EE}{\sigma}$ is constant-stable by Proposition~\ref{Prop:ConstantStableFieldVersion}.
Thus by part~(3) of Proposition~\ref{Prop:generalRChar} we conclude that $z$ is an $R$-monomial.
In particular, since $d_i\mid d_u$ for all $1\leq i\leq u$, there are
$\nu_i\in\{0,\dots,d_u-1\}$ with $\rho_i=\rho^{\nu_i}$. Consequently,
$\sigma(\bar{g}_i)=\rho^{\nu_i}\,\alpha_i\,\bar{g}_i$ for all $1\leq i\leq u$.
Define $g'_i=z^{d_u-\nu_i}\,\bar{g}_i$. Then
$$\sigma({g}_i')=\sigma(z^{d_u-\nu_i}\bar{g}_i)=\rho^{d_u-\nu_i}\rho^{\nu_i}\alpha_i\,z^{d_u-\nu_i}\bar{g}_i=\alpha_i\,{g}'_i.$$
Finally, define the ring homomorphism $\fct{\lambda}{\FF\langle x_1\rangle\dots\langle x_r\rangle}{\EE}$ with $\lambda|_{\EE}=\id$ and $\lambda(x_i)=c_i\,g'_i$ for 
$1\leq i\leq u$. Then by Lemma~\ref{Lemma:ConstructDRHomo} it follows that $\lambda$ is a difference ring homomorphism.
By~\eqref{Equ:OrdRhoiEqdi} we have $d_i=\ord(\rho^{\nu_i})=\ord(\rho^{d_u-\nu_i})$ and by part 2 of Proposition~\ref{Prop:generalRChar} we get $d_i=\ord(z^{\nu_i})=\ord(z^{d_u-\nu_i})$.\\
(3) Let $\EE_k=\FF\langle x_1\rangle\dots\langle x_k\rangle$ and denote by $I_k$ the difference ideal $I_k:=\ker(\lambda|_{\EE_k})$. Let $\dfield{\HH}{\sigma}$ be the $R\Pi$-extension of $\dfield{\FF}{\sigma}$ from case (2a) with $\HH=\EE$ or case (2b) with $\HH=\EE[z]$. In any case, 
\begin{equation}\label{Eq:HIsRP}
\const{\HH}{\sigma}=\KK.
\end{equation}
We show part~(3) by induction on the number of $P$-monomials in $\FF\langle x_1\rangle\dots\langle x_r\rangle$.
Since $\lambda|_{\EE_0}=\lambda|_{\FF}=\id$, $I_0=\{0\}=\langle\rangle$ and the statement holds. Now suppose that 
\begin{equation}\label{Equ:Ek-1IsGood}
I_{k-1}=\ker(\lambda|_{\EE_{k-1}})=\langle x_1^{d_1}-(c_1\,\bar{g}_1)^{d_1},\dots,x_{k-1}^{d_{k-1}}-(c_{k-1}\bar{g}_{k-1})^{d_{k-1}}\rangle_{\EE_{k-1}}
\end{equation}
where $\lambda(x_i^{d_i})=(c_i\bar{g}_i)^{d_i}$ for $1\leq i<k$.\\
\textbf{Case 1:} $k\leq u$. If $d_u=1$, we have $\lambda(x_k)=c_k\,\bar{g}_k$. Since $d_k\mid d_u$, $d_k=1$ and we get trivially $\lambda(x_k^{d_k})=(c_k\,\bar{g}_k)^{d_k}$.
Otherwise, if $d_u>1$, we have $\lambda(x_k)=c_k\,z^{d_u-\nu_k}\,\bar{g}_k$ where $(z^{d_u-\nu_k})^{d_k}=1$. 
Thus again $\lambda(x_k^{d_k})=(c_k\,\bar{g}_k)^{d_k}$.
Since $\lambda(c_k\,\bar{g}_k)=c_k\,\bar{g}_k$, we conclude for both cases that
$\lambda(x_k^{d_k}-(c_k\,\bar{g}_k)^{d_k})=0$. Consequently 
\begin{equation}\label{Equ:hinIk}
h:=x_k^{d_k}-(c_k\,\bar{g}_k)^{d_k}\in I_k.
\end{equation}
We will show that
\begin{equation}\label{Equ:Case1Ideal}
I_{k}=\ker(\lambda|_{\EE_{k}})=\langle x_1^{d_1}-(c_1\bar{g}_1)^{d_1},\dots,x_k^{d_{k-1}}-(c_{k-1}\bar{g}_{k-1})^{d_{k-1}},h\rangle_{\EE_k}.
\end{equation}
The inclusion $\supseteq$ holds trivially by~\eqref{Equ:hinIk}. For the inclusion $\subseteq$ two cases are considered. In what follows, $\xdeg$ is considered w.r.t.\ $x_k$\\
\textbf{Case 1.1:} Suppose that there is no $\mu\in I_k\setminus\langle I_{k-1}\rangle_{\EE_k}$ with $\xdeg(\mu)<d_k=\deg(h)$.
Now let $f\in I_k$ be arbitrary but fixed.
Then by polynomial division we can write
$f=a\,h+b$
where $a,b\in\EE_k$ with $\xdeg(b)<d_k$. Since $f,h\in I_k$, also $b\in I_k$. Thus 
$b\in\langle I_{k-1}\rangle_{\EE_k}$ by the assumption of Case~1.1. Thus~\eqref{Equ:Case1Ideal} holds with~\eqref{Equ:Ek-1IsGood}.\\
\textbf{Case 1.2:} Suppose that there is a $\mu\in I_k\setminus\langle I_{k-1}\rangle_{\EE_k}$
with $\xdeg(\mu)<d_k$. Note that $\xdeg(\mu)\neq0$, since otherwise $\mu=x^l\,c$ for some $l\in\ZZ$ and $c\in\FF^*$. Thus $0=\lambda(\mu)=\lambda(c)\,\lambda(x^l)$ with $\lambda(x^l)\neq0$ which would imply $\lambda(c)=0$. But then $\mu\in\langle I_{k-1}\rangle_{\EE_k}$ which we have excluded. 
Among all $\mu\in I_k\setminus\langle I_{k-1}\rangle_{\EE_k}$
with $0<\xdeg(\mu)<d_k$, we take one such that $\xdeg(\mu)$ is minimal. Write $\mu=\sum_{i}f_i\,x_k^i$ with $f_i\in\EE_{k-1}$. Now consider the Laurent polynomial $\mu'=\mu|_{x_1\mapsto\lambda(x_1),\dots, x_{k-1}\mapsto\lambda(x_{k-1})}$ in $x_k$. Note that $\mu'\neq0$ since otherwise $\lambda(f_i)=0$ for all $i$ and thus
$\mu\in\langle I_{k-1}\rangle_{\EE_k}$. In particular, $\mu'$ must depend on $x_k$, since $0=\lambda(\mu)=\lambda(\mu')=\mu'|_{x_k\mapsto\lambda(x_k)}$. \\
If $d_{k-1}=1$, $\lambda(x_1),\dots,\lambda(x_{k-1})\in\FF^*$ and thus $\mu'\in\FF\langle x_k\rangle$ with $0<\xdeg(\mu')<d_k$.\\
Otherwise, if $d_{k-1}>1$, $\mu'\in\FF[z]\langle x_k\rangle$ with $0<\xdeg(\mu')<d_k$. Write $\mu'=\sum_{i}\mu'_ix_k^i$ with $\mu'_i\in\FF[z]$ where we can take $r\neq0$ with $\mu'_r\neq0$. 
Now consider $\tilde{\mu}=\mu'|_{x_k\mapsto z^{d_u-\nu_k}\,x_k}$. 
Then $\tilde{\mu}=\sum_{i}\mu'_ix_k^i\,(z^{d_u-\nu_k})^i$. Since $(z^{d_u-\nu_k})^r\neq0$, $\mu'_r\,(z^{d_u-\nu_k})^r\neq0$ and consequently $\tilde{\mu}\neq0$. 
Now write $\tilde{\mu}=\sum_{i=0}^{d_u-1}\gamma_i\,z^i$ with $\gamma_i\in\FF\langle x_k\rangle$.
Note that 
$0=\lambda(\mu')=\tilde{\mu}|_{x_k\mapsto\bar{g}_k}$ and thus $\tilde{\mu}$ depends on $x_k$. In particular, 
$0<\xdeg(\tilde{\mu})<d_k$.
Furthermore, we can take an $l$ with $\gamma_l\in\FF\langle x_k\rangle$ and $0<\xdeg(\gamma_l)<d_k$. 
Since $\FF\langle x_k\rangle[z]$ is a $\FF\langle x_k\rangle$-module with basis $z^0,z^1,\dots,z^{d_u-1}$ and
$0=\tilde{\mu}|_{x_k\mapsto\bar{g}_k}=\sum_{i}(\gamma_i|_{x_k\mapsto\bar{g}_i})z^i$, it follows that $\lambda(\gamma_i)=\gamma_i|_{x_k\mapsto\bar{g}_k}=0$ for all $i$. In particular, this holds for $\gamma_l$. Hence, we get $\gamma_l\in\FF\langle x_k\rangle\setminus\FF$ with $0<\xdeg(\gamma_l)<$ and $\lambda(\gamma_l)=0$.
Summarizing there is a $\nu\in\FF\langle x_k\rangle$ (for the case $d_{k-1}=1$ we take $\nu=\mu'$ and for the case $d_{k-1}>1$, we take $\nu=\gamma_l$) such that $\lambda(\nu)=0$ and $0<\xdeg(\nu)<d_{k-1}$.
Note that $\dfield{\FF\langle x_k\rangle}{\sigma}$ is a $P$-extension of $\dfield{\FF}{\sigma}$ and $\lambda|_{\FF\langle x_k\rangle}$ is a difference ring homomorphism with $\lambda|_{\FF}=\id$. Furthermore, since~\eqref{Equ:SmithBasisforM} is a basis of $M((\alpha_1,\dots,\alpha_r),\FF)$, it follows that
$M((\alpha_k),\FF)=\langle d_k\rangle_{\ZZ}$.
Consequently, we may apply Lemma~\ref{Lemma:PropertyInOneVar} and conclude that $\ker(\lambda|_{\FF\langle x_k\rangle})=\langle m\rangle_{\FF\langle x_k\rangle}$ for some $m\in\FF[x_k]$ with $\xdeg(m)=\deg(m)=d_k$; a contradiction to the existence of $\nu$.\\
\textbf{Case 2:} $k>u$. Suppose that $f=\sum_{i}f_i\,x_k^i\in I_k$. As in the previous case define $f'=f|_{x_1\mapsto\lambda(x_1),\dots, x_{k-1}\mapsto\lambda(x_{k-1})}$. Note that $0=\lambda(f)=f'|_{x_k\mapsto x_k}=f'$. Thus $\lambda(f_i)=0$ for all $i$ and therefore $f\in\langle I_{k-1}\rangle_{\EE_{k}}$. This proves that $I_{k}=\langle I_{k-1}\rangle_{\EE_{k}}$ and thus $I_{k}=\langle x_1^{d_1}-(c_1\,\bar{g}_1)^{d_1},\dots,x_u^{d_u}-(c_u\,\bar{g}_u)^{d_u}\rangle_{\EE_k}$ by the induction assumption.\\
Finally suppose that $\dfield{\FF}{\sigma}$ is computable and one can solve homogeneous first-order difference equations in $\dfield{\FF}{\sigma}$. Since $\dfield{\FF}{\sigma}$ is radical-stable, it is also radical-solvable by Lemma~\ref{Lemma:RadicalStableImpliesRadicalComputable}. Thus for $1\leq i\leq u$ one can compute the $d_i$th roots of unity $\rho_i\in\KK^*$ and $g'_i\in\FF^*$ for $1\leq i\leq u$ with  
$\sigma(g'_i)=\rho_i\,\alpha_i\,g'_i$ which enables one to define $\lambda$ explicitly. In particular, the generators of $\ker(\lambda)$ can be calculated.
\end{proof}

\begin{example}[Cont. Ex.~\ref{Exp:SpecialCase1}]\label{Exp:lambdaDefSimple}
Recall that~\eqref{Equ:MBasisSimple} is a basis of $M((\alpha_1,\alpha_2,\alpha_3,\alpha_4),\FF)$ with~\eqref{Equ:Exp:DefineSpecialAlpha}.
Thus by part (1) of Lemma~\ref{Lemma:LambdaConstruction} we can construct the \piE-extension $\dfield{\EE}{\sigma}$ of $\dfield{\FF}{\sigma}$ with $\EE=\FF\langle x_3\rangle\langle x_4\rangle$ and $\sigma(x_3)=\alpha_3\,x_3$, $\sigma(x_4)=\alpha_3\,x_4$. In particular, by part (2b) of Lemma~\ref{Lemma:LambdaConstruction}, we can take the $R$-extension $\dfield{\EE[z]}{\sigma}$ of $\dfield{\EE}{\sigma}$ with $\sigma(z)=-z$.
Furthermore, we can compute the solution
$\bar{g}_1=(4 + x)^2 (5 + x)^2\in\KK(x)$ for $\sigma(g_1)=\alpha_1\,g_1$. In addition, in Example~\ref{Exp:Computeg2} 
we obtained $\bar{g}_2$ with~\eqref{Equ:g'2Explicit} such that $\sigma(\bar{g}_2)=-\alpha_2\,\bar{g}_2$ holds. 
With these $\bar{g}_i$ we can now define the ring homomorphism
$\fct{\lambda}{\HH}{\KK(x)\langle x_3\rangle\langle x_4\rangle[z]}$ with $\HH=\KK(x)\langle x_1\rangle\langle x_2\rangle\langle x_3\rangle\langle x_4\rangle$
given by $\lambda|_{\KK(x)}=\id$ and
\begin{equation}\label{Equ:LambdaDefSimple}
\begin{split}
\lambda(x_1)&=c_1\,\bar{g}_1=c_1(4 + x)^2 (5 + x)^2\\
\lambda(x_2)&=c_2\,\bar{g}_2\,z=c_2(1 + x)^2 (2 + x)^5 (3 + x)^8 (4 + x) (5 + x) z\\
\lambda(x_3)&=x_3\\
\lambda(x_4)&=x_4
\end{split}
\end{equation}
for any $c_1,c_2\in\KK^*$. By part~(2b) of Lemma~\ref{Lemma:LambdaConstruction} this forms a difference ring homomorphism. Finally, by part~(3) of Lemma~\ref{Lemma:LambdaConstruction} we obtain
\begin{multline*}
\ker(\lambda)=\langle x_1-(c_1\,\bar{g}_2)^1,x_2^2-(c_2\,\bar{g}_2)^2\rangle_{\HH}
=\langle x_1-c_1(4 + x)^2 (5 + x)^2,\\x_2^2-c_2^2(1 + x)^4 (2 + x)^{10} (3 + x)^{16} (4 + x)^2 (5 + x)^2\rangle_{\HH}.
\end{multline*}
\end{example}

Specialize Lemma~\ref{Lemma:LambdaConstruction} to Theorem~\ref{Thm:LambdaTauConstruction} yields the following  solution of Problem~\DR.

\begin{theorem}\label{Thm:LambdaTauConstruction}
Let $\dfield{\FF}{\sigma}$ be a radical-stable difference field with $\KK=\const{\FF}{\sigma}$ equipped with a $\KK$-embedding $\fct{\bar{\tau}}{\FF}{\seqK}$.
Let $\dfield{\FF\langle x_1\rangle\dots\langle x_r\rangle}{\sigma}$ be a $P$-extension of $\dfield{\FF}{\sigma}$ with $\alpha_i=\frac{\sigma(x_i)}{x_i}\in\FF^*$ and let $\fct{\tau}{\FF\langle x_1\rangle\dots\langle x_r\rangle}{\seqK}$ be a $\KK$-homomor\-phism with $\tau|_{\FF}=\bar{\tau}$. Suppose that~\eqref{Equ:SmithBasisforM}
with $u\geq1$ and $d_1\mid\dots\mid d_u$  is a basis of $M((\alpha_1,\dots,\alpha_r),\FF)\neq\{\vect{0}\}$. 
Then the following holds.
\begin{enumerate}
\item[(1)]  
$\tau|_{\EE}$ with $\EE=\FF\langle x_{u+1}\rangle\dots\langle x_r\rangle$ is a $\KK$-embedding where $\dfield{\EE}{\sigma}$ is a \piE-extension of $\dfield{\FF}{\sigma}$. 
\item[(2a)] If $d_r=1$, there are $\bar{g}_i\in\FF^*$ with $\sigma(\bar{g}_i)=\alpha_i\,\bar{g}_i$ for $1\leq i\leq u$. Furthermore, one can refine $\lambda$ from (2a) of Lemma~\ref{Lemma:LambdaConstruction} by appropriate $c_1,\dots,c_u\in\KK^*$ such that
 $\tau|_{\EE}(\lambda(f))=\tau(f)$ for all $f\in\FF\langle x_1\rangle\dots\langle x_r\rangle$ holds, i.e., the following diagram commutes:
 \begin{equation}\label{Equ:CommutingDiagramNoRExt}
 \xymatrix@!R=0.7cm@C1.8cm{
 \FF\langle x_1\rangle\dots\langle x_r\rangle\ar@{>>}[r]^{\lambda}\ar[dr]_{\tau} &\EE\ar@{^{(}->}[d]_{\tau|_{\EE}}\\
 &\seqK.}
 \end{equation} 
 \item[(2b)] Otherwise, if $d_u>1$, one can take the $R$-extension $\dfield{\EE[z]}{\sigma}$ of $\dfield{\EE}{\sigma}$ with order $d_u$ from part (2b) of Lemma~\ref{Lemma:LambdaConstruction} and the $\KK$-homomorphism $\fct{\tau'}{\EE[z]}{\seqK}$ with $\tau'|_{\EE}=\tau|_{\EE}$ and $\tau'(z)=\langle \rho^i\rangle_{i\geq0}$ which forms a $\KK$-embedding.
  Furthermore, one can take $\bar{g}_i\in\FF^*$ and $\nu_i\in\NN$ for $1\leq i\leq u$ with $\sigma(\bar{g}_i)=\rho^{\nu_i}\,\alpha_i\,\bar{g}_i$, and one can refine $\lambda$ from (2b) of Lemma~\ref{Lemma:LambdaConstruction} by appropriate $c_1,\dots,c_u\in\KK^*$ such that 
 $\tau'(\lambda(f))=\tau(f)$ for all $f\in\FF\langle x_1\rangle\dots\langle x_r\rangle$ holds, i.e., the following diagram commutes:
 \begin{equation}\label{Equ:CommutingDiagramWithRExt}
 \xymatrix@!R=0.7cm@C1.8cm{
 \FF\langle x_1\rangle\dots\langle x_r\rangle\ar@{>>}[r]^{\lambda}\ar[dr]_{\tau} &\EE[z]\ar@{^{(}->}[d]_{\tau'}\\
 &\seqK.}
 \end{equation} 
 \item[(3)] 
 $\ker(\tau)=\ker(\lambda)=\langle x_1^{d_1}-(c_1\,\bar{g}_1)^{d_1},\dots,x_u^{d_u}-(c_u\,\bar{g}_u)^{d_u}\rangle_{\FF\langle x_1\rangle\dots\langle x_r\rangle}$ where the $c_i$ and $\bar{g}_i$ are given by (2a) or (2b), respectively.
\end{enumerate}
If\footnote{Note that all the requirements of Assumption~\ref{Assum:AlgProp} except item~(5) are needed.} $\dfield{\FF}{\sigma}$ is in addition computable and one can solve first-order homogeneous difference equations in $\dfield{\FF}{\sigma}$, then one can compute $\dfield{\EE}{\sigma}$, $\lambda$ is computable and the generators of $\ker(\lambda)=\ker(\tau)$ can be given explicitly.
\end{theorem}

\begin{proof}
(1) Since $\dfield{\EE}{\sigma}$ is a \piE-extension of $\dfield{\FF}{\sigma}$, $\tau|_{\EE}$ is injective by Theorem~\ref{Thm:RPSImpliesInjective}.\\
(2a) Suppose that $d_1=\dots=d_u=1$. 
For $1\leq i\leq u$ we have $\bar{g}_i\in\FF^*$ with $\sigma(\bar{g}_i)=\alpha_i\,\bar{g}_i$. Thus $\sigma(\frac{x_i}{\bar{g}_i})=\frac{x_i}{\bar{g}_i}$ and consequently $S(\tau(\frac{x_i}{\bar{g}_i}))=\tau(\sigma(\frac{x_i}{\bar{g}_i}))=\tau(\frac{x_i}{\bar{g}_i})$. This implies that $\tau(\frac{x_i}{\bar{g}_i})$ is a constant sequence, i.e., $\tau(\frac{x_i}{\bar{g}_i})=\tau(c_i)$ for some $c_i\in\KK$.  In particular, $\tau(x_i)=\tau(c_i\,\bar{g}_i)$. By part~(1) of Lemma~\ref{Lemma:ZFunctionForField}, $\tau(x_i)\neq\vect{0}$ and thus $c_i\neq0$.
Therefore we can refine the difference ring homomorphism $\fct{\lambda}{\FF\langle x_1\rangle\dots\langle x_r\rangle}{\EE}$ with $\lambda|_{\EE}=\id$ and~\eqref{Equ:MapXiToEE} with the particular choice of $c_i$ given above and get $\tau(\lambda(f))=\tau(f)$ for all $f\in\FF\langle x_1\rangle\dots\langle x_r\rangle$.\\
(2b) Suppose that $d_u>1$ and consider the $R\Pi$-extension $\dfield{\EE[z]}{\sigma}$ of $\dfield{\FF}{\sigma}$. By Theorem~\ref{Thm:RPSImpliesInjective} it follows that the $\KK$-homomorphism $\fct{\tau'}{\EE[z]}{\seqK}$ with $\tau'|_{\EE}=\tau|_{\EE}$ and $\tau'(z)=\langle \rho^i\rangle_{i\geq0}$ is a injective.
As in (2a) we can conclude that there are $c_i\in\KK^*$ for $1\leq i\leq u$ such that $\tau(x_i)=\tau'(c_i\,z^{d_u-\nu_i}\,\bar{g}_i)$. Therefore we can refine as in case (2a) the difference ring homomorphism $\fct{\lambda}{\FF\langle x_1\rangle\dots\langle x_r\rangle}{\EE}$ with $\lambda|_{\EE}=\id$ and $\lambda(x_i)=c_i\,\bar{g}_iz^{d_u-\nu_i}$ for $1\leq i\leq u$ and get $\tau'(\lambda(f))=\tau(f)$ for all $f\in\FF\langle x_1\rangle\dots\langle x_r\rangle$.\\
(3) By Lemma~\ref{Lemma:ConnectLemmaToTu} we conclude that $\ker(\lambda)=\ker(\tau)$. Thus with part (3) of Lemma~\ref{Lemma:LambdaConstruction} the statement is proven.\\
Suppose that one can solve first-order homogeneous linear difference equations in $\dfield{\FF}{\sigma}$. Then by Lemma~\ref{Lemma:LambdaConstruction} $\lambda$ can be given with generic $c_i\in\KK^*$ and the generators of $\ker(\lambda)$ can be computed. Since $\tau$ is $o^{\sigma}$-computable, one can determine the $c_i$ for $1\leq i\leq u$ with $\tau(x_i)=\tau(c_i\,\bar{g}_i)$ in case~2a or $\tau(x_i)=\tau'(c_i\,\bar{g}_iz^{d_u-\nu_i})$ in case~2b.
\end{proof}

\begin{example}[Cont.~Ex.~\ref{Exp:lambdaDefSimple}]\label{Exp:SpecialCaseTauExp}
As elaborated in Example~\ref{Exp:lambdaDefSimple} we can take the $R\Pi$-extension $\dfield{\EE[z]}{\sigma}$ of $\dfield{\FF}{\sigma}$ with $\EE=\FF\langle x_3\rangle\langle x_4\rangle$ and $\sigma(x_3)=\alpha_3\,x_3$, $\sigma(x_4)=\alpha_3\,x_4$ where the $\alpha_3,\alpha_4$ are given in~\eqref{Equ:Exp:DefineSpecialAlpha} and $\sigma(z)=-z$. Now take the evaluation function $\fct{\ev'}{\EE[z]\times\NN}{\KK}$ for $\dfield{\EE[z]}{\sigma}$ with
$\ev'_{\EE\times\NN}=\ev|_{\EE\times\NN}$ and $\ev'(z,n)=(-1)^n$  where $\fct{\ev}{\FF\langle x_1\rangle\langle x_2\rangle\langle x_3\rangle\langle x_4\rangle\times\NN}{\KK}$ is given in Example~\ref{Exp:SpecialCase1}.
This gives the $\KK$-embedding $\fct{\tau'}{\EE[z]}{\seqK}$ with $\tau'(f)=(\ev'(f,n))_{n\geq0}$.
By the ansatz $\tau(x_i)=\tau(c_i\,\bar{g}_i)$ with $i=1,2$ for the given $\bar{g}_i$ in Example~\ref{Exp:lambdaDefSimple}, i.e., by $\ev(x_i,n)=c_i\ev(\bar{g}_i,n)$, we conclude that $c_1=\frac1{400}$ and $c_2=\tfrac1{4199040}$. Thus with these particular values we obtain the surjective difference ring homomorphism $\fct{\lambda}{\HH}{\EE}$ with $\HH=\FF\langle x_1\rangle\langle x_2\rangle\langle x_3\rangle\langle x_4\rangle$, $\lambda|_{\FF}=\id$ and~\eqref{Equ:LambdaDefSimple} such that the diagram~\eqref{Equ:CommutingDiagramWithRExt} commutes. In particular, we get
\begin{multline}\label{equ:tauKerSimple}
\ker(\tau)=\ker(\lambda)=\langle x_1-(c_1\,\bar{g}_1)^1,x_2^2-(c_2\,\bar{g}_2)^2\rangle_{\HH}
=\langle x_1-\frac1{400}(4 + x)^2 (5 + x)^2,\\x_2^2-\tfrac1{4199040^2}(1 + x)^4 (2 + x)^{10} (3 + x)^{16} (4 + x)^2 (5 + x)^2\rangle_{\HH}.
\end{multline}
Note that $\tau=\tau'\circ\lambda$ implies that
\begin{align*}
\ev(x_1,n)&=\tfrac{(4 + n)^2 (5 + n)^2}{400},&
\ev(x_2,n)&=\tfrac{(1 + n)^2 (2 + n)^5 (3 + n)^8 (4 + n) (5 + n)(-1)^n}{4199040}
\end{align*}
holds. In short, we derived a simplification of the products given on the right-hand sides of $\ev(x_1,n)$ and $\ev(x_2,n)$ in ~\eqref{Equ:Exp:DefineSpecialEv}.
\end{example}

Finally, we will show in Theorem~\ref{Thm:OptimalityStatement} that the difference ring construction in Theorem~\ref{Thm:LambdaTauConstruction} is optimal. For this task we will use

\begin{lemma}\label{Lemma:AbstractOptimality}
Let $\dfield{\EE}{\sigma}$ be a \piE-extension of a difference field $\dfield{\FF}{\sigma}$ with 
$\EE=\FF\langle x_1\rangle\dots\langle x_r\rangle$ and $f_i=\frac{\sigma(x_i)}{x_i}\in\FF^*$ for $1\leq i\leq r$, and let $\dfield{\HH}{\sigma}$ be an \rpiE-extension of $\HH=\FF\langle y_1\rangle\dots\langle y_s\rangle[z_1]\dots[z_l]$ where the $y_i$ with $1\leq i\leq s$ are \piE-monomials with $\alpha_i=\frac{\sigma(y_i)}{y_i}\in\FF\langle y_1\rangle\dots\langle y_{i-1}\rangle^*$ and the $z_i$ with $1\leq i\leq l$ are \rE-monomials with $a_i=\frac{\sigma(z_i)}{z_i}\in(\const{\FF}{\sigma})^*$. If there is a difference ring embedding $\fct{\lambda}{\EE}{\HH}$ with $\lambda|_{\FF}=\id$, then
$r\leq s$.
\end{lemma}

\begin{proof}
Let $\KK=\const{\FF}{\sigma}$ and suppose that $r>s$. 
For any $i\in\NN$ with $1\leq i\leq r$ we have 
$\frac{\sigma(\lambda(x_i))}{\lambda(x_i)}=\lambda(\frac{\sigma(x_i)}{x_i})=\lambda(f_i)=f_i.$
With Proposition~\ref{Prop:StructureofProductSol} we get
$\lambda(x_i)=g_i\,y_1^{n_{i,1}}\dots y_s^{n_{i,s}}z_1^{m_{i,1}}\dots z_l^{m_{i,l}}$
for some $n_{i,j}\in\ZZ$ with $1\leq j\leq s$, $m_{i,j}\in\NN$ with $1\leq j\leq l$ and $g_i\in\FF^*$. In particular, we get
\begin{equation}\label{Equ:fiRelation}
f_i=\frac{\sigma(g_i)}{g_i}\alpha_1^{n_{i,1}}\dots \alpha_s^{n_{i,s}}a_1^{m_{i,1}}\dots a_l^{m_{i,l}}.
\end{equation}
Since $r>s$, there is a nonzero vector $(\xi_1,\dots,\xi_r)\in\ZZ^r$ with
\begin{equation}\label{Equ:LASystem}
\left(\begin{smallmatrix} n_{1,1} & n_{2,1}&\dots&n_{r,1}\\
n_{1,2} & n_{2,2}&\dots&n_{r,2}\\
\vdots &\vdots& &\vdots\\
n_{1,s} & n_{2,s}&\dots&n_{r,s}
         \end{smallmatrix}\right)\left(\begin{smallmatrix} \xi_1\\\xi_2\\ \vdots\\ \xi_r\end{smallmatrix}\right)=\vect{0}.
\end{equation}
Exponentiation of~\eqref{Equ:fiRelation} with $\xi_i$ and multiplying these equations for $1\leq i\leq r$ yield
$$f_1^{\xi_1}\dots f_r^{\xi_r}=\frac{\sigma(g_1^{\xi_1}\dots g_r^{\xi_r})}{g_1^{\xi_1}\dots g_r^{\xi_r}}\alpha_{1}^{\tilde{n}_1}\dots\alpha_s^{\tilde{n}_s}a_1^{\tilde{m}_1}\dots a_l^{\tilde{m}_l}$$
with $\tilde{n}_i=n_{1,i}\xi_1+\dots+n_{r,i}\xi_r=0$ for $1\leq i\leq s$ and $\tilde{m}_i=m_{1,i}\xi_1+\dots+m_{r,i}\xi_r\in\NN$ for $1\leq i\leq l$. With~\eqref{Equ:LASystem} we obtain
$$f_1^{\xi_1}\dots f_r^{\xi_r}=\frac{\sigma(\gamma)}{\gamma}\tilde{a}$$
where $\gamma=g_1^{\xi_1}\dots g_r^{\xi_r}\in\FF^*$ and $\tilde{a}=a_1^{\tilde{m}_1}\dots a_l^{\tilde{m}_l}\in\KK^*$. Since $a_1,\dots,a_l$ are roots of unity, also $\tilde{a}$ is a root of unity. Hence we can take a $\nu\in\NN$ with $\tilde{a}^{\nu}=1$. Thus 
$$f_1^{\nu\,\xi_1}\dots f_r^{\nu\,\xi_r}=\frac{\sigma(\gamma^{\nu})}{\gamma^{\nu}}$$
and consequently $\dfield{\FF\langle y_1\rangle\dots\langle y_r\rangle}{\sigma}$ is not a \piE-extension of $\dfield{\FF}{\sigma}$ by Proposition~\ref{Prop:PiEquivalences}, a contradiction. 
\end{proof}

\begin{theorem}\label{Thm:OptimalityStatement}
Suppose that the properties of $\dfield{\FF}{\sigma}$ and 
$\dfield{\FF\langle x_1\rangle\dots\langle x_r\rangle}{\sigma}$ as stated in Theorem~\ref{Thm:LambdaTauConstruction} hold.
Then the constructions in (2a) and (2b) of Theorem~\ref{Thm:LambdaTauConstruction} are optimal: For any $R\Pi$-extension $\dfield{\HH}{\sigma}$ of $\dfield{\FF}{\sigma}$ where $\HH=\FF\langle y_1\rangle\dots\langle y_e\rangle[z_1]\dots[z_l]$ with $\Pi$-monomials $y_i$ ($\frac{\sigma(y_i)}{y_i}\in\FF\langle y_1\rangle\dots\langle y_{i-1}\rangle^*$) and $R$-mono\-mials $z_i$  ($\frac{\sigma(z_i)}{z_i}\in\KK^*$)
equipped with a difference ring homomorphism $\fct{\lambda'}{\FF\langle x_1\rangle\dots\langle x_r\rangle}{\HH}$ with $\lambda'_{\FF}=\id$ and $\KK$-embedding $\fct{\tau''}{\HH}{\seqK}$ with $\tau''|_{\FF}=\bar{\tau}$ such that
\begin{equation}\label{Equ:CommutingDiagramArbitraryExt}
 \xymatrix@!R=0.7cm@C1.8cm{
 \FF\langle x_1\rangle\dots\langle x_r\rangle\ar@{>>}[r]^{\lambda'}\ar[dr]_{\tau} &\HH
\ar@{^{(}->}[d]_{\tau''}\\
 &\seqK}
 \end{equation}
commutes, the following holds.
\begin{enumerate}
\item $r-u\leq e$. 
\item In addition, if $d_u>1$, then $l\geq1$ and
\begin{equation}\label{Equ:OrdOptimality}
\ord(z)\leq\lcm(\ord(z_1),\dots,\ord(z_l))=\ord(z_1)\dots\ord(z_l).
\end{equation}
\end{enumerate} 
\end{theorem}

\begin{proof}
(1) Within both constructions of (2a) and (2b) in Theorem~\ref{Thm:LambdaTauConstruction} we have $\lambda(\FF\langle x_{u+1}\rangle\dots\langle x_r\rangle)=\FF\langle x_{u+1}\rangle\dots\langle x_r\rangle=:\EE$. Thus 
$$\tau(\FF\langle x_{u+1}\rangle\dots\langle x_r\rangle)=\HH_1$$
with $\HH_1=\tau(\lambda(\FF\langle x_{u+1}\rangle\dots\langle x_r\rangle))=\tau(\FF)\langle \tau(x_{u+1})\rangle\dots\langle \tau(x_{r})\rangle$ with $\frac{\Shift(\tau(x_i))}{\tau(x_i)}=\tau(\frac{\sigma(x_i)}{x_i})\in\tau(\FF)$ for $u+1\leq i\leq r$.
Since $\tau|_{\EE}$ is a difference ring embedding, $\dfield{\HH_1}{\Shift}$ is a \piE-extension of $\dfield{\tau(\FF)}{\sigma}$ with the \piE-monomials $\tau(x_{i})$ for $u+1\leq i\leq r$. Furthermore, 
$$\tau''(\FF\langle y_1\rangle\dots\langle y_e\rangle[z_1]\dots[z_l])=\HH_2$$
with $\HH_2=\tau''(\FF)\langle \tau''(y_1)\rangle\dots\langle \tau''(y_e)\rangle[\tau''(z_1)]\dots[\tau''(z_l)])$.
As $\tau''$ is a difference ring embedding, $\dfield{\HH_2}{\Shift}$ is an \rpiE-extension of $\dfield{\tau''(\FF)}{\Shift}$.
Since $\tau(f)=\tau''(\lambda'(f))$ for all $f\in\EE$, $\dfield{\HH_1}{\Shift}$ is a subdifference ring of $\dfield{\HH_2}{\Shift}$. In particular, there is the trivial difference ring embedding $\fct{\tilde{\lambda}}{\HH_1}{\HH_2}$ with $\tilde{\lambda}=\id$. Hence we can apply Lemma~\ref{Lemma:AbstractOptimality} and we conclude that $r-u\leq e$.\\
(2) If $d_u>1$, it remains to show that the order of the \rE-monomial $z$ in our construction (2b) is optimal. 
By~\eqref{Equ:LambdaMap} and $\alpha_u:=\frac{\sigma(x_u)}{x_u}\in\FF^*$ we get 
\begin{equation}\label{Equ:alphauRelation1}
\alpha_u=\lambda(\alpha_u)=\lambda(\frac{\sigma(x_u)}{x_u})=\frac{\sigma(\lambda(x_u))}{\lambda(x_u)}=\rho^{d_u-\nu_u}\frac{\sigma(\bar{g}_u)}{\bar{g}_u}
\end{equation}
for some $\bar{g}_u\in\FF^*$ and $0\leq \nu_u\leq d_u-1$ with
\begin{equation}\label{Equ:Connectrhowithd}
\ord(\rho^{d_u-\nu_u})=d_u.
\end{equation}
By assumption we have
\begin{equation}\label{Equ:alpha_uAssum}
\alpha_u=\lambda'(\alpha_u)=\lambda'(\frac{\sigma(x_u)}{x_u})=\frac{\sigma(\lambda'(x_u))}{\lambda'(x_u)},
\end{equation}
and by Proposition~\ref{Prop:StructureofProductSol} we get
\begin{equation}\label{Equ:lambda(u)Shape}
\lambda'(x_u)=h\,y_1^{n_1}\dots y_e^{n_e}z_1^{m_1}\dots z_l^{m_l}
\end{equation}
with $n_i\in\ZZ$ for $1\leq i\leq e$, $m_i\in\NN$ where\footnote{For~\eqref{Equ:rhoBoundDivMi} we will require that $m_i\neq0$. Thus we allow $m_i=\ord(z_i)$ instead of $m_i=0$.} $1\leq m_i\leq\ord(z_i)$ for $1\leq i\leq l$ and $h\in\FF^*$. Let $a_i=\frac{\sigma(y_i)}{y_i}\in\FF\langle y_1\rangle\dots\langle y_{i-1}\rangle^*$ for $1\leq i\leq e$ and let $\gamma_i=\frac{\sigma(z_i)}{z_i}\in\KK^*$ for $1\leq i\leq l$ where the $\gamma_i$ are roots of unity.
Then plugging~\eqref{Equ:lambda(u)Shape} into~\eqref{Equ:alpha_uAssum} yields
\begin{equation}\label{Equ:alphauRelation2}
\alpha_u=\frac{\sigma(h)}{h}a_1^{n_1}\dots a_{n_e}^{n_e}\,b
\end{equation}
with
\begin{equation}\label{Equ:GammaRelation}
b=\gamma_1^{m_1}\dots\gamma_{l}^{m_l}\in\KK^*
\end{equation}
being a root of unity. Note that we also cover the case $l=0$ (i.e., no $R$-monomials arise); then we simply get $b=1$.
Combining~\eqref{Equ:alphauRelation1} with~\eqref{Equ:alphauRelation2} gives
\begin{equation}\label{Equ:DefinitionofWForRProof}
\frac{\sigma(\gamma)}{\gamma}a_1^{n_1}\dots a_{n_e}^{n_e}=\frac{\rho^{d_u-\nu_u}}{b}=:w\in\KK^*
\end{equation}
with $\gamma=\frac{h}{\bar{g}_u}\in\FF^*$.
Since the right-hand side $w\in\KK^*$ of the last equation is also a root of unity, we can choose a $k>0$ such that $w^k=1$. Thus
$\frac{\sigma(\gamma^k)}{\gamma^k}a_1^{k\,n_1}\dots a_{e}^{k\,n_e}=1.$
Suppose that $(n_1,\dots,n_e)\neq\vect{0}$ and
let $\nu$ be maximal such that $n_{\nu}\neq0$. Then $\sigma(v)=a_{\nu}^{-k\,n_{\nu}}\,v$ with $v:=\gamma^ky_1^{k\,n_1}\dots y_{\nu-1}^{k\,n_{\nu-1}}\in\FF\langle y_1\rangle\dots\langle y_{\nu-1}\rangle\setminus\{0\}$. This contradicts to Proposition~\ref{Prop:PiChar} and the assumption that $y_{\nu}$ is a \piE-monomial.
Thus $n_1=\dots=n_e=0$ and hence with~\eqref{Equ:DefinitionofWForRProof} we get
\begin{equation}\label{Equ:SimplewRelation}
\frac{\sigma(\gamma)}{\gamma}=w(=\frac{\rho^{d_u-\nu_u}}{b}). 
\end{equation}
Therefore
$\sigma^k(\gamma)=\sigma^{k-1}(w\,\gamma)=w\sigma^{k-1}(\gamma)=\dots=w^k\,\gamma=\gamma$
and thus $\gamma\in\const{\FF}{\sigma^k}$. Since there is a $\KK$-embedding $\fct{\bar{\tau}}{\FF}{\seqK}$, 
$\dfield{\FF}{\sigma}$ is constant-stable by Lemma~\ref{Lemma:TauImpliesConstantStable}. Thus $\gamma\in\KK^*$. With~\eqref{Equ:SimplewRelation} it follows that $\rho^{\nu_u}=b$. Further, with~\eqref{Equ:Connectrhowithd} we get $\ord(b)=d_u>1$. In particular, by~\eqref{Equ:GammaRelation} we obtain $l\geq1$. Furthermore,
\begin{equation}\label{Equ:rhoBoundDivMi}
d_u=\ord(\rho^{\nu_u})=\ord(b)=\frac{\ord(\gamma_1)\dots\ord(\gamma_l)}{m_1\dots m_l}.
\end{equation}
The last equality follows from Proposition~\ref{Prop:FromAlphaToRPi}: the $\gamma_i$ are roots of unity where $\gcd(\ord(\gamma_i),\ord(\gamma_j))=1$ for pairwise different $i,j$. Together with~\eqref{Equ:Connectrhowithd} we get $d_u\mid\ord(z_1)\dots\ord(z_e)=\lcm(\ord(\alpha_1),\dots,\ord(\alpha_e))=\lcm(\ord(z_1),\dots,\ord(z_e))$. 
\end{proof}

\begin{example}[Cont.~Example~\ref{Exp:SpecialCaseTauExp}]
In order to rephrase $\dfield{\KK(x)\langle x_1\rangle\langle x_2\rangle\langle x_3\rangle\langle x_4\rangle}{\sigma}$ from Example~\ref{Exp:SpecialCaseTauExp} in an $R\Pi$-extension, one needs at least 2 $\Pi$-monomials and one $R$-monomial of order 2. Such an optimal choice is given by the $R\Pi$-extension $\dfield{\KK(x)\langle x_3\rangle\langle x_4\rangle[z]}{\sigma}$ of $\dfield{\KK}{\sigma}$ 
in Example~\ref{Exp:SpecialCaseTauExp}.
\end{example}

\section{The solution of problem~\DR\ for the general case}\label{Sec:GeneralCase}

Finally, suppose that we are given a difference field $\dfield{\FF}{\sigma}$ with constant field $\KK$ 
equipped with an evaluation function $\fct{\bar{\ev}}{\FF\times\NN}{\KK}$ that 
satisfies the properties enumerated in Assumption~\ref{Assum:AlgProp}. In particular, suppose that the $\KK$-homomorphism $\fct{\bar{\tau}}{\FF}{\seqK}$ defined by $\bar{\tau}(f)=(\bar{\ev}(f,n))_{n\geq0}$ for $f\in\FF$ is injective.

We will solve Problem~\DR\ for a general $P$-extension $\dfield{\EE}{\sigma}$ of $\dfield{\FF}{\sigma}$ with $\EE=\FF\langle\hat{x}_1\rangle\dots\langle\hat{x}_r\rangle$ and $\hat{\alpha}_i=\frac{\sigma(\hat{x}_i)}{\hat{x}_i}\in\FF^*$ for $1\leq i \leq r$ equipped with an evaluation function $\fct{\ev}{\EE\times\NN}{\KK}$ where $\ev|_{\FF\times\NN}=\bar{ev}$.
Note that for the $\KK$-homomorphism $\fct{\tau}{\EE}{\seqK}$ defined by $\tau(f)=(\ev(f,n))_{n\geq0}$ for $f\in\EE$ we have $\tau|_{\FF}=\bar{\tau}$. 

We start to compute a basis ${\mathcal B}=\{z_{i1},\dots,z_{ir}\}_{1\leq i\leq u}$ of the $\ZZ$-module $\hat{V}=M((\hat{\alpha}_1,\dots,\hat{\alpha}_r),\FF)$ of $\ZZ^r$; this is possible by the requirement~(5) in Assumption~\ref{Assum:AlgProp}.

Then we distinguish two cases. If the rank of $\hat{V}$ is $u=0$, i.e., $\hat{V}=\{0\}$, we can conclude that $\dfield{\EE}{\sigma}$ is a \piE-extension of $\dfield{\FF}{\sigma}$  by Proposition~\ref{Prop:PiEquivalences} and that $\fct{\tau}{\EE}{\seqK}$ is injective by Theorem~\ref{Thm:RPSImpliesInjective}. Thus we can take $\HH:=\EE$, $\lambda:=\id$ and $\tau':=\tau$ and obtain a solution of Problem~\DR\ where $\lambda$ is even bijective. Since $\ker(\tau)=\{0\}$, there is no non-zero Laurent polynomial in $\EE$ whose sequence evaluation turns to the zero-sequence. Finally note that there is no solution of Problem~DR with a \piE-extension with less than $r$ \piE-monomials by Lemma~\ref{Lemma:AbstractOptimality}. 

In order to treat the remaining case where $\hat{V}$ has rank $u\geq1$,
we compute the Smith normal form\footnote{So far it was sufficient in our applications to use the standard algorithm (see, e.g.~\cite{Cohn:89}) based on column and row reductions to calculate the Smith Normal form for integer matrices. For faster algorithms we refer to~\cite{SNF1,SNF2}. For an excellent survey on the available strategies see also~\cite{SNF3}.}
\begin{equation}\label{Equ:DMatrix}
\left(\begin{smallmatrix}
d_1&0&\dots&\dots&\dots&0\\
&\ddots\\
&&d_u&0&\dots&0\\
\end{smallmatrix}\right)\in\ZZ^{u\times r}
\end{equation}
with $d_1,\dots,d_u\in\NN\setminus\{0\}$ and $d_1\mid d_2\mid\dots\mid d_u$
of the integer matrix $Z=(z_{ij})_{ij}\in\ZZ^{u\times r}$ whose entries are given from ${\mathcal B}$. 
In particular, we can compute matrices $A\in\ZZ^{u\times u}$ and $B\in\ZZ^{r\times r}$ being invertible over $\ZZ$ with
\begin{equation}\label{Equ:AZBFactorization}
A\,Z\,B=D.
\end{equation}

Note that for any invertible matrix $C\in\ZZ^{u\times u}$ over $\ZZ$, $D$ is again the Smith normal form of $C\,Z=(\tilde{b}_{ij})_{ij}$. Moreover observe that $\tilde{\mathcal{B}}=\{\tilde{b}_{i1},\dots,\tilde{b}_{ir}\}_{1\leq i\leq u}$ is again a basis of $\hat{V}$. In particular $C$ can be considered as a basis transform between the two bases $\mathcal{B}$ and $\tilde{\mathcal{B}}$. In other words, for any choice of basis of $\hat{V}$, we obtain the same Smith normal form $D$. 
We call $u$ the \notion{rank} and $d_r$ the \notion{largest divisor} of $\hat{V}$. 

\begin{example}\label{Equ:MainExample:Start}
Take the $P$-extension
$\dfield{\FF\langle\hat{x}_1\rangle\langle\hat{x}_2\rangle\langle\hat{x}_3\rangle\langle\hat{x}_4\rangle}{\sigma}$ of $\dfield{\FF}{\sigma}$ with $\sigma(\hat{x}_i)=\hat{\alpha}_i\,\hat{x}_i$ for $1\leq i\leq 4$ where the $\hat{\alpha}_i$ are given in~\eqref{Equ:MainAlpha}
Using Karr's algorithm~\cite{Karr:81} we can compute the basis
$\{(-6, 0, -4, 6), (0, 1, 0, -2)\}$ of $\hat{V}=M((\hat{\alpha}_1,\hat{\alpha}_2,\hat{\alpha}_3,\hat{\alpha}_4),\FF)$. We calculate the Smith normal form of $Z=\left(\begin{smallmatrix}
-6& 0& -4& 6\\ 
0& 1& 0& -2
\end{smallmatrix}\right)$
and get 
\begin{equation}\label{Equ:SmithExample}
\underbrace{\left(\begin{array}{cc}
 0 & 1 \\
 1 & 2 \\
\end{array}
\right)}_{=A}\,
Z\,\underbrace{\left(
\begin{array}{cccc}
 -1 & 1 & -2 & 1 \\
 1 & 0 & 0 & 2 \\
 2 & -2 & 3 & 0 \\
 0 & 0 & 0 & 1 \\
\end{array}
\right)}_{=B}=\underbrace{\left(
\begin{array}{cccc}
 1 & 0 & 0 & 0 \\
 0 & 2 & 0 & 0 \\
\end{array}
\right)}_{=D}.
\end{equation}
Thus $u=2$ and $d_1=1$, $d_2=2$; in particular the largest divisor of $\hat{V}$ is $2$.
\end{example}

In the following we will show how one can solve Problem~\DR\ under the assumption that a basis of $\hat{V}=M((\hat{\alpha}_1,\dots,\hat{\alpha}_r),\FF)$ is given. For this task, the following two observations are central.
\begin{enumerate}
\item Let $\tilde{b}_{i,j}\in\ZZ$ be the coefficients of the inverse matrix $B$, i.e., $B^{-1}=(\tilde{b}_{ij})_{ij}$ and define the $P$-extension $\dfield{\FF\langle x_1\rangle\dots\langle x_r\rangle}{\sigma}$ of $\dfield{\FF}{\sigma}$ with 
\begin{equation}\label{Equ:TransformedDR}
\frac{\sigma(x_i)}{x_i}=\alpha_1^{\tilde{b}_{i1}}\dots \alpha_r^{\tilde{b}_{ir}}=:\alpha_i\in\FF\quad 1\leq i\leq r.
\end{equation}
Then by Lemma~\ref{Lemma:IsomorphicDR} below the difference ring $\dfield{\FF\langle\hat{x}_1\rangle\dots\langle \hat{x}_r\rangle}{\sigma}$  is isomorphic to the difference ring $\dfield{\FF\langle x_1\rangle\dots\langle x_r\rangle}{\sigma}$, i.e., they cannot be distinguished up to renaming of the arising objects by an explicitly given bijective map $\fct{\mu}{\FF\langle \hat{x}_1\rangle\dots\langle\hat{x}_r\rangle}{\FF\langle x_1\rangle\dots\langle x_r\rangle}$.
\item Furthermore, using Lemma~\ref{Lemma:TransformationForMultipland} below one can read off a basis of $V=M(\alpha_1,\dots,\alpha_r),\FF)$ by looking at the matrix $D$ from~\eqref{Equ:DMatrix}: it is simplify~\eqref{Equ:SmithBasisforM}. 
\end{enumerate}
Thus using the difference ring isomorphism $\mu^{-1}$ one can carry over all the results of Section~\ref{Sec:SpecialCase}, in particular Theorem~\ref{Thm:OptimalityStatement} can be extended from the special case
$\dfield{\FF\langle x_1\rangle\dots\langle x_r\rangle}{\sigma}$ to the general version $\dfield{\FF\langle \hat{x}_1\rangle\dots\langle\hat{x}_r\rangle}{\sigma}$.

\begin{lemma}\label{Lemma:TransformationForMultipland}
Let $\dfield{\FF}{\sigma}$ be a difference field with $(\hat{\alpha}_1,\dots,\hat{\alpha}_r)\in(\FF^*)^r$. Take a basis $\{(z_{i1},\dots,z_{ir})\}_{1\leq i\leq u}$ of $M((\hat{\alpha}_1,\dots,\hat{\alpha}_r),\FF)$, in particular, take
$w_i\in\FF^*$ for $1\leq i\leq u$ with
$$\frac{\sigma(w_i)}{w_i}=\hat{\alpha}_1^{z_{i1}}\dots\hat{\alpha}_r^{z_{ir}}.$$
Consider the Smith normal form of $Z=(z_{ij})_{ij}\in\ZZ^{u\times r}$, i.e., the equation~\eqref{Equ:AZBFactorization} holds
for some $A=(a_{ij})_{ij}\in\ZZ^{u\times u}$ and $B\in\ZZ^{r\times r}$ being invertible matrices over $\ZZ$, in particular, $B^{-1}=(\tilde{b}_{ij})\in\ZZ^{r\times r}$ and~\eqref{Equ:DMatrix}
being a diagonal matrix.
Define
\begin{align}
\alpha_i&=\hat{\alpha}_1^{\tilde{b}_{i1}}\dots \hat{\alpha}_r^{\tilde{b}_{ir}}\in\FF^*&1\leq i\leq r,\label{Equ:tildeAlphaMap}\\
\tilde{w}_i&=w_1^{a_{i1}}\dots w_u^{a_{iu}}\in\FF^*&1\leq i\leq u.
\end{align}
Then a basis of $M((\alpha_1,\dots,\alpha_r),\FF)$ is 
\begin{equation}\label{Equ:DBasis}
\{(d_1,0,\dots,0),(0,d_2,0,\dots,0),\dots,(0,\dots,0,d_u,0,\dots,0)\}
\end{equation}
and for $1\leq i\leq u$ we get
\begin{equation}\label{tileWiRelation}
\frac{\sigma(\tilde{w}_i)}{\tilde{w}_i}=\alpha_i^{d_i}.
\end{equation}
\end{lemma}
\begin{proof}
Let $(\tilde{z}_{ij})_{ij}=A\,Z\in\ZZ^{u\times r}$.
Since $A$ is invertible over $\ZZ$, $A$ can be considered as a basis transformation. Consequently, also $\{(\tilde{z}_{i1},\dots,{z}_{ir})\}_{1\leq i\leq u}$ is a basis of $M((\hat{\alpha}_1,\dots,\hat{\alpha}_r),\FF)$. In particular, for $1\leq i\leq u$ we get
\begin{equation}\label{tileWiRelation1}
\frac{\sigma(\tilde{w}_i)}{\tilde{w}_i}=\left(\frac{\sigma(w_1)}{w_1}\right)^{a_{i1}}\dots\left(\frac{\sigma(w_u)}{w_u}\right)^{a_{iu}}=\hat{\alpha}_1^{p_{i1}}\dots\hat{\alpha}_r^{p_{ir}}
\end{equation}
with 
$$p_{ik}=a_{i1}z_{1k}+a_{i2}z_{2k}+\dots+a_{iu}z_{uk}=\tilde{z}_{ik},\quad 1\leq k\leq r;$$
the last equality follows from $(p_{ik})_{ik}=A\,Z=(z_{ik})_{ik}\in\ZZ^{u\times r}$.
Let $D=(d_{ij})_{ij}\in\ZZ^{u\times r}$, i.e., $d_{ii}=d_i$ for $1\leq i\leq u$ and all other entries are zero.
For $1\leq i\leq r$ we have
$$\alpha_i^{d_i}=\alpha_1^{d_{i1}}\dots\alpha_u^{d_{iu}}=\hat{\alpha}_1^{q_{i1}}\dots\hat{\alpha}_r^{q_{ir}}$$
with
$$q_{ik}=d_{i1}\tilde{b}_{1k}+d_{i2}\tilde{b}_{2k}+\dots+d_{ir}\tilde{b}_{rk}=\tilde{z}_{ik};$$
the last equality follows from $(\tilde{q}_{ik})_{ik}=D\,B^{-1}=A\,Z=\tilde{Z}=(\tilde{z}_{ik})_{ik}\in\ZZ^{r\times r}$. With~\eqref{tileWiRelation1} we conclude that~\eqref{tileWiRelation} holds.
Thus the elements in~\eqref{Equ:DBasis} are contained in $V=M((\alpha_1,\dots,\alpha_r),\FF)$. Finally, we will show that the linearly independent vectors in~\eqref{Equ:DBasis} generate $V$. 
Let $(n_1,\dots,n_r)\in\ZZ^r$ such that there is a $w\in\FF^*$ with $\frac{\sigma(w)}{w}=\alpha_1^{n_1}\dots\alpha_r^{n_r}$. Then using~\eqref{Equ:tildeAlphaMap} we obtain
\begin{align*}
\frac{\sigma(w)}{w}&=(\hat{\alpha}_1^{\tilde{b}_{11}}\dots\hat{\alpha}_r^{\tilde{b}_{1r}})^{n_1}(\hat{\alpha}_1^{\tilde{b}_{21}}\dots\hat{\alpha}_r^{\tilde{b}_{2r}})^{n_2}\dots(\hat{\alpha}_1^{\tilde{b}_{r1}}\dots\hat{\alpha}_r^{\tilde{b}_{rr}})^{n_r}\\
&=\hat{\alpha}_1^{\tilde{b}_{11}\,n_1+\dots+\tilde{b}_{r1}\,n_r}\dots\hat{\alpha}_r^{\tilde{b}_{1r}\,n_1+\dots+\tilde{b}_{rr}\,n_r}
=\hat{\alpha}_1^{m_1}\dots\hat{\alpha}_r^{m_r}
\end{align*}
with
$$(m_1,\dots,m_r)=(n_1,\dots,n_r)(\tilde{b}_{ij})_{ij}=(n_1,\dots,n_r)B^{-1}.$$
Conversely, $B$ maps $(m_1,\dots,m_r)$ to $(n_1,\dots,n_r)$. 
In particular, the given basis elements $(z_{i,1},\dots,z_{i,r})$ for $1\leq i\leq u$ of $M((\hat{\alpha}_1,\dots,\hat{\alpha}_r),\FF)$ are mapped in bijection to
$$(z_{i,1},\dots,z_{i,r})\,B=(\mu_{i,1},\dots,\mu_{i,r})$$
with $(\mu_{i,j})_{i,j}=Z\,B$. As a consequence, the vectors $(\mu_{i,1},\dots,\mu_{i,r})$ for $1\leq i\leq u$ form a basis of $V$. Since $A$ is invertible over $\ZZ$, these latter basis elements are mapped via a basis transformation with 
$$A\left(\begin{smallmatrix}\mu_{i,1}\\ \vdots\\ \mu_{i,r}\end{smallmatrix}\right)=\left(\begin{smallmatrix}\nu_{i,1}\\ \vdots\\ \nu_{i,r}\end{smallmatrix}\right)$$
to the another basis $\{(\nu_{i,1},\dots,\nu_{i,r})\}_{1\leq i\leq u}$ of $V$. With $A\,Z\,B=D$ we conclude that $(\nu_{i,1},\dots,\nu_{i,r})=(0,\dots,0,d_i,0,\dots,0)$ for $1\leq i\leq u$. Thus~\eqref{Equ:DBasis} is a basis of $V$.
\end{proof}

\begin{lemma}\label{Lemma:IsomorphicDR}
Let
$\dfield{\FF\langle \hat{x}_1\rangle\dots\langle \hat{x}_r\rangle}{\sigma}$ be a $P$-extension of a difference field $\dfield{\FF}{\sigma}$ with $\alpha_i=\frac{\sigma(\hat{x}_i)}{\hat{x}_i}\in\FF^*$ with $1\leq i\leq r$. Let $B=(b_{ij})_{ij}\in\ZZ^{r\times r}$ be an invertible matrix over $\ZZ$, i.e., $B^{-1}=(\tilde{b}_{ij})_{ij}\in\ZZ^{r\times r}$, and let $c_1,\dots,c_r\in\KK^*$.
Take the $P$-extension $\dfield{\FF\langle x_1\rangle\dots\langle x_r\rangle}{\sigma}$ of $\dfield{\FF}{\sigma}$ with~\eqref{Equ:TransformedDR}.
Then the ring homomorphism $\fct{\mu}{\FF\langle \hat{x}_1\rangle\dots\langle \hat{x}_r\rangle}{\FF\langle x_1\rangle\dots\langle x_r\rangle}$ given by $\mu|_{\FF}=\id$ and
\begin{equation}\label{Equ:MuConstruction}
\mu(\hat{x}_i)=c_i\,x_1^{b_{i1}}\dots x_r^{b_{ir}}.\quad 1\leq i\leq r
\end{equation}
is a difference ring isomorphism. Moreover, its inverse is given by $\mu^{-1}|_{\FF}=\id$ and
\begin{equation}\label{Equ:InverseMu}
\mu^{-1}(x_i)=(\tfrac{\hat{x}_1}{c_1})^{\tilde{b}_{i1}}\dots (\tfrac{\hat{x}_r}{c_r})^{\tilde{b}_{ir}},\quad 1\leq i\leq r.
\end{equation}
\end{lemma}
\begin{proof}
Similarly to the proof in Lemma~\ref{Lemma:TransformationForMultipland} one can verify that $\mu$ is invertible with~\eqref{Equ:InverseMu}.
In particular, for $1\leq i\leq r$ we have
\begin{align*}
\sigma(\mu(\hat{x}_i))&=\sigma(c_i\,x_1^{b_{i1}}\dots x_r^{b_{ir}})\\
&=(\alpha_1^{\tilde{b}_{11}}\dots \alpha_r^{\tilde{b}_{1r}})^{b_{i1}}(\alpha_1^{\tilde{b}_{21}}\dots \alpha_r^{\tilde{b}_{2r}})^{b_{i2}}\dots (\alpha_1^{\tilde{b}_{r1}}\dots \alpha_r^{\tilde{b}_{rr}})^{b_{ir}}c_ix_1^{b_{i1}}\dots x_r^{b_{ir}}\\
&=\alpha_1^{b_{i1}\tilde{b}_{11}+b_{i2}\tilde{b}_{21}+\dots+b_{ir}\tilde{b}_{r1}}\dots\alpha_1^{b_{i1}\tilde{b}_{1r}+b_{i2}\tilde{b}_{2r}+\dots+b_{ir}\tilde{b}_{rr}}\mu(\hat{x}_i)\\
&=\alpha_1^{p_{i1}}\dots\alpha_r^{p_{ir}}\mu(\hat{x}_i)
\end{align*}
with $(p_{ij})_{ij}=B\,B^{-1}=I_r$ where $I_r$ is the identity matrix. Consequently
$\sigma(\mu(\hat{x}_i))=\alpha_i\mu(\hat{x}_i)=\mu(\alpha_i\hat{x}_i)=\mu(\sigma(\hat{x}_i))$. Thus $\mu$ is a difference ring isomorphism.
\end{proof}

\begin{example}\label{Exp:ConstructMu}
Take $B$ from~\eqref{Equ:SmithExample} and compute the $\alpha_1,\alpha_2,\alpha_3,\alpha_4$ using~\eqref{Equ:TransformedDR}. This yields~\eqref{Equ:Exp:DefineSpecialAlpha}. Now compute
$$B^{-1}=\left(
\begin{array}{cccc}
 0 & 1 & 0 & -2 \\
 -3 & 1 & -2 & 1 \\
 -2 & 0 & -1 & 2 \\
 0 & 0 & 0 & 1 \\
\end{array}
\right).$$
Then using Lemma~\ref{Equ:MainExample:Start} it follows that the difference ring $\dfield{\FF\langle\hat{x}_1\rangle\langle\hat{x}_2\rangle\langle\hat{x}_3\rangle\langle \hat{x}_4\rangle}{\sigma}$ is isomorphic to $\dfield{\FF\langle x_1\rangle\langle x_2\rangle\langle x_3\rangle\langle x_4\rangle}{\sigma}$. More precisely, there is the difference ring isomorphism
$\fct{\mu}{\FF\langle\hat{x}_1\rangle\langle\hat{x}_2\rangle\langle\hat{x}_3\rangle\langle\hat{x}_4\rangle}{\FF\langle x_1\rangle\langle x_2\rangle\langle x_3\rangle\langle x_4\rangle}$ defined by $\mu|_{\FF}=\id$ and
\begin{align}\label{Equ:muExp}
\mu(\hat{x}_1)&=\frac{x_2
   x_4}{x_1 x_3^2},&\mu(\hat{x}_2)&=
   x_1 x_4^2,\\
\mu(\hat{x}_3)&=
   \frac{x_1^2
   x_3^3}{x_2^2},&\mu(\hat{x}_4)&=x_4;\nonumber
\end{align}
its inverse difference ring isomorphism
$\mu^{-1}$ is given by $\mu^{-1}|_{\FF}=\id$ and
\begin{align}\label{Equ:muInvExp}
\mu^{-1}(x_1)&=\frac{\hat{x}_2}{\hat{x}_4^2},&\mu^{-1}(x_2)&=\frac{\hat{x}_2 \hat{x}_4}{\hat{x}_1^3
   \hat{x}_3^2},\\
\mu^{-1}(x_3)&=\frac{\hat{x}_4^2}{\hat{x}_1^2\hat{x}_3},&\mu^{-1}(x_4)&=\hat{x}_4.\nonumber
\end{align}
Moreover, by Lemma~\ref{Lemma:TransformationForMultipland} one can read off the basis of $V=M((\alpha_1,\alpha_2,\alpha_3,\alpha_4),\FF)$ from $D$ in~\eqref{Equ:SmithExample}. Namely, we get~\eqref{Equ:MBasisSimple}. 
\end{example}

We are now in the position to apply the results from the previous Section~\ref{Sec:SpecialCase} and obtain the following main result. The constructed difference ring homomorphisms that arise in the proof of Theorem~\ref{Thm:MainResult} can be visualized by the following diagram.
\begin{equation}\label{Equ:FullDiagramm}
\xymatrix@!R=0.7cm@C1.8cm{
\FF\langle\hat{x}_1\rangle\dots\langle\hat{x}_r\rangle\ar@{^{(}->>}[r]^{\mu}\ar@{.>>}@/^2.0pc/[rr]^{\hat{\lambda}=\lambda\circ\mu}\ar@/^-1.0pc/[drr]_{\hat{\tau}}&\FF\langle x_1\rangle\dots\langle x_r\rangle\ar@{>>}[r]^{\lambda}\ar@{.>}[dr]_{\tau=\hat{\tau}\circ\mu^{-1}} &\HH\ar@{^{(}->}[d]_{\tau'}\\
& &\seqK.}
\end{equation}

\begin{theorem}\label{Thm:MainResult}
Let $\dfield{\FF}{\sigma}$ be a radical-stable difference field with $\KK=\const{\FF}{\sigma}$ equipped with a $\KK$-embedding $\fct{\bar{\tau}}{\FF}{\seqK}$. Let $\dfield{\FF\langle \hat{x}_1\rangle\dots\langle \hat{x}_r\rangle}{\sigma}$ be a $P$-extension of $\dfield{\FF}{\sigma}$ with $\hat{\alpha}_i=\frac{\sigma(\hat{x}_i)}{\hat{x}_i}\in\FF^*$ for $1\leq i\leq r$ together with a $\KK$-homomorphism $\fct{\hat{\tau}}{\FF\langle \hat{x}_1\rangle\dots\langle\hat{x}_r\rangle}{\seqK}$ with $\hat{\tau}|_{\FF}=\bar{\tau}$. Assume that $M((\hat{\alpha}_1,\dots,\hat{\alpha}_r),\FF)\neq\{\vect{0}\}$, and let $u\geq1$ be its rank and $d\geq1$ be its largest divisor.\\ 
Then there is an $R\Pi$-extension $\dfield{\HH}{\sigma}$ of the difference field $\dfield{\FF}{\sigma}$ together with a $\KK$-embedding $\fct{\tau'}{\HH}{\seqK}$ with $\tau'|_{\FF}=\bar{\tau}$ and a surjective difference ring homomorphism $\fct{\hat{\lambda}}{\FF\langle \hat{x}_1\rangle\dots\langle \hat{x}_r\rangle}{\HH}$
such that the following properties hold.
\begin{enumerate} 
\item We have $\tau'(\hat{\lambda}(f))=\hat{\tau}(f)$ for all $f\in\FF\langle \hat{x}_1\rangle\dots\langle\hat{x}_r\rangle$, i.e, the diagram
 \begin{equation}\label{Equ:CommutingDiagramGeneralSetting}
 \xymatrix@!R=0.7cm@C1.8cm{
 \FF\langle\hat{x}_1\rangle\dots\langle\hat{x}_r\rangle\ar@{>>}[r]^{\hat{\lambda}}\ar[dr]_{\hat{\tau}} &\HH\ar@{^{(}->}[d]_{\tau'}\\
 &\seqK}
 \end{equation} 
commutes. 
\item If $d=1$, $\HH$ is built by $r-u$ \piE-monomials, say $\HH=\FF\langle x_{u+1}\rangle\dots\langle x_{r}\rangle$ with $\alpha_i=\frac{\sigma(x_i)}{x_i}\in\FF^*$ for $u+1\leq i\leq r$. 
Otherwise,
$\HH=\FF\langle x_{u+1}\rangle\dots\langle x_{r}\rangle[z]$ where the $x_i$ are \piE-monomials with $\alpha_i=\frac{\sigma(x_i)}{x_i}\in\FF^*$ for $u+1\leq i\leq r$ and $z$ is an $R$-monomial of order $d$ with $\rho=\frac{\sigma(z)}{z}\in\KK^*$. 
\item $\hat{\lambda}$ is defined by $\hat{\lambda}|_{\FF}=\id$ and
\begin{equation}\label{Equ:LambdaGenDef}
\hat{\lambda}(\hat{x}_i)=\gamma_i\,z^{o_{i}}\,\,x_{u+1}^{m_{i,u+1}}\dots x_{r}^{m_{i,r}}
\end{equation}
for all $1\leq i\leq r$ for some $\gamma_i\in\FF^*$, $o_i\in\NN$ with $0\leq o_i< d$ and $m_{i,j}\in\ZZ$.
\item
If $\bar{\ev}$ is an evaluation function for $\bar{\tau}(=\hat{\tau}|_{\FF}=\tau'|_{\FF})$, then there is an evaluation function
$\ev'$ for $\tau'$ defined by $\ev'|_{\FF\times\NN}=\bar{\ev}$ and
\begin{equation}\label{ExtractBetterEv}
\ev'(x_i,n)=\kappa_i\prod_{k=l_i}^n\ev(\alpha_i,k-1)
\end{equation}
for all $u+1\leq i\leq r$ for some $l_i\in\NN$ and $\kappa_i\in\KK^*$; if $d>1$, we define in addition
$\ev'(z,n)=\rho^n$.
\item Among all possible $R\Pi$-extensions\footnote{\label{Footnote:Shape}This means that the $R\Pi$-extension $\dfield{\HH'}{\sigma}$ of the difference field $\dfield{\FF}{\sigma}$ has the form $\HH'=\FF\langle y_1\rangle\dots\langle y_e\rangle[z_1]\dots[z_l]$ with $\frac{\sigma(y_i)}{y_i}\in\FF\langle y_i\rangle\dots\langle y_{i-1}\rangle^*$ for $1\leq i\leq e$ and $\frac{\sigma(z_i)}{z_i}\in\KK^*$ for $1\leq i\leq l$.} with this property, $r-u$ is the minimal number of \piE-monomials; furthermore, if $d>1$, only one $R$-monomial is necessary and its order $d$ is minimal.
\item There are $n_{i,j}\in\ZZ$ where the rows in $(n_{i,j})_{i,j}\in\ZZ^{u\times r}$ are linearly independent 
such that for $g_i:=\hat{\lambda}(\hat{x}_1^{n_{i,1}}\dots \hat{x}_{r}^{n_{i,r}})\in\FF^*$ with $1\leq i\leq u$ we have
\begin{equation}\label{Equ:GeneralKerTau}
\ker(\hat{\tau})=\ker(\hat{\lambda})=\langle \hat{x}_1^{n_{1,1}}\dots \hat{x}_{r}^{n_{1,r}}-g_1,\dots,\hat{x}_1^{n_{u,1}}\dots \hat{x}_{r}^{n_{u,r}}-g_u\rangle_{\FF\langle \hat{x}_1\rangle\dots\langle \hat{x}_r\rangle}.
\end{equation}
\end{enumerate}
If the properties in Assumption~\ref{Assum:AlgProp} hold, then one can compute $\dfield{\HH}{\sigma}$, $\hat{\lambda}$ and the generators of $\ker(\hat{\lambda})$ explicitly. In particular, the evaluation function $\ev'$ for $\dfield{\HH}{\sigma}$ defined in~\eqref{ExtractBetterEv} can be given explicitly.
\end{theorem}
\begin{proof}
(1,2,3) Let ${\mathcal B}=\{(z_{i1},\dots,z_{ir})\}_{1\leq i\leq u}$ be a basis of $M((\hat{\alpha}_1,\dots,\hat{\alpha}_r),\FF)\neq\{\vect{0}\}$.
Consider the Smith normal form of $Z=(z_{ij})_{ij}\in\ZZ^{u\times r}$, i.e., the decomposition~\eqref{Equ:AZBFactorization}
with $A=(a_{ij})_{ij}\in\ZZ^{u\times u}$ and $B\in\ZZ^{r\times r}$ being invertible matrices over $\ZZ$, in particular, $B^{-1}=(\tilde{b}_{ij})\in\ZZ^{r\times r}$ with the diagonal matrix~\eqref{Equ:DMatrix} with $d_1\mid d_2\mid\dots\mid d_u=d$. Take the $P$-extension $\dfield{\FF\langle x_1\rangle\dots\langle x_r\rangle}{\sigma}$ of $\dfield{\FF}{\sigma}$ defined by~\eqref{Equ:TransformedDR}. By Lemma~\ref{Lemma:IsomorphicDR} there is the difference ring isomorphism $\fct{\mu}{\FF\langle \hat{x}_1\rangle\dots\langle \hat{x}_r\rangle}{\FF\langle x_1\rangle\dots\langle x_r\rangle}$ given by $\mu|_{\FF}=\id$ and~\eqref{Equ:MuConstruction} ($c_1=\dots=c_r=1$).
Define the map $\fct{\tau}{\FF\langle x_1\rangle\dots\langle x_r\rangle}{\seqK}$ with
$\tau=\hat{\tau}\circ\mu^{-1}$ which forms a $\KK$-homomorphism.
In particular, by Lemma~\ref{Lemma:TransformationForMultipland} a basis of $M((\alpha_1,\dots,\alpha_r),\FF)$ is given in~\eqref{Equ:DBasis}. Hence we can apply Theorem~\ref{Thm:LambdaTauConstruction} and obtain
a surjective difference ring homomorphism $\fct{\lambda}{\FF\langle x_1\rangle\dots\langle x_r\rangle}{\HH}$ with $\lambda|_{\FF}=\id$ and $\tau'(\lambda(f))=\tau(f)$ for all $f\in\FF\langle x_1\rangle\dots\langle x_r\rangle$. Thus with the surjective difference ring homomorphism 
\begin{equation}\label{Equ:HatLamdba}
\hat{\lambda}=\lambda\circ\mu 
\end{equation}
we get $\tau'(\hat{\lambda}(f))=\tau'(\lambda(\mu(f)))=\tau(\mu(f))=\hat{\tau}(f)$ which proves the first part. Part~(2) follows by (2a) and (2b) of Theorem~\ref{Thm:LambdaTauConstruction}. In particular, part (3) follows by~\eqref{Equ:HatLamdba}, the definition of $\mu$ given in~\eqref{Equ:MuConstruction} and the definition of $\lambda$ given in~\eqref{Equ:MapXiToEE} or ~\eqref{Equ:LambdaMap}.\\
(4) Since $\FF$ is a field, there is a $z$-function by part~(2) of Lemma~\ref{Lemma:ZFunctionForField} for $\bar{\ev}$. Thus we can activate Lemma~\ref{Lemma:ExtractEvFunction} and statement (4) follows.\\
(5) Suppose that there is an $R\Pi$-extension $\dfield{\HH'}{\sigma}$ of $\dfield{\FF}{\sigma}$ with the shape given in Footnote~\ref{Footnote:Shape} equipped with 
a difference ring homomorphism  $\fct{\lambda'}{\FF\langle \hat{x}_1\rangle\dots\langle \hat{x}_r\rangle}{\HH'}$ as claimed in statement (1). Suppose further that $\dfield{\HH'}{\sigma}$ is built by less than $r-u$ \piE-mononials or, if $d>1$, it can be built without \rE-monomial or is built with \rE-monomials but the product of their orders is smaller than $d$. Then $\dfield{\HH'}{\sigma}$ together with $\fct{\lambda''}{\FF\langle x_1\rangle\dots\langle x_r\rangle}{\HH'}$ where $\lambda''=\lambda'\circ\mu^{-1}$ yields a better construction than $\dfield{\HH}{\sigma}$ with $\lambda$, a contradiction to Theorem~\ref{Thm:OptimalityStatement}.\\
(6) By part~(3) of Theorem~\ref{Thm:LambdaTauConstruction} we have $\ker(\lambda)=\ker(\tau)=\langle x_1^{d_1}-g_1,\dots,x_u^{d_u}-g_u\rangle$ with $g_i:=\lambda(x_i^{d_i})$. With $\hat{\lambda}=\lambda\circ\mu$ and $\hat{\tau}=\tau\circ\mu$
we get
$$\ker(\hat{\lambda})=\ker(\hat{\tau})=\langle \mu^{-1}(x_1)^{d_1}-\mu^{-1}(g_1),\dots,\mu^{-1}(x_u)^{d_u}-\mu^{-1}(g_u)\rangle.$$ 
With $\mu^{-1}(g_i)=g_i$ for $1\leq i\leq u$ and~\eqref{Equ:InverseMu} it follows that 
$$\mu^{-1}(x_i)^{d_i}-\mu^{-1}(g_i)=(\hat{x}_1^{\tilde{b}_{i1}}\dots \hat{x}_r^{\tilde{b}_{ir}})^{d_i}-g_i$$ 
for $1\leq i\leq u$. In particular, for $n_{i,j}:=d_i\tilde{b}_{ir}$ we get~\eqref{Equ:GeneralKerTau}.
Since $B^{-1}=(\tilde{b}_{i,j})_{i,j}$, the first $u$ rows in $B^{-1}$ are linearly independent and thus the rows in $(n_{i,j})_{i,j}\in\ZZ^{u\times r}$ are linearly independent. 
Finally, with~\eqref{Equ:HatLamdba} it follows that $\hat{\lambda}((\hat{x}_1^{\tilde{b}_{i1}}\dots \hat{x}_r^{\tilde{b}_{ir}})^{d_i})=\lambda(\mu(\mu^{-1}(x_i))^{d_i})=\lambda(x_i^{d_i})=g_i$ which completes the proof of part~(6).\\
Suppose that the properties in Assumption~\ref{Assum:AlgProp} hold.
Then a basis ${\mathcal B}$ can be computed. Further, one can compute the Smith normal form of $Z$, and obtains $\mu$ and its inverse $\mu^{-1}$ explicitly. The remaining constructions in (1),(2) and~(6) follow by the constructive statements of Theorem~\ref{Thm:OptimalityStatement}. By Lemma~\ref{Lemma:ExtractEvFunction} we conclude that the claimed $\kappa_i$ and $l_i$ in~\eqref{ExtractBetterEv} can be calculated.
\end{proof}

\begin{example}[Details for Ex.~\ref{Exp:IntroExp}]\label{Exp:MainExampleResult}
We will illustrate the construction given in the proof of Theorem~\ref{Thm:MainResult} in order to obtain the calculations given in Example~\ref{Exp:IntroExp} above.\\
We take the rational difference field $\dfield{\FF}{\sigma}$ with $\FF=\KK(x)$, $\sigma(x)=x+1$ and $\const{\FF}{\sigma}=\KK=\QQ(\iota)$, which is a \pisiE-field over $\KK$. Furthermore we construct the $P$-extension $\dfield{\EE}{\sigma}$ of $\dfield{\FF}{\sigma}$ with $\EE=\FF\langle\hat{x}_1\rangle\langle\hat{x}_2\rangle\langle\hat{x}_3\rangle\langle\hat{x}_4\rangle$ where $\sigma(\hat{x}_i)=\hat{\alpha}_i\,\hat{x}_i$ for $1\leq i\leq 4$ and~\eqref{Equ:MainAlpha} as introduced in Example~\ref{Equ:MainExample:Start}. Further, we take the $\KK$-homomorphism $\fct{\hat{\tau}}{\FF\langle\hat{x}_1\rangle\langle\hat{x}_2\rangle\langle\hat{x}_3\rangle\langle\hat{x}_4\rangle}{\seqK}$ with the evaluation function
$\hat{\ev}$ defined by~\eqref{Equ:EvalRat} for $p,q\in\KK[x]$ and~\eqref{Equ:EvalhatXi} for $1\leq i\leq 4$ with~\eqref{Equ:DefineFi}. In a nutshell, we model the products~\eqref{Equ:DefineFi} with $\hat{x}_1,\hat{x}_2,\hat{x}_3,\hat{x}_4$, respectively. Utilizing the calculations in Example~\ref{Equ:MainExample:Start} we proceed with the constructions given in the proof of Theorem~\ref{Thm:MainResult} as follows.
As elaborated in Example~\ref{Exp:ConstructMu} we take the difference ring
$\dfield{\FF\langle{x}_1\rangle\langle{x}_2\rangle\langle{x}_3\rangle\langle{x}_4\rangle}{\sigma}$ 
with $\sigma(x_i)=\alpha_i\,t_i$ for $1\leq i\leq 4$ where the $\alpha_i$ are given in~\eqref{Equ:Exp:DefineSpecialAlpha}. By construction $\fct{\mu}{\FF\langle \hat{x}_1\rangle\langle \hat{x}_2\rangle\langle \hat{x}_3\rangle\langle \hat{x}_4\rangle}{\FF\langle{x}_1\rangle\langle{x}_2\rangle\langle{x}_3\rangle\langle{x}_4\rangle}$ forms a  difference ring isomorphism and a basis of $V=M((\alpha_1,\alpha_2,\alpha_3,\alpha_4),\FF)$ is~\eqref{Equ:MBasisSimple}. 
Now we can utilize the construction of Section~\ref{Sec:SpecialCase}. As preprocessing step, we define the $\KK$-homomorphism $\fct{\tau}{\FF\langle{x}_1\rangle\langle{x}_2\rangle\langle{x}_3\rangle\langle{x}_4\rangle}{\seqK}$ with
$\tau=\hat{\tau}\circ\mu^{-1}$. Here one obtains the evaluation function $\fct{\ev}{\FF\langle{x}_1\rangle\langle{x}_2\rangle\langle{x}_3\rangle\langle{x}_4\rangle\times\NN}{\KK}$ with
\begin{align*}
\ev(x_1,n)&=\big(
        \text{\scriptsize$\displaystyle\prod_{k=1}^n$} \tfrac{26244 k^2 (2+k)^2}{(3+k)^2}
\big)
\big(\text{\scriptsize$\displaystyle\prod_{k=1}^n$}\tfrac{-(5+k)}{162 k (2+k)}\big)^2,\\
\ev(x_2,n)&= \big(
        \text{\scriptsize$\displaystyle\prod_{k=1}^n$} \tfrac{26244 k^2 (2+k)^2}{(3+k)^2}
\big)
\big(\text{\scriptsize$\displaystyle\prod_{k=1}^n$}\tfrac{-(3+k)^3}{13122 k (1+k)}\big)^3 \big(
        \text{\scriptsize$\displaystyle\prod_{k=1}^n$}\tfrac{-162 k (2+k)}{5+k}
\big)
\big(\text{\scriptsize$\displaystyle\prod_{k=1}^n$} \tfrac{-729 \iota (5+k)}{k (2+k)^3}\big)^2\\
\ev(x_3,n)&= \big(
        \text{\scriptsize$\displaystyle\prod_{k=1}^n$}\tfrac{-(3+k)^3}{13122 k (1+k)}\big)^2 \big(
        \text{\scriptsize$\displaystyle\prod_{k=1}^n$}\tfrac{-162 k (2+k)}{5+k}\big)^2 \text{\scriptsize$\displaystyle\Big(\prod_{k=1}^n$}\tfrac{-729 \iota(5+k)}{k (2+k)^3}\Big),\\
\ev(x_4,n)&=\text{\scriptsize$\displaystyle\prod_{k=1}^n$} \tfrac{-162 k (2+k)}{5+k}.
\end{align*}
Now we activate Theorem~\ref{Thm:LambdaTauConstruction}. Repeating the construction from Example~\ref{Exp:SpecialCaseTauExp}
one obtains the $R\Pi$-extension $\dfield{\EE[z]}{\sigma}$ of $\dfield{\FF}{\sigma}$ with $\EE=\FF\langle x_3\rangle\langle x_4\rangle$ and $\sigma(x_3)=\alpha_3\,x_3$, $\sigma(x_4)=\alpha_4\,x_4$, $\sigma(z)=-z$ where the $\alpha_3,\alpha_4$ are given in~\eqref{Equ:Exp:DefineSpecialAlpha}.
Moreover, we can define the
evaluation function $\fct{\ev'}{\EE[z]\times\NN}{\KK}$ with $\ev'(f,n)=\ev(f,n)$ for all $f\in\FF$, $\ev'(x_3,n)=\ev(x_3,n)$, $\ev'(x_4,n)=\ev(x_4,n)$ and $\ev'(z,n)=(-1)^n$ yielding the $\KK$-embedding $\fct{\tau'}{\EE[z]}{\seqK}$. By coincidence these evaluations equal the evaluations given in~\eqref{Equ:Exp:DefineSpecialEv}; this comes from the fact the lower bounds are all the same. Thus we can use the construction from Example~\ref{Exp:lambdaDefSimple} 
and obtain $\lambda$ with~\eqref{Equ:LambdaDefSimple} such that the diagram~\eqref{Equ:CommutingDiagramWithRExt} commutes. Finally, we take $\fct{\hat{\lambda}}{\FF\langle{x}_1\rangle\langle{x}_2\rangle\langle{x}_3\rangle\langle{x}_4\rangle}{\FF\langle x_3\rangle\langle x_4\rangle[z]}$
 with $\hat{\lambda}=\lambda\circ\mu$. In other words, 
$\hat{\lambda}$ is determined by $\lambda|_{\FF}=\id$ and~\eqref{Equ:LambdaMainExp}.
It follows that $\hat{\tau}=\tau'\circ\hat{\lambda}$. Summarizing, we have carried out the construction visualized in~\eqref{Equ:FullDiagramm} with $r=4$ and $\HH=\EE[z]$.
In particular, we have~\eqref{equ:tauKerSimple}. Thus applying the inverse of $\mu$ defined in~\eqref{Equ:muInvExp} to the entries given in~\eqref{equ:tauKerSimple} we obtain~\eqref{Equ:kertauExp} (compare also~\eqref{Equ:IdealSetZ} with~\eqref{Equ:riDef}).
\end{example}

Suppose that we are given a $P$-extension $\dfield{\FF\langle \hat{x}_1\rangle\dots\langle \hat{x}_r\rangle}{\sigma}$ of a difference field $\dfield{\FF}{\sigma}$ with $\KK=\const{\FF}{\sigma}$, $\sigma(\hat{x}_i)=\hat{\alpha}_i\hat{x}_i$ for $1\leq i\leq r$ and 
an evaluation function $\fct{\hat{\ev}}{\FF\langle \hat{x}_1\rangle\dots\langle \hat{x}_r\rangle\times\NN}{\KK}$ 
of the form 
\begin{equation}\label{Equ:EvHatxi}
\hat{\ev}(\hat{x}_i,n)=F_i(n)=\prod_{k=l_i}^n\ev(\hat{\alpha}_i,k-1).
\end{equation}
This yields the $\KK$-homomorphism $\fct{\tau}{\FF\langle \hat{x}_1\rangle\dots\langle \hat{x}_r\rangle}{\seqK}$ defined by $\hat{\tau}(f)=(\hat{\ev}(f,n))_{n\geq0}$. In addition, suppose that the (algorithmic) properties in Assumption~\ref{Assum:AlgProp} hold.
By Theorem~\ref{Thm:MainResult} any expression in terms of products modeled in $\dfield{\FF\langle \hat{x}_1\rangle\dots\langle \hat{x}_r\rangle}{\sigma}$ together with  $\hat{\ev}$ can be also modeled in an $R\Pi$-extension $\dfield{\HH}{\sigma}$ of $\dfield{\FF}{\sigma}$ together with an evaluation function $\fct{\ev'}{\HH\times\NN}{\KK}$. 
This yields the following consequences.\\
(1) Using the symbolic summation toolbox described in~\cite{Schneider:07d,Schneider:08c,Schneider:10a,Schneider:10b,Schneider:15} any expression of indefinite nested sums defined over these products can be modeled in a \sigmaE-extension defined over $R\Pi$-extensions. In particular, the efficient simplification machinery of the summation package \texttt{Sigma} for indefinite nested sums can be applied in this general setting.\\
(2) Restricting to single nested products given in~\eqref{Equ:InputProds}, we obtain an improved calculation formula for the  evaluation function $\hat{\ev}$.
Namely, using the composition $\hat{\tau}=\tau'\circ\hat{\lambda}$ together with the definition of $\hat{\lambda}$ given in~\eqref{Equ:LambdaGenDef} and the definition of the evaluation $\ev'$ of $\tau'$ given in~\eqref{ExtractBetterEv} one gets:
\begin{equation}\label{Equ:BetterRepres}
\hat{\ev}(\hat{x}_i,n)=(\rho^n)^{o_i}\ev(\gamma_i,n)\ev'(x_{u+1},n)^{m_{i,u+1}}\dots\ev'(x_{r},n)^{m_{i,r}}
\end{equation}
for some $l_i\in\NN$ and $\kappa_i\in\KK^*$
for $1\leq i\leq r$. 
By setting $\phi_{i-u}(k):=\ev(\alpha_{i-u},k-1)$ equation~\eqref{ExtractBetterEv}
turns to 
$$\ev'(x_i)=\kappa_i\,\Phi_{i-u}(n)=\kappa_i\,\prod_{k=l_{i-u}}^n\phi_{i-u}(k),\quad\quad u+1\leq i\leq r.$$
In short, the products given in~\eqref{Equ:InputProds} can be simplified to expressions in terms of the products given in~\eqref{Equ:OutputProds}.
In particular, the sequences $\tau'(x_i)=(\ev'(x_i,n))_{n\geq0}=\kappa_i(\Phi_{i-u}(n))_{n\geq0}$ are algebraically independent over the sequences $\tau(\FF[z])$ while the sequences of the products in~\eqref{Equ:EvHatxi} are usually algebraically dependent (except for the special case $r=s$).

\begin{example}[Cont.~Ex.~\ref{Exp:MainExampleResult}]\label{Exp:MainExampleResult2}
Consider the difference ring $\dfield{\FF\langle\hat{x}_1\rangle\langle\hat{x}_2\rangle\langle\hat{x}_3\rangle\langle\hat{x}_4\rangle}{\sigma}$ with the evaluation function $\hat{\ev}$ defined by $\hat{\ev}(\hat{x}_i,n)=F_i(n)$ for $i=1,2,3,4$ with~\eqref{Equ:DefineFi}. The construction from
Example~\ref{Exp:MainExampleResult} yields an improved way to define $\ev(\hat{x}_i,n)$. Namely, we constructed the $R\Pi$-extension $\dfield{\EE[z]}{\sigma}$ of $\dfield{\FF}{\sigma}$ with
$\EE=\FF\langle x_3\rangle\langle x_4\rangle$ and $\sigma(x_3)=\alpha_3\,x_3$, $\sigma(x_4)=\alpha_4\,x_4$ and $\sigma(z)=-z$ where the $\alpha_3,\alpha_4$ are given in~\eqref{Equ:Exp:DefineSpecialAlpha}.
Moreover, we obtained the induced evaluation function $\fct{\ev'}{\EE[z]\times\NN}{\KK}$ defined by~\eqref{Eq:MainExpamleEvP}.
Since $\hat{\tau}=\tau'\circ\hat{\lambda}$ by construction, it follows with~\eqref{Equ:LambdaMainExp} that the evaluation $\hat{\ev}$ can be given also in the following form~\eqref{Equ:AlternativeEv}
which is precisely~\eqref{Equ:FiEvalNew}. 
\end{example}

\section{Conclusion}\label{Sec:Conclusion}

Given a finite set of products~\eqref{Equ:InputProds} whose multiplicands can be modeled in a difference field $\dfield{\FF}{\sigma}$, 
we presented a general framework in Theorem~\ref{Thm:MainResult} to find a minimal $R\Pi$-extension defined over $\dfield{\FF}{\sigma}$ in which the products can be modeled. In particular, the class of mixed-multibasic hypergeometric products are covered in this machinery,
As a consequence, the input products can be rephrased
by alternative products~\eqref{Equ:OutputProds} which are algebraically independent among each other and by one product of the form $\gamma^n$ with a root of unity $\gamma$. In particular, the number $s$ of output products and the order of the root of unity $\gamma$ are minimal among the possible choices of product representations. Moreover, we are able to compute a finite set of generators that produce all relations among the input products.

We remark that the analogous result for indefinite nested sums has been elaborated in~\cite[Thm.~3.13]{DR3}. A natural task is to merge the product and sum representations accordingly to find the difference ideal of all relations of indefinite nested sums defined over mixed-multibasic hypergeometric products.

The underlying algorithms for Theorem~\ref{Thm:MainResult} require that the ground difference field $\dfield{\FF}{\sigma}$ satisfies certain (algorithmic) properties enumerated in Assumption~\ref{Assum:AlgProp}. An interesting question is whether these requirements can be relaxed to weaker properties in order to calculate such representations or to find all relations among the given input products.

Furthermore, we showed explicitly that this machinery can be applied to the rational difference field (see Examples~\ref{Exp:RatDFPlusEv} and~\ref{Exp:RatDFPlusEv2})
and more generally to the mixed-rational difference field (see Example~\ref{Exp:MixedDF}). An novel task will be the application of this machinery to more general classes of difference fields that satisfy these requirements.


\begin{thebibliography}{10}

\bibitem{CALadder:16}
J.~Ablinger, A.~Behring, J.~Bl\"umlein, A.~D. Freitas, A.~von Manteuffel, and
  C.~Schneider.
\newblock {Calculating Three Loop Ladder and V-Topologies for Massive Operator
  Matrix Elements by Computer Algebra}.
\newblock {\em Comput. Phys. Comm.}, 202:33--112, 2016.
\newblock arXiv:1509.08324 [hep-ph].

\bibitem{HugeSummation:18}
J.~Ablinger, J.~Bl\"umlein, A.~D. Freitas, A.~Goedicke, C.~Schneider, and
  K.~Sch\"onwald.
\newblock {The Two-mass Contribution to the Three-Loop Gluonic Operator Matrix
  Element $A_{gg,Q}^{(3)}$}.
\newblock {\em Nucl. Phys. B}, 932:129--240, 2018.
\newblock arXiv:1804.02226 [hep-ph].

\bibitem{Abramov:71}
S.~A. Abramov.
\newblock On the summation of rational functions.
\newblock {\em Zh. vychisl. mat. Fiz.}, 11:1071--1074, 1971.

\bibitem{Abramov:75}
S.~A. Abramov.
\newblock The rational component of the solution of a first-order linear
  recurrence relation with a rational right-hand side.
\newblock {\em U.S.S.R. Comput. Maths. Math. Phys.}, 15:216--221, 1975.
\newblock Transl. from Zh. vychisl. mat. mat. fiz. 15, pp. 1035--1039, 1975.

\bibitem{Abramov:89a}
S.~A. Abramov.
\newblock Rational solutions of linear differential and difference equations
  with polynomial coefficients.
\newblock {\em U.S.S.R. Comput. Math. Math. Phys.}, 29(6):7--12, 1989.

\bibitem{Petkov:10}
S.~A. Abramov and M.~Petkov{\v{s}}ek.
\newblock Polynomial ring automorphisms, rational {$(w,\sigma)$}-canonical
  forms, and the assignment problem.
\newblock {\em J. Symbolic Comput.}, 45(6):684--708, 2010.

\bibitem{AZ:06}
M.~Apagodu and D.~Zeilberger.
\newblock {Multi-variable Zeilberger and Almkvist-Zeilberger algorithms and the
  sharpening of Wilf-Zeilberger theory}.
\newblock {\em Adv. Appl. Math.}, 37:139--152, 2006.

\bibitem{Bauer:99}
A.~Bauer and M.~Petkov{\v{s}}ek.
\newblock Multibasic and mixed hypergeometric {Gosper}-type algorithms.
\newblock {\em J.~Symbolic Comput.}, 28(4--5):711--736, 1999.

\bibitem{Bron:00}
M.~Bronstein.
\newblock On solutions of linear ordinary difference equations in their
  coefficient field.
\newblock {\em J.~Symbolic Comput.}, 29(6):841--877, 2000.

\bibitem{CK:12}
S.~Chen and M.~Kauers.
\newblock {Order-Degree Curves for Hypergeometric Creative Telescoping}.
\newblock In J.~van~der Hoeven and M.~van Hoeij, editors, {\em {Proceedings of
  ISSAC 2012}}, pages 122--129, 2012.

\bibitem{Chyzak:00}
F.~Chyzak.
\newblock An extension of {Z}eilberger's fast algorithm to general holonomic
  functions.
\newblock {\em Discrete Math.}, 217:115--134, 2000.

\bibitem{Cohn:89}
P.~M. Cohn.
\newblock {\em Algebra}, volume~2.
\newblock John Wiley \& Sons, 2nd edition, 1989.

\bibitem{Cohn:65}
R.~M. Cohn.
\newblock {\em Difference Algebra}.
\newblock John Wiley \& Sons, 1965.

\bibitem{SNF3}
M.~Elsheikh, M.~Giesbrecht, A.~Novocin, and B.~D. Saunders.
\newblock Fast computation of {S}mith forms of sparse matrices over local
  rings.
\newblock In {\em Proc. I{SSAC} 2012}, pages 146--153. ACM, New York, 2012.

\bibitem{Ge:93}
G.~Ge.
\newblock {\em Algorithms {R}elated to the {M}ultiplicative {R}epresentation of
  {A}lgebraic {N}umbers}.
\newblock PhD thesis, Univeristy of California at Berkeley, 1993.

\bibitem{Gosper:78}
R.~W. Gosper.
\newblock Decision procedures for indefinite hypergeometric summation.
\newblock {\em Proc. Nat. Acad. Sci. U.S.A.}, 75:40--42, 1978.

\bibitem{Singer:08}
C.~Hardouin and M.~Singer.
\newblock Differential {G}alois theory of linear difference equations.
\newblock {\em Math. Ann.}, 342(2):333--377, 2008.

\bibitem{Karr:81}
M.~Karr.
\newblock Summation in finite terms.
\newblock {\em J.~ACM}, 28:305--350, 1981.

\bibitem{Karr:85}
M.~Karr.
\newblock Theory of summation in finite terms.
\newblock {\em J.~Symbolic Comput.}, 1:303--315, 1985.

\bibitem{Kauers:08}
M.~Kauers and B.~Zimmermann.
\newblock Computing the algebraic relations of c-finite sequences and
  multisequences.
\newblock {\em Journal of Symbolic Computation}, 43(11):787--803, 2008.

\bibitem{Koutschan:13}
C.~Koutschan.
\newblock Creative telescoping for holonomic functions.
\newblock In C.~Schneider and J.~Bl\"umlein, editors, {\em {Computer Algebra in
  Quantum Field Theory: Integration, Summation and Special Functions}}, Texts
  and Monographs in Symbolic Computation, pages 171--194. Springer, 2013.
\newblock arXiv:1307.4554 [cs.SC].

\bibitem{PauleNemes:97}
I.~Nemes and P.~Paule.
\newblock A canonical form guide to symbolic summation.
\newblock In A.~Miola and M.~Temperini, editors, {\em Advances in the Design of
  Symbolic Computation Systems}, Texts Monogr. Symbol. Comput., pages 84--110.
  Springer, Wien-New York, 1997.

\bibitem{OS:18}
E.~D. Ocansey.
\newblock Representing ($q$-)hypergeometric products and mixed versions in
  difference rings.
\newblock pages 175--213.

\bibitem{Ocansey:19}
E.~D. Ocansey.
\newblock {\em Difference ring algorithms for nested products}.
\newblock PhD thesis, RISC, J.~Kepler University, 2019.

\bibitem{Paule:95}
P.~Paule.
\newblock Greatest factorial factorization and symbolic summation.
\newblock {\em J.~Symbolic Comput.}, 20(3):235--268, 1995.

\bibitem{PauleRiese:97}
P.~Paule and A.~Riese.
\newblock A {M}athematica q-analogue of {Z}eil\-ber\-ger's algorithm based on
  an algebraically motivated approach to $q$-hypergeometric telescoping.
\newblock In M.~Ismail and M.~Rahman, editors, {\em Special Functions, q-Series
  and Related Topics}, volume~14, pages 179--210. AMS, 1997.

\bibitem{PauleSchorn:95}
P.~Paule and M.~Schorn.
\newblock A {M}athematica version of {Z}eilberger's algorithm for proving
  binomial coefficient identities.
\newblock {\em J.~Symbolic Comput.}, 20(5-6):673--698, 1995.

\bibitem{AequalB}
M.~Petkov{\v s}ek, H.~S. Wilf, and D.~Zeilberger.
\newblock {\em $A=B$}.
\newblock A. K. Peters, Wellesley, MA, 1996.

\bibitem{ZimingLi:11}
G.~F. S.~Chen, R.~Feng and Z.~Li.
\newblock On the structure of compatible rational functions.
\newblock In {\em {Proceedings of ISSAC 2011}}, pages 91--98, 2011.

\bibitem{SNF1}
B.~D. Saunders and Z.~Wan.
\newblock Smith normal form of dense integer matrices fastalgorithms into
  practice.
\newblock In J.~Gutierrez, editor, {\em Proc. ISSAC'04}, pages 274--281. ACM
  Press, 2004.

\bibitem{Schneider:01}
C.~Schneider.
\newblock Symbolic summation in difference fields.
\newblock Technical Report 01-17, RISC-Linz, J.~Kepler University, November
  2001.
\newblock PhD Thesis.

\bibitem{Schneider:05c}
C.~Schneider.
\newblock Product representations in ${\Pi}{\Sigma}$-fields.
\newblock {\em Ann. Comb.}, 9(1):75--99, 2005.

\bibitem{Schneider:07d}
C.~Schneider.
\newblock Simplifying sums in {$\Pi\Sigma$}-extensions.
\newblock {\em J. Algebra Appl.}, 6(3):415--441, 2007.

\bibitem{Schneider:07a}
C.~Schneider.
\newblock Symbolic summation assists combinatorics.
\newblock {\em S\'em.~Lothar. Combin.}, 56:1--36, 2007.
\newblock Article B56b.

\bibitem{Schneider:08c}
C.~Schneider.
\newblock A refined difference field theory for symbolic summation.
\newblock {\em J. Symbolic Comput.}, 43(9):611--644, 2008.
\newblock [arXiv:0808.2543v1].

\bibitem{Schneider:10c}
C.~Schneider.
\newblock Parameterized telescoping proves algebraic independence of sums.
\newblock {\em Ann. Comb.}, 14:533--552, 2010.
\newblock [arXiv:0808.2596]; for a preliminary version see FPSAC 2007.

\bibitem{Schneider:10a}
C.~Schneider.
\newblock Structural theorems for symbolic summation.
\newblock {\em Appl. Algebra Engrg. Comm. Comput.}, 21(1):1--32, 2010.

\bibitem{Schneider:10b}
C.~Schneider.
\newblock A symbolic summation approach to find optimal nested sum
  representations.
\newblock In A.~Carey, D.~Ellwood, S.~Paycha, and S.~Rosenberg, editors, {\em
  {Motives, Quantum Field Theory, and Pseudodifferential Operators}}, volume~12
  of {\em Clay Mathematics Proceedings}, pages 285--308. Amer. Math. Soc, 2010.
\newblock arXiv:0808.2543.

\bibitem{DR2}
C.~Schneider.
\newblock A streamlined difference ring theory: Indefinite nested sums, the
  alternating sign and the parameterized telescoping problem.
\newblock In F.~Winkler, V.~Negru, T.~Ida, T.~Jebelean, D.~Petcu, S.~Watt, and
  D.~Zaharie, editors, {\em Symbolic and Numeric Algorithms for Scientific
  Computing (SYNASC), 2014 15th International Symposium}, pages 26--33. IEEE
  Computer Society, 2014.
\newblock arXiv:1412.2782v1 [cs.SC].

\bibitem{Schneider:15}
C.~Schneider.
\newblock Fast algorithms for refined parameterized telescoping in difference
  fields.
\newblock In M.~W. J.~Guitierrez, J.~Schicho, editor, {\em Computer Algebra and
  Polynomials}, number 8942 in Lecture Notes in Computer Science (LNCS), pages
  157--191. Springer, 2015.
\newblock arXiv:1307.7887 [cs.SC].

\bibitem{DR1}
C.~Schneider.
\newblock A difference ring theory for symbolic summation.
\newblock {\em J. Symb. Comput.}, 72:82--127, 2016.
\newblock arXiv:1408.2776 [cs.SC].

\bibitem{DR3}
C.~Schneider.
\newblock {Summation Theory II: Characterizations of $R\Pi\Sigma$-extensions
  and algorithmic aspects}.
\newblock {\em J. Symb. Comput.}, 80(3):616--664, 2017.
\newblock arXiv:1603.04285 [cs.SC].

\bibitem{Singer:16}
M.~Singer.
\newblock Algebraic and algorithmic aspects of linear difference equations.
\newblock In M.~S. C.~Hardouin, J.~Sauloy, editor, {\em Galois theories of
  linear difference equations: An Introduction}, volume 211 of {\em
  Mathematical Surveys and Monographs}. AMS, 2016.

\bibitem{Singer:97}
M.~van~der Put and M.~Singer.
\newblock {\em Galois theory of difference equations}, volume 1666 of {\em
  Lecture Notes in Mathematics}.
\newblock Springer-Verlag, Berlin, 1997.

\bibitem{SNF2}
P.~G. A.~S. W.~Eberly, M.~Giesbrecht and G.~Villard.
\newblock Faster inver-sion and other black box matrix computations using
  efficient block projections.
\newblock In C.~Brown, editor, {\em Proc. ISSAC'07}, pages 143--150. ACM Press,
  2007.

\bibitem{Wegschaider}
K.~Wegschaider.
\newblock Computer generated proofs of binomial multi-sum identities.
\newblock Master's thesis, RISC, J. Kepler University, May 1997.

\bibitem{Wilf:92}
H.~Wilf and D.~Zeilberger.
\newblock An algorithmic proof theory for hypergeometric (ordinary and ``q'')
  multisum/integral identities.
\newblock {\em Invent. Math.}, 108:575--633, 1992.

\bibitem{Zeilberger:90a}
D.~Zeilberger.
\newblock A holonomic systems approach to special functions identities.
\newblock {\em J.~Comput. Appl. Math.}, 32:321--368, 1990.

\bibitem{Zeilberger:91}
D.~Zeilberger.
\newblock The method of creative telescoping.
\newblock {\em J.~Symbolic Comput.}, 11:195--204, 1991.

\end{thebibliography}

\end{document}